\documentclass[aps,article,twocolumn,superscriptaddress,nofootinbib]{revtex4-2}

\usepackage{graphicx}
\usepackage{mathtools}
\usepackage{mathrsfs}
\usepackage{mathbbol}
\usepackage{amsmath,amsfonts}
\usepackage{accents}
\usepackage{dsfont}
\usepackage{cancel}
\usepackage{amssymb,amsthm}
\usepackage[dvipsnames]{xcolor}
\usepackage{tikz}
\usepackage{titlesec}
\usepackage{marginnote}
\usepackage{ifthen}
\usepackage{etoolbox} 
\usepackage{physics}
\usepackage{relsize}
\usepackage{zref-savepos}
\usepackage{caption}
\usepackage{hyperref}
\usepackage{multirow}
\usepackage{soul}

\usepackage{wrapfig}

\newtheorem{prop}{Proposition}
\newtheorem{corol}{Corollary}

\NewDocumentCommand{\newauthornote}{mmm}{
      \NewDocumentCommand{#1}{s m}{%
            \textcolor{#3}{\IfBooleanTF{##1}%
                  {\emph{\footnotesize \textbf{#2}:~##2}} 
                  {{##2}}
      }}
}

\makeatletter
\renewcommand\onecolumngrid{%
  \do@columngrid{one}{\@ne}%
  \def\set@footnotewidth{\onecolumngrid}
  \def\footnoterule{\kern-6pt\hrule width 1.5in\kern6pt}
}
\renewcommand\twocolumngrid{%
  \def\footnoterule{
    \dimen@\skip\footins
    \divide\dimen@\thr@@
    \kern-\dimen@\hrule width.5in\kern\dimen@
  }
  \do@columngrid{mlt}{\tw@}%
}
\makeatother

\renewcommand{\Tilde}{\widetilde}

\makeatletter
\let\save@mathaccent\mathaccent
\newcommand*\if@single[3]{%
   \setbox0\hbox{${\mathaccent"0362{#1}}^H$}%
   \setbox2\hbox{${\mathaccent"0362{\kern0pt#1}}^H$}%
   \ifdim\ht0=\ht2 #3\else #2\fi
   }
\newcommand*\rel@kern[1]{\kern#1\dimexpr\macc@kerna}
\newcommand*\widebar[1]{\@ifnextchar^{{\wide@bar{#1}{0}}}{\wide@bar{#1}{1}}}
\newcommand*\wide@bar[2]{\if@single{#1}{\wide@bar@{#1}{#2}{1}}{\wide@bar@{#1}{#2}{2}}}
\newcommand*\wide@bar@[3]{%
   \begingroup
   \def\mathaccent##1##2{%
      \let\mathaccent\save@mathaccent
      \if#32 \let\macc@nucleus\first@char \fi
      \setbox\z@\hbox{$\macc@style{\macc@nucleus}_{}$}%
      \setbox\tw@\hbox{$\macc@style{\macc@nucleus}{}_{}$}%
      \dimen@\wd\tw@
      \advance\dimen@-\wd\z@
      \divide\dimen@ 3
      \@tempdima\wd\tw@
      \advance\@tempdima-\scriptspace
      \divide\@tempdima 10
      \advance\dimen@-\@tempdima
      \ifdim\dimen@>\z@ \dimen@0pt\fi
      \rel@kern{0.6}\kern-\dimen@
      \if#31
         \overline{\rel@kern{-0.6}\kern\dimen@\macc@nucleus\rel@kern{0.4}\kern\dimen@}%
         \advance\dimen@0.4\dimexpr\macc@kerna
         \let\final@kern#2%
         \ifdim\dimen@<\z@ \let\final@kern1\fi
         \if\final@kern1 \kern-\dimen@\fi
      \else
         \overline{\rel@kern{-0.6}\kern\dimen@#1}%
      \fi
   }%
   \macc@depth\@ne
   \let\math@bgroup\@empty \let\math@egroup\macc@set@skewchar
   \mathsurround\z@ \frozen@everymath{\mathgroup\macc@group\relax}%
   \macc@set@skewchar\relax
   \let\mathaccentV\macc@nested@a
   \if#31
      \macc@nested@a\relax111{#1}%
   \else
      \def\gobble@till@marker##1\endmarker{}%
      \futurelet\first@char\gobble@till@marker#1\endmarker
      \ifcat\noexpand\first@char A\else
         \def\first@char{}%
      \fi
      \macc@nested@a\relax111{\first@char}%
   \fi
   \endgroup
}

\renewcommand{\Bar}{\widebar}

\def\cA{{\mathcal{A}}}
\def\cB{{\mathcal{B}}}
\def\cE{{\mathcal{E}}}

\def\cN{{\mathcal{N}}}

\def\cO{{\mathcal{O}}}

\def\cV{{\mathcal{V}}}

\def\q{\qquad}
\def\be#1\ee{\begin{align}#1\end{align}}
\def\bg#1\eg{\begin{gather}#1\end{gather}}
\def\bsub#1\esub{\begin{subequations}#1\end{subequations}}
\def\p{\partial}
\def\wth{\text{with}}

\NewDocumentCommand{\PP}{s}{%
   \IfBooleanTF{#1}
      {\hat{\mathlarger{\Pi}}}    
      {\mathlarger{\Pi}}             
}

\NewDocumentCommand{\Pinv}{s}{%
   \IfBooleanTF{#1}
      {\hat{\mathlarger{\Pi}}_{\text{inv}}}    
      {\mathlarger{\Pi}_{\text{inv}}}      
}

\NewDocumentCommand{\Pphys}{s}{%
   \IfBooleanTF{#1}
      {{\PP*}_{\text{phys}}}    
      {{\PP}_{\text{phys}}}      
}

\makeatletter
\g@addto@macro\appendix{%
  \renewcommand\section[1]{%
    \refstepcounter{section}%
    \vspace{1em}
    \centerline{\bfseries Appendix \Alph{section}:~#1}%
    \vspace{0.5em}
  }%
}
\makeatother

\let\oldsection\section
\renewcommand{\section}[1]{%
    \vspace{-0.8em} 
    \oldsection{#1} 
    \vspace{-1em}   
}

\def\H{\mathcal{H}}
\def\Hk{\mathcal{H}_\text{kin}}
\def\Hkin{\Hk}
\def\Hphys{\mathcal{H}_\text{phys}}

\def\Id{\mathds{1}}

\def\phys{{\text{phys}}}

\let\ll\l

\def\l{\left(}
\def\r{\right)}

\newauthornote{\francesco}{F}{Mahogany}

\begin{document}

\title{Relational entanglement entropies and quantum reference frames in gauge theories}

\author{Gonçalo Ara\'ujo-Regado}
\email{goncalo.araujo@oist.jp}
\author{Philipp A.\ H\"ohn}
\email{philipp.hoehn@oist.jp}
\author{Francesco Sartini}
\email{francesco.sartini@oist.jp}
\affiliation{Qubits and Spacetime Unit, Okinawa Institute of Science and Technology Graduate University, Onna, Okinawa 904-0495, Japan}

\begin{abstract}
It has been shown that defining gravitational entanglement entropies relative to quantum reference frames (QRFs) intrinsically regularizes them. Here, we demonstrate that such relational definitions also have an advantage in lattice gauge theories, where no ultraviolet divergences occur. To this end, we introduce QRFs for the gauge group via Wilson lines on a lattice with global boundary, realizing edge modes on the bulk entangling surface. Overcoming challenges of previous nonrelational approaches, we show that defining gauge-invariant subsystems associated with subregions \emph{relative} to such QRFs naturally leads to a factorization across the surface, yielding distillable relational entanglement entropies. Distinguishing between extrinsic and intrinsic QRFs, according to whether they are built from the region or its complement, leads to extrinsic and intrinsic relational algebras ascribed to the region. The ``electric center algebra'' of previous approaches is recovered as the algebra that all extrinsic QRFs agree on, or by incoherently twirling any extrinsic algebra over the electric corner symmetry group. Similarly, a generalization of previous proposals for a ``magnetic center algebra'' is obtained as the algebra that all intrinsic QRFs agree on, or, in the Abelian case, by incoherently twirling any intrinsic algebra over a dual magnetic corner group. Altogether, this leads to a compelling regional algebra and relative entropy hierarchy. Invoking the corner twirls, we find that the extrinsic/intrinsic relational entanglement entropies are upper bounded by the non-distillable electric/magnetic center entropies. Finally, using extrinsic QRFs, we discuss the influence of ``asymptotic'' symmetries on regional entropies. Our work thus unifies and extends previous approaches and reveals the interplay between entropies and regional symmetry structures.
\end{abstract}

\maketitle


\section{Introduction}
\label{sec:intro}

Entanglement entropy has become an omnipresent tool in modern physics to probe the structure of quantum correlations, with manifold applications and deep implications in quantum field theory \cite{Witten:2018zxz,Casini:2022rlv,Srednicki:1993im}, gravity \cite{Bombelli:1986rw,Jacobson:2015hqa,Chandrasekaran:2022cip,Jensen:2023yxy,Kudler-Flam:2023qfl,Kirklin:2024gyl}, and holography \cite{Ryu:2006bv,Hubeny:2007xt,VanRaamsdonk:2010pw,Nishioka:2009un,VanRaamsdonk:2016exw}. However, defining it requires a suitable subsystem partition and specifying one which reflects spacetime locality is \emph{a priori} a challenge in theories with gauge symmetry, where simple Hilbert space factorizations across the entangling surface are obstructed by Gau\ss's law. 

For example, insisting in lattice gauge theories on a spacetime background notion of locality leads to overlapping regional observable algebras (i.e.\ algebras with center) and non-distillable entropy contributions \cite{Buividovich:2008gq,Donnelly:2011hn,Donnelly:2014gva,Casini:2013rba,Ghosh:2015iwa,Soni:2015yga,VanAcoleyen:2015ccp,Delcamp:2016eya,Aoki:2015bsa,Hung:2015fla,Bianchi:2024aim,Dong:2023kyr}, entailing a mismatch with operational notions of entanglement \cite{Soni:2015yga,VanAcoleyen:2015ccp}. To sidestep the locality issues, many of these constructions rely on embedding the global physical Hilbert space into an \emph{auxiliary} extended Hilbert space to exploit its tensor product structure (TPS) for entropy computations \cite{Buividovich:2008gq,Donnelly:2011hn,Donnelly:2014gva,Ghosh:2015iwa,VanAcoleyen:2015ccp,Delcamp:2016eya,Aoki:2015bsa}. 

A key development came with the recognition that these questions are intimately intertwined with the appearance of \emph{edge modes} and physical corner symmetries on the entangling surface when defining local subsystems in gauge theory and gravity \cite{Donnelly:2016auv,Donnelly:2014fua,Donnelly:2015hxa,Geiller:2019bti,Blommaert:2018oue,Ball:2024hqe,Frenkel:2023yuw,Fliss:2025kzi,Akers:2024wab,Law:2025ktz,Ball:2024hqe,Balasubramanian:2023dpj,Ball:2024gti,Carrozza:2021gju,Araujo-Regado:2024dpr,Gomes:2018shn,Gomes:2019xto,Riello:2021lfl,Freidel:2023bnj,Ciambelli:2022vot,Ciambelli:2021vnn, Freidel:2020xyx,Freidel:2020svx,Freidel:2023bnj,Donnelly:2020xgu,Donnelly:2022kfs,Speranza:2017gxd,Blommaert:2018rsf,Klinger:2023tgi,Wieland:2017cmf,Pulakkat:2025eid,Carrozza:2022xut,Kabel:2023jve,Giesel:2024xtb}. Naturally, these edge modes contribute to regional \cite{Donnelly:2016auv,Donnelly:2014fua,Donnelly:2015hxa,Geiller:2019bti,Blommaert:2018oue,Ball:2024hqe,Frenkel:2023yuw,Fliss:2025kzi,Akers:2024wab,Law:2025ktz,Balasubramanian:2023dpj}, as well as asymptotic entanglement entropies \cite{Chen:2023tvj,Chen:2024kuq}. However, the interpretation of these edge modes remains debated, and their precise role in understanding entanglement structures is arguably not settled.

Here, we revisit these questions with a novel \emph{relational} approach to gauge-invariant subsystem locality, TPSs, and ultimately entropies \cite{Hoehn:2023ehz,AliAhmad:2021adn,DeVuyst:2024pop,DeVuyst:2024uvd}. To link with previous entropy constructions, we explore this approach in lattice gauge theories subject to any finite or compact Lie structure group. Our approach is based on quantum reference frames (QRFs) \cite{delaHamette:2021oex,Hoehn:2023ehz,AliAhmad:2021adn,Hoehn:2019fsy,Hoehn:2020epv,Carrozza:2024smc,DeVuyst:2024pop,DeVuyst:2024uvd,Vanrietvelde:2018dit,Vanrietvelde:2018pgb,Giacomini:2021gei,Hoehn:2021flk,delaHamette:2021piz,Castro-Ruiz:2019nnl,Suleymanov:2025nrr,Hoehn:2023axh,Giacomini:2017zju,Bartlett:2006tzx,Castro-Ruiz:2021vnq,Krumm:2020fws,Loveridge:2017pcv,Carette:2023wpz,Fewster:2024pur,delaHamette:2021iwx,delaHamette:2020dyi,Giacomini:2018gxh,Ballesteros:2020lgl}, which constitute a universal toolkit for dealing with symmetries in quantum theory and thus naturally apply in this context. Since we face both gauge and physical symmetries, it will be necessary to invoke also two different kinds of QRF frameworks, one for either type of symmetry. 

First, QRFs for the \emph{gauge} group are built according to the \emph{perspective-neutral} framework \cite{delaHamette:2021oex,Hoehn:2023ehz,AliAhmad:2021adn,Hoehn:2019fsy,Hoehn:2020epv,Carrozza:2024smc,DeVuyst:2024pop,DeVuyst:2024uvd,Vanrietvelde:2018dit,Vanrietvelde:2018pgb,Hoehn:2021flk,delaHamette:2021piz,Giacomini:2021gei,Castro-Ruiz:2019nnl,Suleymanov:2025nrr,Hoehn:2023axh} and can be realized, for example, with Wilson lines. In particular, for Wilson lines ending on the entangling surface, this provides a realization of edge modes as explained for the continuum in \cite{Carrozza:2021gju,Araujo-Regado:2024dpr}. In analogy to the relativity of simultaneity in special relativity, this framework shows that partitioning a quantum system gauge-invariantly into subsystems is not absolute but depends on the choice of QRF \cite{Hoehn:2023ehz}. This subsystem relativity \cite{AliAhmad:2021adn,delaHamette:2021oex,Hoehn:2023ehz,DeVuyst:2024pop,DeVuyst:2024uvd} is the root of the QRF-dependence of physical properties \cite{Giacomini:2017zju,Giacomini:2018gxh,Hoehn:2023ehz,Cepollaro:2024rss}, such as entanglement entropies \cite{Hoehn:2023ehz,Cepollaro:2024rss,Suleymanov:2025nrr,DeVuyst:2024pop,DeVuyst:2024uvd}. 

In particular, as shown in \cite{DeVuyst:2024pop,DeVuyst:2024uvd} (see also \cite{Kirklin:2024gyl}), this same relational approach has implicitly been invoked in the recent discussion of gravitational von Neumann algebras \cite{Chandrasekaran:2022cip,Jensen:2023yxy,Kudler-Flam:2023qfl}, which change type upon the inclusion of an ``observer'' (QRF), leading to an intrinsic regularization of otherwise divergent gravitational entanglement entropies. Here, we will see that the same approach also has advantages on a lattice where ultraviolet divergences are absent.

Applying this to the lattice, we will define gauge-invariant subsystems associated with subregions \emph{relative} to a choice of QRF. Following \cite{Araujo-Regado:2024dpr},  it will become important to distinguish \emph{extrinsic} and \emph{intrinsic} edge mode QRFs, according to whether they are built from Wilson lines in the entangling surface or in the regional complement. This results in qualitatively distinct extrinsic and intrinsic algebras of observables describing the regional physics \emph{relative} to the chosen QRF, with the extrinsic algebra encoding certain cross-boundary relational observables. Along the way, we also explain how Wilson loops and dressed electric fields can be written as relational observables relative to a QRF. As in the classical continuum \cite{Carrozza:2021gju,Araujo-Regado:2024dpr}, the electric corner symmetries are realized by \emph{reorientations} of the extrinsic edge mode QRF, while the intrinsic one possesses no such symmetries. This leads to a \emph{gauge-invariant} boundary QRF for the physical electric corner symmetries, which we identify with the regional Goldstone mode (GM) \cite{Araujo-Regado:2024dpr}. Its orientations parametrize a regional degeneracy with respect to certain cross-boundary relations.

Crucially, it turns out that relative to extrinsic QRFs, the physical Hilbert space and observable algebras always \emph{factorize} into what it ``sees'' as the region and its complement. The same is true for intrinsic QRFs for Abelian theories, realizing in that case a quantum lattice version of a classical phase space factorization observed in \cite{Ball:2024hqe,Araujo-Regado:2024dpr}. In either case, this leads to a TPS on the physical Hilbert space itself without any auxiliary structures and thereby to standard distillable  entanglement entropies for the region in a relational manner, as previously observed in toy models in \cite{Hoehn:2023ehz}. For extrinsic QRFs, it turns out that this distillable entanglement is related (via the dressing/gauge fixing equivalence) to ``gauge-fixed'' entropies discussed in \cite{Casini:2013rba,VanAcoleyen:2015ccp}. Missing their relationally local content, these entropies were somewhat dismissed in these works for their background nonlocality.

Akin to how all inertial observers in special relativity agree on Lorentz scalars, we then recover regional electric and magnetic center algebras of previous nonrelational constructions \cite{Casini:2013rba,Soni:2015yga,Delcamp:2016eya} as the intersection of all possible extrinsic and intrinsic relational algebras, respectively. That is, these center algebras constitute the observables on which all corresponding QRFs agree. This results in a compelling hierarchy of regional observable algebras with an associated hierarchy of relative entropies. For continuous structure groups, this magnetic center definition generalizes previous proposals \cite{Casini:2013rba,Delcamp:2016eya} that were only formulated for finite groups. Interestingly, it turns out that exactly for this generalized case, magnetic center entropies are physically rather arbitrary due to the absence of a physically distinguished trace.

Second, QRFs for the \emph{physical} corner symmetry groups are treated according to either the \emph{quantum information} \cite{Bartlett:2006tzx,Castro-Ruiz:2021vnq,Krumm:2020fws} or the \emph{operational} \cite{Loveridge:2017pcv,Carette:2023wpz,Fewster:2024pur} 
frameworks, which essentially agree for finite and compact groups for the purposes relevant to our discussion. This applies in particular to the GM for the electric symmetries. Incoherently twirling an extrinsic relational algebra over the GM's orientations leads to electric boundary charge superselection and recovers the electric center algebra. Similarly, twirling an intrinsic relational algebra over the orientations of a (Pontryagin) dual QRF for a dual magnetic corner symmetry group in the Abelian case leads to a magnetic boundary charge superselection and recovers the magnetic center algebra. We exploit these observations to prove that the relational entanglement entropies of the extrinsic/intrinsic algebras are upper bounded by the non-distillable entropies of the electric/magnetic center algebras.

Finally, working on a lattice with global boundary, we will also briefly discuss the influence of large gauge transformations on the regional entanglement entropies. As in the classical continuum, this is made possible by the fact that extrinsic QRFs ``pull in'' global boundary symmetries to the entangling surface \cite{Araujo-Regado:2024dpr}.

Altogether, our work opens up a novel relational perspective on entanglement entropies in theories with gauge symmetry, encompasses and unifies previous approaches, and reveals rich algebraic and Hilbert space structures made possible by resorting to QRFs that allow for a deeper understanding of entanglement in such theories. Indeed, restricting to nonrelational regional observable algebras would be akin to restricting to Lorentz scalars only in special relativity and missing out on all the interesting frame-covariant part of the physics.

\section{Kinematics of Hamiltonian lattice gauge theory}
\label{sec:LGT}

We adopt the Hamiltonian formulation of lattice gauge theories \cite{kogut1975,Kogut_Stephanov_2003,Kogut:1982ds}, which also constitutes the arena of previous entanglement entropy discussions \cite{Buividovich:2008gq,Donnelly:2011hn,Donnelly:2014gva,Casini:2013rba,Ghosh:2015iwa,Soni:2015yga,VanAcoleyen:2015ccp,Delcamp:2016eya,Aoki:2015bsa,Hung:2015fla}. We briefly recall its main features here and refer to App.~\ref{app:kin} for further details. 

We consider pure gauge theory with arbitrary compact (incl.~finite) structure Lie group $G$ (e.g., $\rm{SU}(3)$ for QCD, $\rm{U}(1)$ for QED) on a spatial lattice $\mathcal{L}$ of trivial topology and with a global boundary $\mathcal{B}$. $\mathcal{L}$ consists of a set $\mathcal{V}$ of vertices and a set $\mathcal{E}$ of edges, cf.~Fig.~\ref{fig:Lattice_2d}. Each edge $e\in\mathcal{E}$ is equipped with an orientation and a Hilbert space $\H_e=L^2(G)$ comprising the regular representation of $G$ and carrying the Wilson line link variables $\hat{g}_e$, taking value in some (typically the fundamental) representation $\rho$ of $G$. The link variables thus constitute the configuration variables. 
$L^2(G)$ admits two commuting unitary actions of $G$, defined by left $U_e(g)$ and right $V_e(g)$ multiplication on the group basis: $U_e(g_1)V_e(g_2)\ket{g}_e=\ket{g_1gg_2^{-1}}_e$, $\forall\,g,g_1,g_2\in G$. Hence, for continuous $G$, there are two types of generators on $\H_e$, which we call left $L^a_e$ and right $R_e^a$ electric fields and which comprise the conjugate variables, obeying $[L^a_e,R^b_{e'}]=0$ and
\begin{align}
    [L_e^a,\hat{g}_{e'}]&=\delta_{ee'}T^a_\rho\hat{g}_e\,,\,\,\,\,\,\qquad [R^a_e,\hat{g}_{e'}]=\delta_{ee'}\hat{g}_eT^a_\rho\,,\\
    [L_e^a,L^b_{e'}]&=-i\delta_{ee'} f^{abc}L^c_e\,,\,\,\,\, [R^a_e,R^b_{e'}]=i\delta_{ee'}f^{abc}R^c_e\,,
\end{align}
where the $T^a_\rho$ are the generators of $G$ in representation $\rho$ and $f^{abc}$ are the structure constants of the Lie algebra $\mathfrak{g}$. Left and right electric fields are related by ${R}^a_e=(\hat{g}_e^\text{Adj})^{ab}{L}^b_e$ and thus not independent, where $\hat{g}_e^\text{Adj}$ is the link variable in adjoint representation.
The kinematical Hilbert space is then given by ${\H_{\rm kin}=\bigotimes_{e\in\mathcal{E}}\H_e}\simeq L^2(G^E)$, where $E=|\mathcal{E}|$.\footnote{Boundary conditions are imposed on the $\H_e$ with $e\in\mathcal{B}$.}\\

\hspace{-12pt}\begin{minipage}{.48\textwidth}
    \centering
    \vspace{5pt}
    \captionsetup{hypcap=false} 
    \includegraphics[width=0.46\linewidth]{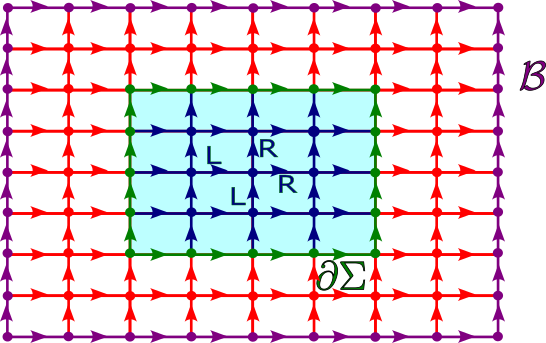}\hspace{16pt}
    \includegraphics[width=0.46\linewidth]{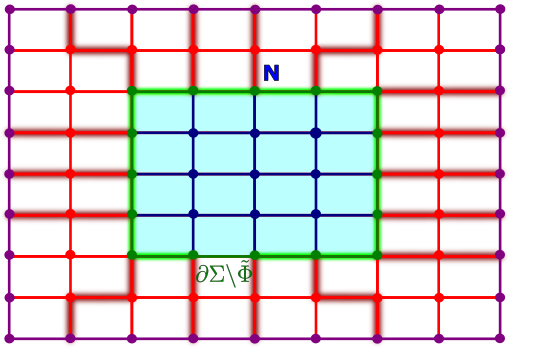}
    \captionof{figure}{\small \emph{Left:} 2d example of a regular oriented lattice with subregion $\Sigma$ (turquoise) and boundaries $\p\Sigma$ and $\mathcal{B}$. \emph{Right:} \emph{extrinsic} (thick red, $\mathcal{B}$-anchored) and \emph{intrinsic} (thick blue, rooted at $N$) Wilson line QRF  for nodes in $\p\Sigma$.}
\label{fig:Lattice_2d}
\end{minipage}\\

Small gauge transformations act independently at each of the bulk vertices $v\in\mathcal{V}_{\rm b}:=\mathcal{V}\setminus\left(\mathcal{V}\cap\mathcal{B}\right)$ and in our convention as 
\begin{equation} \label{eq:gaugetransfconv}U_v(h)=\bigotimes_{e\in\mathcal{E}_v^{\rm in}}U_e(h)\bigotimes_{e'\in\mathcal{E}_v^{\rm out}}V_e(h)\,, 
\end{equation}
where $\mathcal{E}_v^{\rm in/out}$ denote the sets of in- and outgoing edges at $v$. The (small) gauge group is thus $\mathcal{G}\simeq G^V$ with $V=|\mathcal{V}_{\rm b}|$, hence compact on a finite lattice, and acts on $\H_{\rm kin}$ as $U(\mathbf{g}):=\bigotimes_{v\in\mathcal{V}_{\rm b}} U_v(g_v)$, $\mathbf{g}\in\mathcal{G}$. Bulk link variables therefore transform as
\begin{equation}
     \hat{g}_e\mapsto U(\mathbf{h})\,\hat{g}_e \,U^\dag(\mathbf{h})=\rho(h^{-1}_f)\,\hat{g}_e\,\rho(h_i)\,,
\end{equation}
where $i,f$ label the initial and final vertex of $e$, while bulk electric fields transform in their respective adjoint representation $D$:
\begin{align}
   L^a_e&\mapsto U(\mathbf{h})\,L^a_e \,U^\dag(\mathbf{h})=D(h_f^{-1})^{ab}\,L^b_e\\
   R^a_e&\mapsto U(\mathbf{h})\,R^a_e\, U^\dag(\mathbf{h})=D(h_i^{-1})^{ab} \,R^b_e\,.
\end{align} 

The lattice analog of large gauge transformations acts exclusively on the vertices in the boundary $\mathcal{B}$ with the same convention as in Eq.~\eqref{eq:gaugetransfconv} and is assumed to leave the boundary conditions (left implicit) invariant. As in the continuum, we will treat these as \emph{physical} transformations; no Gauss law is imposed on boundary vertices. 

The reason we explicitly include $\mathcal{B}$ in our discussion\,---\,unlike previous works on entanglement entropy in lattice gauge theories \cite{Buividovich:2008gq,Donnelly:2011hn,Donnelly:2014gva,Casini:2013rba,Ghosh:2015iwa,Soni:2015yga,VanAcoleyen:2015ccp,Delcamp:2016eya,Aoki:2015bsa,Hung:2015fla}\,---\,is that it will permit us to define boundary-anchored Wilson line QRFs for the subregion of interest \cite{Carrozza:2021gju,Araujo-Regado:2024dpr}. Our discussion of entanglement entropy will thereby connect with the edge mode and corner symmetry program \cite{Donnelly:2016auv,Donnelly:2014fua,Donnelly:2015hxa,Geiller:2019bti,Ball:2024hqe,Ball:2024gti,Carrozza:2021gju,Araujo-Regado:2024dpr,Gomes:2018shn,Gomes:2019xto,Riello:2021lfl,Freidel:2023bnj,Ciambelli:2022vot,Ciambelli:2021vnn, Freidel:2020xyx,Freidel:2020svx,Freidel:2023bnj,Donnelly:2020xgu,Donnelly:2022kfs}. 

\section{Gauge QRFs on the lattice}
\label{sec:qrf}

Among the approaches to QRFs, it is the \emph{perspective-neutral} framework \cite{delaHamette:2021oex,Hoehn:2023ehz,AliAhmad:2021adn,Hoehn:2019fsy,Hoehn:2020epv,Carrozza:2024smc,DeVuyst:2024pop,DeVuyst:2024uvd,Vanrietvelde:2018dit,Vanrietvelde:2018pgb,Hoehn:2021flk,delaHamette:2021piz,Giacomini:2021gei,Castro-Ruiz:2019nnl,Suleymanov:2025nrr,Hoehn:2023axh}  that is built for systems with gauge symmetries and thus applies directly here (see \cite[Sec.~II]{Hoehn:2023ehz} for an introduction and comparison with special covariance). In particular, our construction will be based on its general group formulation in \cite{delaHamette:2021oex}. 

In this approach, a QRF $R$ for the gauge group $\mathcal{G}$ is, roughly speaking, a choice of split between redundant (the QRF) and non-redundant data, comprehensively called the system $S$. This split is typically accompanied by a corresponding factorization $\H_{\rm kin}\simeq\H_R\otimes\H_S$ and a unitary product representation of $\mathcal{G}$, $U(\mathbf{g})=U_R(\mathbf{g})\otimes U_S(\mathbf{g})$, $\mathbf{g}\in\mathcal{G}$. While such factorizations are generally unavailable in the continuum for quantum field theories, they do exist on the lattice, as we shall see. 

In order for $R$ to absorb all gauge redundancy, $\mathcal{G}$ must act regularly on its configurations, called \emph{orientations}, which are then in one-to-one correspondence with $\mathcal{G}$ and so can be labeled by group elements $\ket{\mathbf{g}}_R$. These orientation states are generalized coherent states, transforming covariantly $U_R(\mathbf{g}')\ket{\mathbf{g}}_R=\ket{\mathbf{g}'\mathbf{g}}_R$. If $\{\ket{\mathbf g}_R\}_{\mathbf{g}\in\mathcal{G}}$ furthermore furnishes a basis for $\H_R$, then $R$ is \emph{complete}. Note that the orientation states may also comprise a basis of $\H_R$ when $\mathcal{G}$ acts transitively, but not freely on them, i.e.\ when each orientation state is invariant under an isotropy subgroup of $\mathcal{G}$. In this case, $R$ cannot provide a complete parametrization of $\mathcal{G}$-orbits and, accordingly, it will be called an \emph{incomplete} QRF. Below, we will find both complete and incomplete QRFs on the lattice. 

The aim is to describe $S$ relative to $R$, which can be done equivalently by gauge fixing $R$, or by gauge-invariantly dressing $S$ with $R$. In general, there are many distinct $RS$ splits.

Since $\mathcal{G}\simeq G^V$ acts independently at each bulk node, a complete QRF for $\mathcal{G}$ for our lattice $\mathcal{L}$ is given by a ``field'' of complete QRFs $R_v$ for the structure group $G$ at each $v\in\mathcal{V}_{\rm b}$. Thus, $\H_R=\bigotimes_{v\in\mathcal{V}_{\rm b}}\H_{R_v}$, such that $U_v$ acts regularly on a basis of $\H_{R_v}$ (and trivially on any $\H_{R_{v'\neq v}}$); this requires a nonlocal refactorization of $\H_{\rm kin}$.

Such an $R$ is afforded by the Wilson lines in any spanning tree/forest $\mathcal{T}$: a graph without loops reaching all bulk nodes and rooted on (possibly multiple points
of) $\mathcal{B}$, where $\mathcal{G}$ acts trivially. To formalize this, consider a path $\gamma$ connecting the nodes $v_0,v_1,\ldots,v_n$ and refactorize $\bigotimes_{e\in\gamma}\H_e\to\bigotimes_{i=1}^n\H_{R_{v_i}}$ by performing a unitary change of basis from edge-wise to cumulative node-wise group variables associated with Wilson lines anchored on $v_0$. If the Wilson lines and $e\in\gamma$ are uniformly oriented, this reads $\ket{g_1}\otimes\cdots\otimes\ket{g_n}=\ket{g_{v_1}}\otimes\cdots\otimes\ket{g_{v_n}}$ with $g_{v_1}=g_1$ and $g_{v_i}=g_{i}g_{v_{i-1}}$ for $i>1$ (see Fig.~\ref{fig:Wilson_line_reparam} and App.~\ref{app_kinrefactor} for the associated reorganization of the electric fields); otherwise, for each edge pointing in opposite direction to the Wilson line, replace $g_i\to g_i^{-1}$ in the $g_{v_i}$.\\

\begin{minipage}{.45\textwidth}
    \centering
    \vspace{5pt}
    \captionsetup{hypcap=false} 
    \includegraphics[width=0.9\linewidth]{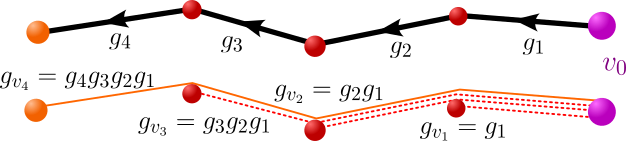}
    \captionof{figure}{\small Nonlocal reparametrization from edge-wise link variables to Wilson lines anchored on $v_0$ (violet).} 
\label{fig:Wilson_line_reparam}
\end{minipage}\\

The gauge transformations at the nodes act as $\otimes_{i=1}^n U_{v_1}(h_1)\triangleright\ket{g_{v_1},\ldots,g_{v_n}}=\ket{h_1 g_{v_1},\ldots,h_n g_{v_n}}$, and so $\{\ket{g_{v_i}}\}_{g_{v_i}\in G}$ furnishes an orthonormal orientation basis for the local $R_{v_i}$ with $\H_{R_{v_i}}\simeq L^2(G)$. Such QRFs with perfectly distinguishable orientation states are called \emph{ideal} \cite{delaHamette:2021oex,Hoehn:2023ehz,Hoehn:2019fsy}. Applying this procedure to all of $\mathcal{T}$ yields an ideal lattice QRF field $R$ for $\mathcal{G}$ with $\H_R\simeq L^2(G^V)$, while $S$ is given by the complement of $\mathcal{T}$.

The orthogonal projector onto the gauge-invariant kinematical subalgebra $\mathcal{A}^{\mathcal{G}}_{\rm kin}:=\mathcal{G}(\mathcal{B}(\H_{\rm kin}))\subset\mathcal{B}(\H_{\rm kin})$ is given by the $\mathcal{G}$-twirl (\emph{incoherent} group average), ${\mathcal{G}(\bullet):=1/\rm{Vol}(\mathcal{G})\int_\mathcal{G}\dd\mathbf{g}\, U(\mathbf{g})\bullet U^\dag(\mathbf{g})}$. $\mathcal{A}_{\rm kin}^\mathcal{G}$ contains the relational observables describing $S$ relative to $R$ \cite{delaHamette:2021oex,Hoehn:2019fsy,Hoehn:2021flk,Hoehn:2023ehz,DeVuyst:2024uvd}:
\begin{equation}\label{eq:relobs}
    O_{f_S|R}(\mathbf{g})=\mathcal{G}\left(\ket{\mathbf{g}}\!\bra{\mathbf{g}}_R\otimes f_S\right)\,,
\end{equation}
which measures $f_S$ conditional on $R$ being in orientation $\mathbf{g}\in\mathcal{G}$, where $f_S\in\mathcal{B}(\H_S)$. In App.~\ref{app:relational_obs}, we show that, when $f_S$ is some bulk Wilson line, $O_{f_S|R}$ is either a bulk Wilson loop or a Wilson line starting and ending on $\mathcal{B}$. Furthermore, when $f_S=L^a_e,R^a_e$ is a left or right electric field at some bulk vertex $v$, then $O_{f_S|R}$ is only non-zero when $R$'s Wilson line to $v$ is anchored on $\mathcal{B}$. Thus, the electric flux or color charge $f_S=E^a_e:=L^a_e-R^a_e$ along a bulk edge $e$ \cite{kogut1975} can only be dressed to the global boundary. On the other hand, consider two bulk vertices $v,w$ and $R$ such that its Wilson lines touching these vertices are not anchored on $\mathcal{B}$. Then objects of second order in the electric fields, such as 
 $f_S=\l\mathfrak{G}^v_e\r^a\l\mathfrak{G}^w_{e'}\r^b$, where $\mathfrak{G}^v_e$ is the appropriate choice of $L_e,R_e$ at vertex $v$, yield relational observables of the form $O_{f_S|R}\propto(\hat{\mathfrak{G}}^v_e)^a(\hat{g}_\ell^\text{Adj})^{ab}(\hat{\mathfrak{G}}^w_{e'})^b$, where $\ell$ is $R$'s unique bulk Wilson line connecting $v$ and $w$. All observables can be built from the ones we discussed. Thus, all pure bulk-supported relational observables are functions of two-sided dressings.

Let us now restrict attention to a finite region $\Sigma$ and QRF fields on its boundary $\p \Sigma$, cf.~Fig.~\ref{fig:Lattice_2d}; the details of $R$ elsewhere will not matter in what follows. These QRFs  on $\p\Sigma$ are lattice versions of edge modes \cite{Carrozza:2021gju,Araujo-Regado:2024dpr}. Echoing the classical discussion \cite{Araujo-Regado:2024dpr}, we distinguish two types, both of which are ideal here:
\begin{itemize}
  \item \emph{Extrinsic QRFs} $\Phi$, built with Wilson lines from $\mathcal{B}$ to each node in $\partial \Sigma$, which do not share  edges, see Fig.~\ref{fig:Lattice_2d}. Thus, $\H_\Phi=\bigotimes_{v\in\p\Sigma}\H_{R_v}\simeq L^2(G^{V_{\p\Sigma}})$ with each $v_0\in\mathcal{B}$ and $V_{\p\Sigma}=|\mathcal{V}\cap\p\Sigma|$. 
  \item \emph{Intrinsic QRFs} $\tilde\Phi$, constructed from Wilson lines along a spanning tree on $\partial \Sigma$ rooted at an arbitrary node $v_0\in\p\Sigma$, see Fig.~\ref{fig:intrinsic frame}. Hence, $\H_{\tilde\Phi}=\bigotimes_{\p\Sigma\ni v\neq v_0}\H_{\tilde{R}_v}\simeq L^2(G^{(V_{\p\Sigma}-1)})$.
\end{itemize}

$\Phi$ is therefore complete for gauge transformations on $\p\Sigma$, while $\tilde\Phi$ is not: it cannot deparametrize one factor of $G$, e.g.\ at its root $v_0$. Indeed, for both types, the local QRF $R_v$ at $v\in\p\Sigma$ transforms under gauge transformations at its root $v_0$ as $U_{v_0}(h_0)\triangleright\ket{g_v}=\ket{g_v h_0^{-1}}$. The difference is that, for extrinsic frames, $U_{v_0}(h_0)$ is a \emph{large} gauge transformation and thus physical, constituting a \emph{reorientation} of the QRF \cite{Carrozza:2021gju,Araujo-Regado:2024dpr,delaHamette:2021oex,Hoehn:2023ehz}, while for intrinsic frames it is \emph{small} and therefore redundancy. Reorientations change the relation between $\Phi$ and $\Sigma$, commute with small gauge transformations and, in the continuum, realize the corner symmetries of $\p\Sigma$ \cite{Carrozza:2021gju,Araujo-Regado:2024dpr}. The group $\mathbb{G}_{\p\Sigma}$ of $\Phi$'s reorientations is isomorphic to the gauge group on the corner, $\mathbb{G}_{\p\Sigma}\simeq G^{V_{\p\Sigma}}\simeq\mathcal{G}_{\p\Sigma}$, and is generated by the right electric fields $R^a_e$ on the $V_{\p\Sigma}$ edges inside $\Phi$'s Wilson lines that are incident on $\mathcal{B}$. We identify it as the lattice corner symmetry group; when all of $\Phi$'s anchor points on $\mathcal{B}$ are distinct, all of its reorientations can be realized by large gauge transformations. As in the continuum, $\tilde\Phi$ admits no reorientations \cite{Araujo-Regado:2024dpr}. 

\noindent\begin{minipage}{0.48\textwidth}
\centering
    \vspace{5pt}
    \captionsetup{hypcap=false} 
\includegraphics[width=0.8\linewidth]{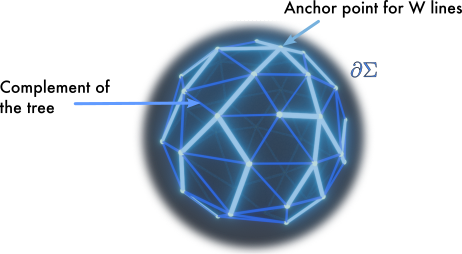}
    \captionof{figure}{\small Illustration of a spanning tree defining an intrinsic edge mode QRF $\tilde\Phi$ on a two-dimensional entangling surface $\p \Sigma$, in lighter color highlighted with respect to its complement.}
\label{fig:intrinsic frame}
\end{minipage}\\

Since only two-sided dressings matter for intrinsic frames, the incompleteness of $\tilde\Phi$ is not an obstruction to building small gauge-invariant observables on $\p\Sigma$ via Eq.~\eqref{eq:relobs}. The system $S_{\p\Sigma}$ on $\p\Sigma$ is given by the $|\mathcal{E}\cap\p\Sigma|-(V_{\p\Sigma}-1)$ edges in the complement of the tree defining $\tilde\Phi$. For corners of dimension $D\leq2$, the Euler characteristic  guarantees that these edges are in one-to-one correspondence with independent Wilson loops in $\p\Sigma$, which encode magnetic degrees of freedom intrinsic to $\p\Sigma$. The loops result from inserting the non-tree link variables into $f_S$ in Eq.~\eqref{eq:relobs} and using $\tilde\Phi$ in $R$, which glues two of its Wilson lines at $v_0$. The electric conjugate is given by inserting a quadratic combination of $L_e,R_e$ for $f_S$ instead; there is only one independent intrinsically dressed electric field per non-tree edge as all $\rho$-representation indices must be contracted to obtain a singlet (and $L_e,R_e$ are related by the link adjoint). Note that this is different for extrinsic dressings, where $\rho$-indices on $\mathcal{B}$ may remain open (cf.~App.~\ref{app:relational_obs}).  Thus, the incompleteness of $\tilde\Phi$ means that it yields less independent relational observables for $S_{\p\Sigma}$ than $\Phi$, except when $G$ is Abelian and $\rho$ an irrep (hence, one-dimensional). 

The extrinsic and intrinsic QRFs yield a change of tensor product structure (TPS) on $\H_{\rm kin}$. First, the regional boundary Hilbert space refactorizes as $\H_{\p\Sigma}\simeq \H_{\tilde\Phi}\otimes\H^{\rm phys}_{\p\Sigma\setminus\tilde\Phi|\tilde\Phi}$ with $\H^{\rm phys}_{\p\Sigma\setminus\tilde\Phi|\tilde\Phi}:=\bigotimes_{\text{loops}\in\p\Sigma} L^2(G)$, and so
\be \label{eq:kin_refact} 
\Hkin \simeq \H_\Phi \otimes \H_{\tilde{\Phi}} \otimes \H_{\mathring{\Sigma}} \otimes \H_{\Bar{\Sigma}\backslash\Phi} \otimes \H^{\rm phys}_{\p\Sigma\setminus\tilde\Phi|\tilde\Phi}\,. \ee 
Here,  
$\H_{\mathring{\Sigma}}=\bigotimes_{e\in\Sigma\setminus\p\Sigma}\H_e$ contains all internal regional degrees of freedom, while $\H_{\Bar{\Sigma}\backslash\Phi}$ contains all degrees of freedom of the complement that do not belong to $\Phi$; this part depends on the choice of extrinsic QRF.

\section{Physical Hilbert space factorizations}
\label{sec:phys}

The (small) gauge-invariant physical Hilbert subspace ${\H_{\rm phys}=\Pi_{\rm phys}\left(\H_{\rm kin}\right)}\simeq L^2(G^E/G^V)$ is the image of the orthogonal \emph{coherent} group averaging projector ${\Pi_{\rm phys}:=1/\rm{Vol}(\mathcal{G})\int_\mathcal{G}\dd\mathbf{g}\,U(\mathbf{g})}$. As a subspace, it generically does not inherit the {kinematical} TPS, which is the main obstacle to unambiguously assigning a physical subsystem to the kinematically defined $\Sigma$ \cite{VanAcoleyen:2015ccp,Donnelly:2011hn,Casini:2013rba,Donnelly:2014gva,Soni:2015yga,Ghosh:2015iwa}. 

A core insight of the internal QRF program is that, rather than invoking an unphysical \emph{kinematical} notion of subsystem decompositions and locality, one must define them in a physical, \emph{gauge-invariant} manner \cite{Hoehn:2023ehz,AliAhmad:2021adn}. All physical properties, specifically entanglement entropies, should be defined relative to a physical subsystem partition. Crucially, the internal perspective of any QRF \emph{is} exactly such a physical partition and for ideal complete QRFs further is a \emph{gauge-invariant TPS} on $\Hphys$ \cite{Hoehn:2023ehz,DeVuyst:2024uvd}.

The internal perspective of some QRF $R$ on its complement $S$ is given by the relational observables \eqref{eq:relobs} (evaluated on $\Hphys$, where they are complete \cite{delaHamette:2021oex}), or, equivalently, by gauge fixing its orientation with a (generalized) Page-Wootters (PW) reduction, ${\mathcal{R}_R(\mathbf{g}):=\bra{\mathbf{g}}_R\otimes \mathds{1}_S}:\Hphys\to\H_S$ \cite{delaHamette:2021oex,Hoehn:2019fsy}. This map is unitary onto its image, which for ideal complete QRFs is all of $\H_S$. In that case, when $\H_S=\bigotimes_\alpha\H_{S_\alpha}$ has subsystem structure, the unitary $\mathcal{R}_R$ constitutes a gauge-invariant TPS on $\Hphys\simeq\bigotimes_\alpha\H_{S_\alpha}$ (see App.~\ref{app:factor} and \cite{Hoehn:2023ehz,DeVuyst:2024uvd}). In particular, 
\begin{equation}\label{eq:obsred}
\mathcal{R}_R(\mathbf{g})O_{f_S|R}(\mathbf{g})\mathcal{R}^\dag_R(\mathbf{g})=f_S\in\bigotimes_\alpha\mathcal{B}(\H_{S_\alpha})\,, 
\end{equation}
so relational observables are local w.r.t.\ this TPS \cite{delaHamette:2021oex,Hoehn:2023ehz,Hoehn:2019fsy}. We have transitioned from a kinematical to a gauge-invariant notion of subsystem locality; it is the local subsystem structure ``as seen'' by $R$ and thus \emph{relational}. Indeed, different choices of QRF yield different TPSs on $\Hphys$ \cite{Hoehn:2023ehz,DeVuyst:2024uvd}, a phenomenon known as \emph{subsystem relativity} 
\cite{Hoehn:2023ehz,AliAhmad:2021adn,delaHamette:2021oex,DeVuyst:2024pop,DeVuyst:2024uvd,Castro-Ruiz:2021vnq,AliAhmad:2024wja} and responsible for QRF-dependent entanglement entropies \cite{Hoehn:2023ehz,DeVuyst:2024pop,DeVuyst:2024uvd}. 

This observation directly applies here: the global lattice QRF field $R$ is complete and ideal, and so $\mathcal{R}_R(\mathbf{g}):\Hphys\to\H_S=\bigotimes_{e\in\mathcal{E}\setminus\mathcal{T}}\H_e$ defines a physical TPS, including a factorization between what $R$ ``sees'' as the region $\H_{\Sigma|R}:=\bigotimes_{e\in(\mathcal{E}\setminus\mathcal{T})\cap(\Sigma\cup\p\Sigma)}\H_e$ and its complement $\H_{\bar\Sigma|R}:=\bigotimes_{e\in(\mathcal{E}\setminus\mathcal{T})\cap\bar\Sigma}\H_e$. \emph{Relative} to  $R$, we thus have everything we need for standard von Neumann entanglement entropy calculations. 

In fact, for our purposes it suffices to gauge fix only on the boundary $\p\Sigma$, since $\mathcal{G}$ acts on each bulk vertex independently. To this end, we may invoke the extrinsic edge mode frame $\Phi$, which is ideal and complete for gauge transformations on $\p\Sigma$. In App.~\ref{app:factor}, we show that its PW reduction $\mathcal{R}_\Phi(\mathbf{g}_{\p\Sigma})$ defines a (coarser) tensor factorization between inside and outside
\begin{equation}\label{eq:extfacthphys}
\Hphys\simeq\H^{\rm phys}_{\mathring{\Sigma}}\otimes\H_{\p\Sigma\setminus\tilde\Phi}\otimes\H_{\tilde\Phi}\otimes\H_{\bar\Sigma\setminus\Phi}^{\rm phys}\,,
\end{equation}
where the `phys' superscript denotes that all bulk gauge constraints are implemented, \emph{except} on the corner $\p\Sigma$. For later, let us combine the first three factors into $\H_{\Sigma}^{\rm phys}:=\H^{\rm phys}_{\mathring{\Sigma}}\otimes\H_{\p\Sigma\setminus\tilde\Phi}\otimes\H_{\tilde\Phi}
$. 

It will be convenient to isolate the regional degrees of freedom transforming non-trivially under the corner symmetries $\mathbb{G}_{\p\Sigma}$, which act non-trivially on all four factors above. In non-Abelian theories this turns out to be possible only on a subspace of $\mathcal{H}^\text{phys}_\Sigma$.
We explain in App.~\ref{app:GM} how to reorganize the regional factor $\H^{\rm phys}_{\Sigma}$ by redressing as many regional degrees of freedom as possible with the intrinsic frame $\tilde\Phi$ (which kinematically transforms trivially under $\mathbb{G}_{\p\Sigma}$), while dressing $\tilde\Phi$ itself with the extrinsic frame, yielding the relational observables $O_{\tilde\Phi|\Phi}$ in Eq.~\eqref{eq:relobs}, which encode their relative orientations. As in the continuum \cite{Araujo-Regado:2024dpr}, we identify this lattice field $O_{\tilde\Phi|\Phi}$ on $\p\Sigma$ as the corner-symmetry \emph{Goldstone mode} (GM) since, on the appropriate subspace of $\mathcal{H}^\text{phys}_\Sigma$, these are now the only regional degrees of freedom that transform nontrivially under $\mathbb{G}_{\p\Sigma}$. This is because they are the only regional degrees of freedom that encode (and parametrize) the relation between $\Sigma$ and $\bar{\Sigma}$. The end result is that this does not quite result in a new factorization of the full $\Hphys$, but rather we get
\begin{align}\label{eq:Hphysref}
&\Hphys\simeq\\
&\quad\l\l\mathcal{H}_\text{GM}\otimes\mathcal{H}_{\Sigma\setminus\tilde\Phi}^\text{int}\r\oplus\l \mathcal{H}_{\Tilde{\Phi}}\otimes P^\mathbb{H}(\mathcal{H}_{\Sigma\setminus\tilde\Phi}^\text{phys})^\perp\r\r\otimes \mathcal{H}_{\bar{\Sigma}\setminus\Phi}^\text{phys}\,,\nonumber
\end{align}
due to $\tilde\Phi$'s incompleteness, which entails that its orientation basis admits an isotropy subgroup $H\simeq G\subset\mathcal{G}_{\p\Sigma}$, leaving its orientations invariant (up to phase). This means that it can only distinguish other degrees of freedom that are invariant under its own isotropy group \cite{delaHamette:2021oex,Chataignier:2024eil}. It does not ``see'' all regional degrees of freedom and dressing some of them yields trivial observables. Indeed, we saw earlier that dressing linear electric field combinations with an intrinsic frame returns zero. This subspace that $\tilde\Phi$ does not ``see'' is given by the second summand in the orthogonal sum of Eq.~\eqref{eq:Hphysref}; it includes the regional degrees of freedom in the orthogonal complement of those invariant under its isotropy group.  

On the other hand, the subspace that $\tilde\Phi$ does ``see'' is given by the first summand in the direct sum of Eq.~\eqref{eq:Hphysref}. This subspace splits cleanly into a GM sector that captures all the regional degrees of freedom charged under reorientations and those (in $\H^{\rm int}_{\Sigma\setminus\tilde\Phi}$) that are invariant. Indeed, in this subspace, the corner group's unitary representation is of the form 
\begin{equation}
\mathcal{U}(\mathbb{G}_{\p\Sigma})=\mathbb{U}_{\rm GM}(\mathbb{G}_{\p\Sigma})\otimes\mathds{1}_{\Sigma\setminus\tilde\Phi}\otimes\mathbb{U}_{\bar\Sigma\setminus\Phi}(\mathbb{G}_{\p\Sigma})\,.
\end{equation}

As we show in App.~\ref{app:GM}, for Abelian theories, the second subspace in Eq.~\eqref{eq:Hphysref} is trivial and thus, in this case, one obtains a clean factorization  on the full $\Hphys$ between the regional GM sector, the regional reorientation-invariant data (the intrinsic frame-dressed information), and the complement of the region. For the first two factors, this is a lattice Hilbert space version of the continuum regional phase space factorization established for Maxwell theory in \cite{Ball:2024hqe,Araujo-Regado:2024dpr}. Perturbatively, it was also shown in the continuum (under certain assumptions) that such a clean factorization is not consistent with the dynamics in Yang-Mills theory \cite{Ball:2024gti}, being broadly consistent with our nonperturbative lattice argument above that in that case no such factorization exists in the quantum theory.

Importantly, the TPS in Eq.~\eqref{eq:extfacthphys} (and the reorganization in Eq.~\eqref{eq:Hphysref}) is \emph{relative to the extrinsic QRF $\Phi$}; two distinct choices $\Phi,\Phi'$ define two distinct TPSs and realizations of $\mathbb{G}_{\p\Sigma}$ (see App.~\ref{app:subrel} for details on the following). In particular, denote by $\mathcal{A}^{\rm phys}_{S_\alpha|\Phi}=\mathcal{R}_\Phi^\dag(\mathbf{g}_{\p\Sigma})\mathcal{A}_{S_\alpha}\mathcal{R}_\Phi(\mathbf{g}_{\p\Sigma})$ the algebra of relational observables on $\Hphys$, describing $S_\alpha$ relative to $\Phi$, where $S_\alpha$ can be any of $\mathring{\Sigma},\p\Sigma\setminus\tilde\Phi,\tilde\Phi,\bar\Sigma\setminus\Phi$ above, or any of their unions, and $\mathcal{A}_{S_\alpha}:=\mathcal{B}(\H^{\rm phys}_{S_\alpha})$. As $\mathcal{R}_\Phi$ is unitary and $\mathcal{A}_{S_\alpha}$ is a Type I factor, so is $\mathcal{A}_{S_\alpha|\Phi}^{\rm phys}$. For any two non-overlapping extrinsic frames, we have
\begin{equation}\label{eq:subrel}
    \mathcal{A}^{\rm phys}_{S_\alpha|\Phi}\cap\mathcal{A}^{\rm phys}_{S_\alpha|\Phi'}=\mathcal{A}^{\mathcal{G}_{\p\Sigma}}_{S_\alpha}\Pi^{\rm phys}_{\p\Sigma}\,,
\end{equation}
 where $\Pi_{\p\Sigma}^{\rm phys}$ denotes coherent group averaging over the remaining gauge transformations on the corner. Thus, $\Phi,\Phi'$ assign distinct relational algebras to $S_\alpha$ and share only those $S_\alpha$-observables, which are intrinsically invariant without dressing to either $\Phi,\Phi'$. The two extrinsically dressed algebras differ exactly in the cross-boundary observables, e.g.\ Wilson lines crossing the corner and hitting $\mathcal{B}$ with both ends, see Fig.~\ref{fig:algebras&sub_rel}. This is subsystem relativity on the lattice. In particular, $\mathcal{A}^{\rm phys}_{S_\alpha|\Phi}$ is $S_\alpha$-local in the TPS $\mathcal{R}_\Phi$ by \eqref{eq:obsred}, but not in $\mathcal{R}_{\Phi'}$.\\

The perspective of an \emph{intrinsic edge mode frame} $\tilde\Phi$, on the other hand, generally does not define a TPS on $\Hphys$, see App.~\ref{app:factor} for details. Again, this is a consequence of its incompleteness. Its perspective is a \emph{subspace} of  $\H_{\mathring{\Sigma}}^{\rm phys}\otimes\H_{\p\Sigma\setminus\tilde\Phi}\otimes\H_{\Phi}\otimes\H^{\rm phys}_{\bar\Sigma\setminus\Phi}$, defined by an average over $\tilde\Phi$'s isotropy group $
H$ that acts globally on all factors; $\tilde\Phi$ can only resolve what is invariant under its isotropy group \cite{delaHamette:2021oex,Chataignier:2024eil}. Interestingly, this is a lattice and quantum realization of \cite[Thm.~6.1]{Gomes:2019xto}, asserting that, given two regional classical solutions to the constraints in $\Sigma,\bar\Sigma$, their gluing to a global solution is fixed up to a constant isotropy subgroup of $\mathcal{G}_{\p\Sigma}$. 
An important exception is when $G$ is Abelian, in which case $\tilde\Phi$ \emph{does} define a TPS, factorizing between in- and outside. This is because, due to the two-sided dressings, the action of the constantly acting isotropy group cancels. 

In summary, it is the kinematically nonlocal construction of the QRFs that allows one to trade the inherent nonlocality of invariant lattice states for  relationally local TPSs on $\Hphys$. This contrasts to previous treatments of lattice regions, which embed $\Hphys$ into an \emph{auxiliary} Hilbert space that is factorized between in- and outside to define a partial trace and entropies\,---\,and which is the source of non-distillable entanglement contributions \cite{Buividovich:2008gq,Donnelly:2011hn,Donnelly:2014gva,Ghosh:2015iwa,VanAcoleyen:2015ccp,Delcamp:2016eya,Aoki:2015bsa}. It is also the reason why the algebra of all gauge-invariant observables \emph{kinematically} supported in $\Sigma$ features a nontrivial center, i.e.\ a nontrivial subalgebra commuting with all other elements \cite{Casini:2013rba,Soni:2015yga,VanAcoleyen:2015ccp,Delcamp:2016eya,Bianchi:2024aim}.\\

\begin{minipage}{.45\textwidth}
    \centering
    \vspace{5pt}
    \captionsetup{hypcap=false} 
    \includegraphics[width=0.48\linewidth]{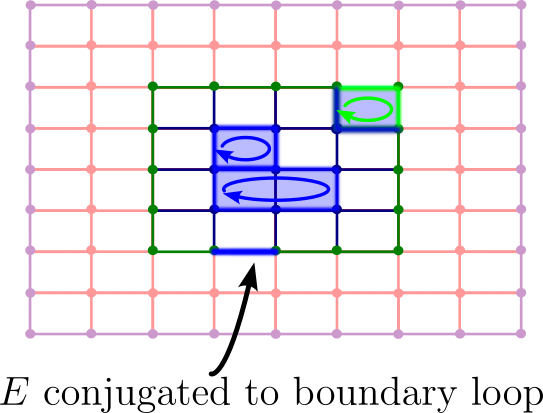}
    \includegraphics[width=0.48\linewidth]{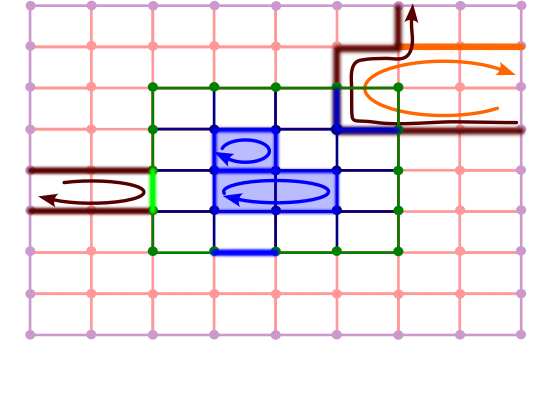}
    \captionof{figure}{\small \emph{Left:} Intrinsic algebra $\mathcal{A}_{\Sigma\setminus\tilde\Phi|\tilde\Phi}^{\rm phys}$ with loops and only one electric field on the boundary in the complement of $\tilde\Phi$; intrinsic frame dressing (light green). Distinct intrinsic frames define different algebras: they share all regional loops, but differ in the boundary tangential electric fields. \emph{Right:} Extrinsic algebra $\mathcal{A}^{\rm phys}_{\Sigma|\Phi}$ and subsystem relativity: two distinct extrinsic frames (fat orange and brown) dress regional variables (fat blue) to distinct cross-boundary observables (top right).}
\label{fig:algebras&sub_rel}
\end{minipage}\\

\section{A hierarchy of regional algebras}
\label{sec:alg}

Let us now put all of this into one coherent picture, with details in Apps.~\ref{app:alg_content}--\ref{app_magneticcenter}.

We have seen that, relative to any extrinsic frame, $\mathcal{A}^{\rm phys}_{S_\alpha|\Phi}$ is a factor, hence centerless. The total physical algebra factorizes accordingly
\begin{equation}
    \mathcal{A}_{\rm phys}:=\mathcal{B}(\Hphys)=\mathcal{A}^{\rm phys}_{\mathring{\Sigma}|\Phi}\otimes\mathcal{A}^{\rm phys}_{\p\Sigma\setminus\tilde\Phi|\Phi}\otimes\mathcal{A}^{\rm phys}_{\tilde\Phi|\Phi}\otimes\mathcal{A}^{\rm phys}_{\bar\Sigma\setminus\Phi|\Phi}\,,\nonumber
\end{equation}
with the first factors, ${\mathcal{A}^{\rm phys}_{\Sigma|\Phi}:=\mathcal{A}^{\rm phys}_{\mathring{\Sigma}|\Phi}\otimes\mathcal{A}^{\rm phys}_{\p\Sigma\setminus\tilde\Phi|\Phi}\otimes\mathcal{A}^{\rm phys}_{\tilde\Phi|\Phi}}$, constituting the subsystem that $\Phi$ associates with region $\Sigma$ (which can be equivalently reorganized according to Eq.~\eqref{eq:Hphysref}). Let us call it an \emph{extrinsic relational algebra} of $\Sigma$; it contains all relational observables describing $\Sigma$ relative to $\Phi$, including some cross-boundary data, and is subject to subsystem relativity, cf.~Fig.~\ref{fig:algebras&sub_rel}. In fact, it is generated by cross-boundary observables, namely open regional Wilson lines and linear electric fields $L_e,R_e$ $\Phi$-dressed to $\mathcal{B}$. As the representation indices on $\mathcal{B}$ are uncontracted, one can combine these generators to also obtain all gauge-invariant observables purely supported in $\Sigma$. 

The GM is a \emph{gauge-invariant} QRF-field for the \emph{physical} electric corner symmetries $\mathbb{G}_{\p\Sigma}$. Its orientations parametrize $\mathbb{G}_{\p\Sigma}$-orbits and thereby those relations of $\Sigma$ with $\bar\Sigma$ captured by the dressing with the extrinsic $\Phi$. Averaging over these orientations eliminates all cross-boundary relational data. As this is a physical symmetry, averaging is implemented by \emph{incoherently} twirling over $\mathbb{G}_{\p\Sigma}$. Indeed, this is the average of the quantum information \cite{Bartlett:2006tzx,Castro-Ruiz:2021vnq,Krumm:2020fws} and operational \cite{Loveridge:2017pcv,Carette:2023wpz,Fewster:2024pur} QRF frameworks, which are the appropriate ones for physical symmetries (here on $\Hphys$) and essentially agree for compact groups.
(By contrast, $\Phi,\,\tilde\Phi,\,R$ are non-invariant QRFs for the \emph{gauge} symmetries $\mathcal{G}_{\p\Sigma}$ on $\Hkin$, for which \emph{coherent} averaging and the perspective-neutral framework invoked so far are appropriate.)
The result is \emph{superselection} across $\mathbb{G}_{\p\Sigma}$'s irreps and the algebra of all gauge-invariant observables kinematically supported in $\Sigma$:
\begin{equation}\label{eq:electriccenterdef}
    \mathcal{A}_E:=\mathbb{G}_{\p\Sigma}\left(\mathcal{A}^{\rm phys}_{\Sigma|\Phi}\right)=\mathcal{A}^{\mathcal{G}_{\p\Sigma}}_{\Sigma}\,\Pi^{\rm phys}_{\p\Sigma}=\!\!\!\bigcap_{\text{{extrinsic QRFs }}\Phi}\mathcal{A}^{\rm phys}_{\Sigma|\Phi}\,.
\end{equation}
This recovers precisely the \emph{electric center algebra} of \cite{Casini:2013rba,Soni:2015yga,VanAcoleyen:2015ccp}; the center consists of the irrep Casimir operators, which are given by  the squared \emph{exterior and normal} electric fields on $\p\Sigma$.
The first equality holds for any extrinsic QRF. The last equality is implied by subsystem relativity in Eq.~\eqref{eq:subrel} and holds on any sufficiently large lattice that non-overlapping extrinsic QRFs exist. That is, the electric center algebra is exactly what all extrinsic QRFs agree on.

We noted that, owing to its incompleteness, an intrinsic QRF $\tilde\Phi$ can only resolve what is invariant under its isotropy group $H$. This is reflected in its relational observable algebra on $\Hphys$ describing its regional complement: 
\begin{equation}\label{eq:incintalg}
  {\mathcal{A}}^{\rm phys}_{\Sigma\setminus\tilde\Phi|\tilde\Phi}:=\mathcal{R}^\dag_{\tilde\Phi}(\mathring{\mathbf{g}}_{\p\Sigma})\left(\mathcal{A}_{\mathring{\Sigma}}\otimes\mathcal{A}_{\p\Sigma\setminus\tilde\Phi}\right)^H\mathcal{R}_{\tilde\Phi}(\mathring{\mathbf{g}}_{\p\Sigma})\,,
\end{equation}
where $\left(\mathcal{A}_{\mathring{\Sigma}}\otimes\mathcal{A}_{\p\Sigma\setminus\tilde\Phi}\right)^H$ is the \emph{jointly isotropized}, i.e.\ $H$-invariant (twirled) subalgebra of $\mathcal{A}_{\mathring{\Sigma}}\otimes\mathcal{A}_{\p\Sigma\setminus\tilde\Phi}=\mathcal{B}(\H^{\rm phys}_{\mathring{\Sigma}})\otimes\mathcal{B}(\H_{\p\Sigma\setminus\tilde\Phi})$. We call it an \emph{intrinsic relational algebra}; it contains \emph{all} relational observables describing $\tilde\Phi$'s regional complement $\Sigma\setminus\tilde\Phi$ relative to it. This algebra is generated by all regional Wilson loops and quadratic combinations of electric fields in $\Sigma\setminus\tilde\Phi$ dressed with regional Wilson lines. In contrast to extrinsic algebras, all representation indices must be contracted, in order to obtain singlets for gauge transformations. Thus, any intrinsic algebra contains ($\tilde\Phi$-dressed) tangential electric fields only in the complement of the tree defining $\tilde\Phi$, so ${\mathcal{A}}^{\rm phys}_{\Sigma\setminus\tilde\Phi|\tilde\Phi}\subsetneq\mathcal{A}_E$. Given their independence for different trees, ${\mathcal{A}}^{\rm phys}_{\Sigma\setminus\tilde\Phi|\tilde\Phi}$ is subject to subsystem relativity too, cf.~Fig.~\ref{fig:algebras&sub_rel}. For non-Abelian theories, intrinsic algebras possess a nontrivial center because the $H$-twirl introduces superselection. The center is given by its $\tilde\Phi$-dressed Casimir; it is the $\tilde\Phi$-dressing of the square of the \emph{total electric field tangential to  $\p\Sigma\setminus\tilde\Phi$ and interior transversal to $\p\Sigma$}. In the Abelian case, $H$ acts trivially, $\left(\mathcal{A}_{\mathring{\Sigma}}\otimes\mathcal{A}_{\p\Sigma\setminus\tilde\Phi}\right)^H=\mathcal{A}_{\mathring{\Sigma}}\otimes\mathcal{A}_{\p\Sigma\setminus\tilde\Phi}$, and the Casimir vanishes, so that the intrinsic algebra is a Type I factor. We show in App.~\ref{app:GM} that intrinsic algebras have a natural representation on the factor $\H_{\Sigma\setminus\tilde\Phi}^{\rm int}$ in Eq.~\eqref{eq:Hphysref}.

Lastly, the intersection of all intrinsic algebras,
\begin{equation}\label{eq:magneticcenterdef}
    \mathcal{A}_M:=\bigcap_{\text{{intrinsic QRFs }}\tilde\Phi}\mathcal{A}^{\rm phys}_{\Sigma\setminus\tilde\Phi|\tilde\Phi}\,\,\underset{\text{fin.\ Abel.}}{=}\hat{G}_{\p\Sigma}\l\mathcal{A}_{\Sigma\setminus\tilde\Phi|\tilde\Phi}^{\rm phys}\r\,,
\end{equation}
defines a \emph{magnetic center algebra}; its center consists of all independent Wilson loops in the corner $\p\Sigma$ (magnetic part) and of certain corner Wilson loop dressings of the \emph{total electric flux through} $\p\Sigma$ (electric part), see below. $\mathcal{A}_M$ is generated by all regional Wilson loops, regional Wilson line dressed pairs of electric fields and the center. 
For Abelian theories, the electric part vanishes since the total electric flux is zero, recovering exactly the magnetic center of \cite{Casini:2013rba}; $\mathcal{A}_M$ then constitutes the lattice counterpart to intrinsic phase spaces in the continuum, which lack boundary electric fields \cite{Araujo-Regado:2024dpr}. For non-Abelian theories, we show explicitly in App.~\ref{app:intersection_int_frame} that for one-dimensional corners $\p\Sigma$ the electric central element of $\mathcal{A}_M$ is given by the total electric flux through the corner, projected onto the unique Wilson loop group element via the Lie algebra inner product:
\begin{equation}\label{eq:elcentral}
Z_{1d}=\sum_{v\in\p\Sigma}\Tr_\rho\l \hat{g}_{W(v)} T_\rho^a\r\sum_{e\in\mathcal{E}_v\cap\mathring{\Sigma}}\l\mathfrak{G}^v_e\r^a\,.
\end{equation}
Here, $\hat{g}_{W(v)}$ denotes the Wilson loop holonomy from $v$ to itself before tracing. Together with the magnetic part, this reproduces exactly the generalized magnetic center of \cite{Delcamp:2016eya}, which too was formulated for one-dimensional corners (using fusion bases), however, unlike us, restricted to finite groups. 

Our magnetic center definition $\mathcal{A}_M$ thus constitutes a generalization of previous proposals, encompassing continuous groups and lattices of any dimension. In particular, we conjecture that $\mathcal{A}_M$ will also include an electric part in higher dimensions. A natural candidate for central elements in higher dimensions are objects as in Eq.~\eqref{eq:elcentral}  with $W$ a  \emph{Hamiltonian cycle} Wilson loop in the corner. These are Wilson loops that visit every vertex $v\in\p\Sigma$ once and only once. Indeed, such Hamiltonian cycle loops exist in planar graphs, and a spherical entangling surface in a lattice theory is such a planar graph. However, we leave investigating this an open question.
In summary, the magnetic center algebra is exactly what all intrinsic QRFs agree on. 

In App.~\ref{app_magneticcenter}, we show that, in the finite Abelian case, $\mathcal{A}_M$ can also be obtained by twirling any intrinsic algebra over a corner loop generated (Pontryagin) dual magnetic corner group $\hat{G}_{\p\Sigma}$, yielding the last equality in Eq.~\eqref{eq:magneticcenterdef}; in analogy to the electric case, this can be interpreted as averaging over the orientations of a dual magnetic corner QRF. For continuous $G$, this fails and accompanies an infinite ambiguity in defining a trace (and entropies) on $\mathcal{A}_M$. For non-Abelian $G$, there are no canonical dual groups. 

Clearly, we have the regional hierarchy (for any $\Phi,\tilde\Phi$):
\begin{equation}\label{eq:alghierarchy}
    \mathcal{A}^{\rm phys}_{\Sigma|\Phi}\,\,\,\supsetneq\,\,\,\mathcal{A}_E\,\,\,\supsetneq\,\,\,\mathcal{A}^{\rm phys}_{\Sigma\setminus\tilde\Phi|\tilde\Phi}\,\,\,\supsetneq\,\,\,\mathcal{A}_M\,.
\end{equation}

The electric/magnetic center algebra is to extrinsic/intrinsic QRFs what Lorentz scalars are to Lorentz observers; they correspond precisely to the observables that all frames agree on and that are invariant under all corresponding QRF transformations, cf.~\cite{Hoehn:2023ehz}. Restricting to these algebras\,---\,as the previous literature on entropies in gauge theories\,---\,misses out on all the QRF-covariant relational information.\\

\noindent\begin{minipage}{0.48\textwidth}
\centering
    \vspace{5pt}
    \captionsetup{hypcap=false} 
\begin{tabular}{ ccc }
\hline\hline
algebra & entanglement entropy & distillable\\\hline
\footnotesize{$\mathcal{A}_{\Sigma|\Phi}^{\rm phys}$}& \footnotesize{$S_{\rm vN}(\rho_{\Sigma|\Phi}^{\rm phys})$}& $\checkmark$\vspace{3pt}\\
\footnotesize{$\mathcal{A}_E$} & \footnotesize{$H(\{p_\mathfrak{e}\})+\sum_\mathfrak{e} p_\mathfrak{e}\l\log d_\mathfrak{e}+S_{\rm vN}(\rho_\mathfrak{e})\r$}& $\cross$ \vspace{3pt}\\ 
\multirow{2}{*}{\footnotesize{$\mathcal{A}^{\rm phys}_{\Sigma\setminus\tilde\Phi|\tilde\Phi}$}} & \footnotesize{A: $\,\,\,S_{\rm vN}(\rho^{\rm phys}_{\Sigma\setminus\tilde\Phi|\tilde\Phi})$} & $\checkmark$ \\
 & \footnotesize{nA: $\,\,\,H(\{p_{\hat{\mathfrak{e}}}\})+\sum_{\hat{\mathfrak{e}}} p_{\hat{\mathfrak{e}}}\l\log d_{\hat{\mathfrak{e}}}+S_{\rm vN}(\rho_{\hat{\mathfrak{e}}})\r$} & $\cross$ \vspace{3pt}\\
 \footnotesize{$\mathcal{A}_M$} & \footnotesize{$H(\{p^{m}_{\bar{\mathfrak{e}}}\})+\sum_{m,\bar{\mathfrak{e}}}p^{m}_{\bar{\mathfrak{e}}}\l\log d^{m}_{\bar{\mathfrak{e}}}+S_{\rm vN}(\rho^{m}_{\bar{\mathfrak{e}}})\r$}& $\cross$\\
 \hline
\hline
\end{tabular}
\captionof{table}{\small  
 Entanglement entropies for the algebras in Eq.~\eqref{eq:alghierarchy}. 
A/nA stands for Abelian/non-Abelian theories. 
$H(\{p_i\})=-\sum_i p_i\log p_i$ is the Shannon entropy, $\mathfrak{e},\hat{\mathfrak{e}},\bar{\mathfrak{e}}$ are electric charge superselection sector labels for the respective centers, 
while $m$ is a magnetic charge label and the $d_i$ measure degeneracy of these labels. For continuous $G$, entropies for $\mathcal{A}_M$ are highly ambiguous.} \label{tab:1}
\end{minipage}

\section{A hierarchy of regional entropies}
\label{sec:entro}

The entanglement entropies for the algebras in Eq.~\eqref{eq:alghierarchy} are summarized in Table~\ref{tab:1} (see Apps.~\ref{app_electriccenter}--\ref{app:entropy_proof} for details). Crucially, the extrinsic (resp.~intrinsic) relational algebras are always (resp.~for Abelian theories) \emph{factors} and so yield simple relational subsystem definitions and fully distillable entanglement as measured by the standard von Neumann entropy. In all other cases, there is a nontrivial center, resulting in non-distillable contributions to the entanglement entropy, determined by the classical probability distribution over the different electric and magnetic charge superselection sectors (see \cite{VanAcoleyen:2015ccp,Soni:2015yga} for a discussion of distillation in gauge theories). 

 Note that the extrinsic and intrinsic algebras yield \emph{relational} entanglement entropies, measuring the regional entropy ``as seen'' by the corresponding QRF; as such, they are QRF-dependent as observed in other contexts in \cite{Hoehn:2023ehz,DeVuyst:2024pop,DeVuyst:2024uvd,Cepollaro:2024rss}. In particular, this is the lattice gauge theory version of the regional entanglement entropies in perturbative gravity \cite{Chandrasekaran:2022cip,Jensen:2023yxy,Kudler-Flam:2023qfl} that were shown to be relational and QRF-dependent in precisely the same way \cite{DeVuyst:2024pop,DeVuyst:2024uvd}. For example, due to subsystem relativity, different extrinsic frames yield \emph{different} subregion factor algebras; for each, there will exist global physical states such that the regional state is pure\,---\,yielding vanishing entropy\,---\,relative to them, but not other QRFs.

The entropy for the electric center found in \cite{Soni:2015yga} (see also \cite{Bianchi:2024aim}) differs from that computed in extended Hilbert space constructions \cite{Ghosh:2015iwa,VanAcoleyen:2015ccp,Delcamp:2016eya}, missing the $d_\mathfrak{e}$-term. Remarkably, our novel electric center computation agrees with the extended Hilbert space results, although remaining at the algebraic level, as \cite{Soni:2015yga}, but by exploiting properties of the $\mathbb{G}_{\p\Sigma}$-twirl in Eq.~\eqref{eq:electriccenterdef}, and thus not invoking any auxiliary extended Hilbert spaces. 

Our magnetic center algebra generalizes previous proposals \cite{Casini:2013rba,Delcamp:2016eya} to encompass continuous structure groups $G$. Interestingly, as we explain in Apps.~\ref{app_magneticambiguity},~\ref{app_nonAbmagnetic}, precisely in that case there is no natural trace on $\mathcal{A}_M$. Like any Type I algebra with center, it has an infinite ambiguity in defining a trace on it, since the trace in each superselection sector can be weighed differently \cite{Takesaki1,Sorce:2023fdx}. To break this ambiguity, typically one would want to inherit a unique trace on $\mathcal{A}_M$ from some larger algebra within which it is contained. This works in the electric center construction, where the standard Hilbert space trace induces a trace on $\mathcal{A}_E$. However, it fails for the magnetic center when $G$ is continuous, in which case $\mathcal{A}_M$ contains \emph{no} operator that is trace-class w.r.t.\ the standard Hilbert space trace. One can still define traces on the algebra \cite{Takesaki1}, but subject to an infinite ambiguity and no obvious physical way to break it, thus rendering entropy constructions in this case physically arbitrary. For finite $G$, on the other hand, magnetic center entropies are well-defined, reproducing \cite{Casini:2013rba,Delcamp:2016eya}. In Table~\ref{tab:1}, we have therefore only included the finite group case.

Invoking the monotonicity of \emph{relative} entropy \cite{Uhlmann} (see also \cite{perez2023uhlmann,Witten:2018zxz}), the algebra hierarchy in Eq.~\eqref{eq:alghierarchy} immediately entails a corresponding hierarchy of regional relative entropies. (All of them are unital $*$-algebras, and embedding one into any of the larger ones is a unital $*$-homomorphism, so a completely positive trace-preserving map.) Let $\rho_{\rm phys},\sigma_{\rm phys}\in\mathcal{S}(\Hphys)$ be two global physical states and $\rho_i,\sigma_i$ be their reduction to $\mathcal{A}_i$ defined via $\Tr_i\l\rho_i a_i\r=\Tr_{\rm phys}\l\rho_{\rm phys}a_i\r$ for all $a_i\in\mathcal{A}_i$, $i=|\Phi,E,|\tilde\Phi,M$. Then:
\begin{equation}
    S(\rho_{|\Phi}||\sigma_{|\Phi})\geq S(\rho_E||\sigma_E)\geq S(\rho_{|\tilde\Phi}||\sigma_{|\tilde\Phi})\geq S(\rho_M||\sigma_M) \,.
\end{equation}
Clearly, removing observables from regional algebras cannot increase the distinguishability of states, and extrinsic QRFs have the greatest distinguishing capacity.

Perhaps more interestingly, there are also inequalities for \emph{entanglement} entropies between the relational and center algebras, following directly from $\rho_E=\mathbb{G}_{\p\Sigma}(\rho_{|\Phi})$ and, for finite Abelian $G$, $\rho_M=\hat{G}_{\p\Sigma}(\rho_{|\tilde\Phi})$. These incoherent corner twirls decohere the relational states across symmetry-related sectors, turning them into superselection sectors. In the basis diagonalizing the corner charges, this amounts to erasing the off-diagonal blocks. $G$-twirling maps a finer observable algebra to a coarser one, and the von Neumann entropy increases accordingly.
For any global state $\rho_{\rm phys}$, we have (see App.~\ref{app:entropy_proof})
\begin{equation}
    S_{\rm vN}\l\rho_{|\Phi}\r\leq S_{\rm vN}\l\rho_E\r\,,\quad S_{\rm vN}(\rho_{|\tilde\Phi})\leq S_{\rm vN}\l\rho_M\r
\end{equation}
(the latter inequality holds only for finite Abelian $G$). 
 
 Note that there is no such inequality between $\mathcal{A}_E$ and $\mathcal{A}^{\rm phys}_{\Sigma\setminus\tilde\Phi|\tilde\Phi}$. For example, consider $G$ Abelian, so that Eq.~\eqref{eq:Hphysref} factorizes and $\mathcal{A}^{\rm phys}_{\Sigma\setminus\tilde\Phi|\tilde\Phi}$ is a factor. Then there exist global pure states, with $\rho_{|\tilde\Phi}$ pure and superposition in the electric corner charges in the GM sector, which due to twirling, yields a mixed $\rho_E$. Thus, the electric center entropy does not lower bound the intrinsic one, although the intrinsic algebra is a subalgebra of the electric center algebra.

This entropy structure is intimately intertwined with the behavior of global charges, defined on $\cB$ as the sum of incoming electric fields at each boundary node: $Q_v^a := \sum_{e\in\mathcal{E}_v} (\mathfrak{G}^v_e)^a$, with $\mathfrak{G}^v_e$ the appropriate assignment of $L_e,R_e$. These charges commute with the physical Hamiltonian (constructed from Wilson loops and electric fields) and are preserved under time evolution. As generators of large gauge transformations, they also generate extrinsic frame reorientations. 

Whenever the global state commutes with the global charges $Q_v^a$, the reduced state in any subregion necessarily is reorientation invariant (cf.~App.~\ref{app:global_decoherence}) and so commutes with the induced electric corner charge operators generating $\mathbb{G}_{\p\Sigma}$ (the regional data cannot distinguish whether the reorientation is generated by large gauge or other transformations in the complement). This introduces three qualitatively distinct regimes. If the global state contains coherent superpositions across charge sectors, then it carries off-diagonal correlations between distinct superselection blocks of $\mathcal{A}_E$. These coherences are invisible to electric center observables, but detectable by extrinsic relational observables. In this case, the entropy calculated from $\mathcal{A}_E$ underestimates the information captured by $\mathcal{A}^{\rm phys}_{\Sigma|\Phi}$, especially the cross-boundary relational data. If, instead, the state is an incoherent mixture over charge sectors, then both algebras yield the same entropy, since the coherence is not there from the outset. Finally, if the state lies within a single global charge sector, there is no superselection ambiguity (though the extrinsic relational state may be an incoherent mixture over charge sectors), and, again, the entropies agree.

\section{Conclusions}
\label{sec:concl}

We have introduced QRFs into lattice gauge theories by translating their classical continuum counterparts \cite{Carrozza:2021gju,Araujo-Regado:2024dpr} into the lattice quantum theory using the perspective-neutral framework \cite{delaHamette:2021oex}. This permitted us to develop a relational approach to entanglement entropies that we showed encompasses and unifies previous nonrelational approaches via an algebra and entropy hierarchy, which captures the complexity of regional symmetry structures. Importantly, unlike previous extended Hilbert space constructions \cite{Donnelly:2011hn,Buividovich:2008gq,Ghosh:2015iwa,VanAcoleyen:2015ccp,Delcamp:2016eya,Donnelly:2014gva,Aoki:2015bsa,Hung:2015fla}, our treatment does \emph{not} embed $\Hphys$ into any auxiliary Hilbert space, only invoking its intrinsic properties and observable algebras thereon. Our subsystem partitions are defined gauge-invariantly and relationally \cite{Hoehn:2023ehz,AliAhmad:2021adn}. The key to our observations is the distinction, as in \cite{Araujo-Regado:2024dpr}, between extrinsic and intrinsic edge mode QRFs for a subregion. Relative to any extrinsic QRF, the physical Hilbert space \emph{always} factorizes between subregion and complement, and the same is true for intrinsic QRFs in Abelian theories. In contrast to nonrelational constructions, this yields fully distillable relational entanglement entropies. This demonstrates a further advantage of relational entropy definitions that were previously shown to lead to an intrinsic regularization (type transition) of regional gravitational von Neumann algebras \cite{Chandrasekaran:2022cip,Jensen:2023yxy,Kudler-Flam:2023qfl,Kirklin:2024gyl,DeVuyst:2024pop,DeVuyst:2024uvd,Fewster:2024pur,AliAhmad:2024wja}.
Finally, using intrinsic QRFs, we also generalized previous proposals for a magnetic center algebra \cite{Casini:2013rba,Delcamp:2016eya}, revealing that these permit only physically rather arbitrary entropy computations when $G$ is continuous.



\begin{center}
    \textbf{Acknowledgments}
\end{center}
\noindent 
We thank Bianca Dittrich and Aldo Riello for many helpful discussions on their magnetic center construction in \cite{Delcamp:2016eya}, Bilyana Tomova for initial collaboration, and Jesse Woods for discussions on frames in non-Abelian theories \cite{jesse}. This work was supported by funding from Okinawa Institute of Science and Technology Graduate University and also made possible through the support of the ID\# 62312 grant from the John Templeton Foundation, as part of the \href{https://www.templeton.org/grant/the-quantum-information-structure-of-spacetime-qiss-second-phase}{\textit{`The Quantum Information Structure of Spacetime'} Project (QISS)}.~The opinions expressed in this project are those of the authors and do not necessarily reflect the views of the John Templeton Foundation. 

\onecolumngrid
\appendix

\vspace{.7cm}
\section{Lattice Gauge Theory}
\label{app:kin}

In this appendix, we review some fundamental formulas of lattice gauge theory to establish the notation used in the main text and to elaborate on specific results referenced therein.

In this work, we took the convention that, for each oriented edge on the lattice, carrying an Hilbert space $\H_e=L^2(G)$, the commuting left and right unitary actions of $G$ are $U(h)\ket{g}=\ket{hg}$ and $V(h)\ket{g}=\ket{gh^{-1}}$. We recall that the link variable $\hat{g}$ is a quantum operator valued in a representation $\rho$ of $G$. From this, we have the actions:
\begin{align}
    U(h)\hat{g}U^\dagger(h)= \rho(h^{-1})\hat{g}\quad \text{and} \quad V(h)\hat{g}V^\dagger(h)=\hat{g}\rho(h)\,.
\end{align}
So $\hat{g}$ itself transforms in the representation space $\mathbf{V}_\rho^L\otimes\l\mathbf{V}_\rho^R\r^*$ under the action of gauge transformations. These act independently at each of the bulk vertices $v\in\mathcal{V}_{\rm b}:=\mathcal{V}\setminus\left(\mathcal{V}\cap\mathcal{B}\right)$, in our convention as \eqref{eq:gaugetransfconv} (opposite of \cite{VanAcoleyen:2015ccp}):
\be
U_v(h)=\bigotimes_{e\in\mathcal{E}_v^{\rm in}}U_e(h)\bigotimes_{e'\in\mathcal{E}_v^{\rm out}}V_e(h)\,, 
\ee
where $\mathcal{E}_v^{\rm in/out}$ denote the sets of in- and outgoing edges at $v$. Bulk link variables thus transform as 
\be
\hat{g}_e\mapsto U_{v_i}(h_i)U_{v_f}(h_f)\hat{g}_e U_{v_f}^\dag(h_f)U^\dag_{v_i}(h_i)=V_e(h_i)U_e(h_f)\hat{g}_e U_e^\dag(h_f)V_e^\dag(h_i)=\rho(h^{-1}_f)\hat{g}_e\rho(h_i)\,.
\ee

Defining the generators of the left and right actions as
\begin{align}
    U(h)=e^{-i\theta^a \hat{L}^a} \quad \text{and} \quad V(h)=e^{+i\theta^a \hat{R}^a}\,,\quad \text{for}\quad h=e^{i\theta^a T^a}\,.
\end{align}
We remark that $\hat{L}^a$ and $\hat{R}^a$ are both quantum operators but do not live in the Lie algebra $\mathfrak{g}$.
Expanding to linear order, we get the Hilbert space commutators
\begin{align}
\label{eq:RL_G_commutators}
    [\hat{L}^a, \hat{g}]=T^a_\rho \hat{g}\quad \text{and}\quad [\hat{R}^a,\hat{g}]=\hat{g}T^a_\rho\,,
\end{align}
where $T^a_\rho$ are the generators of $\mathfrak{g}$ in representation $\rho$ and matrix indices are implicitly contracted. The Hilbert space brackets for the two sets of generators are defined by
\begin{align}\label{commutators}
    [\hat{L}^a,\hat{L}^b]=-i f^{abc}\hat{L}^c\,,\quad[\hat{R}^a,\hat{R}^b]=i f^{abc}\hat{R}^c\,,\quad [\hat{L}^a,\hat{R}^b]=0\,.
\end{align}
where $f_{abc}$ are the structure constants of $\mathfrak{g}$.
Under group multiplication, the generators transform in the adjoint representation
\begin{align}
    U(h)\hat{L}^aU(h)^\dagger=D(h^{-1})^{ab}\hat{L}^b\quad \text{and}\quad V(h)\hat{R}^aV(h)^\dagger=D(h^{-1})^{ab}\hat{R}^b\,, 
\end{align}
where $D$ denotes the adjoint representation of $G$. 
The generators transform under independent actions of $G_L$ and $G_R$, respectively. If we have an oriented edge, with associated Hilbert space, we can obtain the same equivalent description by ‘‘inverting the arrow" mapping:
$\hat g \to \hat g^\dagger$ and $\hat L^i \leftrightarrow \hat R^i$.

Finally, we introduce the following notation for the generators of gauge transformation at node $v$, such that $U_v(h)=e^{-i\theta^a \mathfrak{G}_v^{\,a}}$,  for $h=e^{i\theta^a T^a}$. From \eqref{eq:gaugetransfconv} we get 
\be
\label{eq:Gauss_Law_generators}
\mathfrak{G}_v^{\,a} = \sum_{e\in\mathcal{E}_v}(\mathfrak{G}_e^v)^a\,,
\q \wth \q
(\mathfrak{G}_e^v)^a =\left \{ 
\begin{array}{rl}
  L_e^a   & \text{if}\;\;e\in\mathcal{E}_v^{\rm in} \\[5pt]
  -R_e^a   & \text{if}\;\;e\in\mathcal{E}_v^{\rm out}
\end{array}
\right.
\ee
These are the operators entering the Gauss law $\mathfrak{G}_v^{\,a}=0$, $\forall v\in \cV_b$.

A central role is played by the Peter–Weyl theorem, which states that for any compact group $G$:
\begin{align}
    L^2(G)\cong \bigoplus_{[r]}\mathcal{H}_r\otimes\mathcal{H}_{r^*}\,,
\end{align}
where $[r]$ labels isomorphism classes of irreducible representations of $G$ and $r^*$ denotes the dual of $r$. This means that there is a natural (dual) basis for $L^2(G)$ respecting this decomposition
\begin{align}
    L^2(G) = \text{span}_\mathbb{C}\{\ket{r,i,j}|\, r\in\{[r]\},\,i,j=1,...,d_r\}\,,\quad \text{where}\quad d_r=\text{dim}(\mathcal{H}_r)\,.
\end{align}
with the relation to the group basis given by
\begin{align}
    \ket{r,i,j} = \sqrt{\frac{d_r}{{\rm Vol}(G)}}\int_G dg \, (\rho^r(g))^i_j \ket{g}\,,
\end{align}
where $\rho^r$ is  is the matrix representative of $g$ in irrep $r$. The action of the generators on the dual basis follows from
\begin{align}
    U(h)\ket{r,i,j}=(\rho^r(h))^i_k\ket{r,k,j}\quad \text{and}\quad V(h)\ket{r,i,j}=(\rho^{r^*}(h))^j_k\ket{r,i,k}\,,
\end{align}
by expanding to linear order
\begin{align}
    \hat{L}^a\ket{r,i,j}=-(T^a_r)^i_k\ket{r,k,j}\quad \text{and}\quad \hat{R}^a\ket{r,i,j}=-(T^a_r)^k_j\ket{r,i,k}\,,
\end{align}
where we used the fact that $(T^a_{r^*})^i_j=-(T^a_r)^j_i$ for any $(r,r^*)$ dual pair. 

This allows us to show that the left and right electric fields are not independent, but are related by the link operator, which carries a color charge. To see this, define the link operator in the adjoint representation
\begin{align}
    \hat{g}_\text{Adj}:=\int [dg] D(g) \ket{g}\bra{g}\,.
\end{align}
We now consider the action of the composite operator $(\hat{g}_\text{Adj})^{ab}\hat{L}^b$ on the dual basis $\ket{r,i,j}$ of $L^2(G)$
\begin{align}
    (\hat{g}_\text{Adj})^{ab}\hat{L}^b\ket{r,i,j}&=-(\hat{g}_\text{Adj})^{ab}(T^b_r)^i_k\ket{r,k,j}\\
    &=-(\hat{g}_\text{Adj})^{ab}(T^b_r)^i_k \sqrt{\frac{d_r}{{\rm Vol}(G)}}\int_G dg \, (\rho^r(g))^k_j \ket{g}\\
    &=-\sqrt{\frac{d_r}{{\rm Vol}(G)}}\int_G dg\,(\rho^r(g))^k_j \underbrace{(T^b_r)^i_k D(g)^{ab}}_{\l\rho^r(g)T^a_r\rho^r(g^{-1})\r^i_k}\ket{g}\\
    &=-\sqrt{\frac{d_r}{{\rm Vol}(G)}}\int_G dg\,\l\rho^r(g)T^a_r\r^i_j\ket{g}\\
    &=-(T^a_r)^k_j\sqrt{\frac{d_r}{{\rm Vol}(G)}}\int_G dg\,(\rho^r(g))^i_k\ket{g}\\
    &=-(T^a_r)^k_j\ket{r,i,k} = \hat{R}^a\ket{r,i,j}\,\quad\forall r,i,j\,.
\end{align}
In going from the third to the fourth line, we made use of the adjoint action of a group on its Lie algebra. In conclusion, we have
\begin{align}\label{eq:l=r}
    \hat{R}^a=(\hat{g}_\text{Adj})^{ab}\hat{L}^b \overset{\dagger}{\iff} \hat{R}^a = \hat{L}^b (\hat{g}_\text{Adj})^{ab}\,.
\end{align}

The left and right Casimir quantum operators are defined as
\begin{align}
    \hat{C}_L:=\hat{L}^a\hat{L}^a \quad \text{and} \quad \hat{C}_R:=\hat{R}^a\hat{R}^a\,,
\end{align}
while the Lie algebra Casimir is given by $C_\rho:=T^a_\rho T^a_\rho$ in the $\rho$ representation. From \eqref{eq:l=r} it immediately follows that $\hat{C}_L=\hat{C}_R$.
We have the following Hilbert space brackets
\begin{align}
    [\hat{C}_L,\hat{g}]=\hat{L}^aT^a_\rho\hat{g}+T^a_\rho\hat{g}\hat{L}^a\quad \text{and} \quad [\hat{C}_R,\hat{g}]=\hat{R}^a\hat{g}T^a_\rho+\hat{g}T^a_\rho\hat{R}^a
\end{align}

The link color charge operator is given by 
\begin{align}
    \hat{E}^a := \hat{L}^a-\hat{R}^a\,.
\end{align}

\section{Kinematical refactorization into a QRF Wilson line basis}\label{app_kinrefactor}

We briefly discuss how to compose the electric fields of two consecutive edges in the unitary refactorization, fundamental in defining Wilson lines frame (see Fig. \ref{fig:Wilson_line_reparam}). For this, 
consider two consecutive edges, oriented in the same direction. Following the arrows, we have e.g. $g_1$ as the first one and then $g_2$. We define the cumulative variables as $g_{v_1} =g_1$, $g_{v_2}=g_2 g_1$,  corresponding to path-ordered products (the first operator applied is $g_1$). For a sequence of edges stopping at node $n$ we would similarily define $g_{v_n} = g_n \dots g_1$, so that $g_n = g_{v_n} g_{v_{n-1}}^{-1}$ and $g_{v_1} = g_1$. We introduce the unitary map: 
\be
\mathfrak{U}: \H_1 \otimes \H_2 \to \H_{v_1} \otimes \H_{v_2}\,,\q\q \ket{g_1}_1 \ket{g_2}_2 \mapsto  \ket{g_1}_{v_1} \ket{g_2 g_1}_{v_2}\,,
\ee
which induces the transformation of wavefunctions $\psi_{12} (g_1, g_2)\mapsto\psi_{v_1v_2} (g_1, g_2 g_1)$ and of group operators
\be
 \mathfrak{U}^\dagger \hat g_{v_1} \mathfrak{U} = \hat g_1\,,\q\q \mathfrak{U}^\dagger \hat g_{v_2} \mathfrak{U} = \hat g_2 \hat g_1\,.
\ee
More explicitly,
\be
 \mathfrak{U} (\Id_{v_1} \otimes \hat g_{v_2})_{ij} \mathfrak{U}^\dagger =  (\hat g_1)_{kj} \otimes (\hat g_2)_{ik} 
\ee

We now examine how the exponentiated electric field operators (left and right actions) transform under 
$\mathfrak{U}$.
  \begin{itemize}
\item Right action on edge 1:
\be
 \mathfrak{U}\, V_1(h) \mathfrak{U}^\dagger \ket{g_{v_1}, g_{v_2}}  &= \mathfrak{U}\, V_1(h) \ket{g_{v_1}}_1 \ket{ g_{v_2} g_{v_1}^{-1}}_2 \notag\\
 &= \mathfrak{U}\, \ket{g_{v_1} h^{-1}}_1 \ket{g_{v_2} g_{v_1}^{-1}}_2 = \ket{ g_{v_1} h^{-1} , g_{v_2} h^{-1}} = V_{v_1}(h) \otimes V_{v_2}(h) \ket{ g_{v_1}, g_{v_2}}
\ee
\item Right action on edge 2:
\be
 \mathfrak{U}\, V_2(h) \mathfrak{U}^\dagger \ket{g_{v_1}, g_{v_2}}  &= \mathfrak{U}\, V_2(h) \ket{g_{v_1}}_1 \ket{ g_{v_2} g_{v_1}^{-1}}_2 \notag\\
 &= \mathfrak{U}\, \ket{g_{v_1}}_1 \ket{ g_{v_2} g_{v_1}^{-1} h^{-1}}_2 = \ket{g_{v_1}, g_{v_2} g_{v_1}^{-1} h^{-1} g_{v_1}}
\ee
This transformation is not a simple tensor product of actions on the $\H_{v_i}$'s.
\item Left action on edge 1:
\be
 \mathfrak{U}\, U_1(h) \mathfrak{U}^\dagger \ket{g_{v_1}, g_{v_2}}  &= \mathfrak{U}\, U_1(h) \ket{g_{v_1}}_1 \ket{g_{W} g_{v_1}^{-1} }_2 \notag\\
 &= \mathfrak{U}\, \ket{ h g_{v_1} }_1 \ket{ g_{v_2} g_{v_1}^{-1}}_2 = \ket{h g_{v_1}, g_{v_2}  g_{v_1}^{-1} h g_{v_1}}
\ee
Again, this is not a factorized action on the $\H_{v_i}$'s.
\item Left action on edge 2:
\be
 \mathfrak{U}\, U_2(h) \mathfrak{U}^\dagger \ket{g_{v_1}, g_{v_2}}  &= \mathfrak{U}\, U_2(h) \ket{g_{v_1}}_1 \ket{g_{v_2} g_{v_1}^{-1}}_2 \notag\\
 &= \mathfrak{U}\, \ket{ g_{v_1}}_1 \ket{h g_{W} g_{v_1}^{-1} }_2 = \ket{g_{v_1}, h g_{v_2}}  = U_{v_2}(h) \ket{g_{v_1}, g_{v_2}}
\ee
\end{itemize}
From the above, only $V_1$ and $U_2$ yield simple, factorized expressions under $\mathfrak{U}$. In particular the first one corresponds to a large gauge transformation, acting on all terms. The last one is a gauge transformations at the endpoint $v_2$ (eventually belonging to $\p \Sigma$). Under this small gauge transformation, the Wilson line is nicely covariant. Finally, at the intermediate node $v_1$, the relevant gauge transformation corresponds to the combined action (cf.~Eq.~\eqref{eq:gaugetransfconv} $V_2(h)\otimes U_1(h)$, which under $\mathfrak{U}$ becomes:
\be
\mathfrak{U}\,\left( V_2(h)\otimes U_1(h) \right)\mathfrak{U}^\dagger = U_{v_1}(h)\otimes \Id_{v_2}
\ee
We propose using the combinations $R_1$, $L_2$ and $L_1 - R_2$  as the effective electric field operators to be dressed. These have well-defined, local transformations under the kinematical unitary map and align with natural interpretations as generators of local gauge transformations, for the left vector fields acting on the Hilbert spaces $H_{v_i}$, and large gauge transformation (frame reorientations) for the joint action of right multiplications on all $v_i$'s factors. By contrast, attempting to study $L_1$ and $R_2$ individually would unnecessarily complicate the analysis due to their non-factorizing action.


\section{Relational Observables}
\label{app:relational_obs}

In this appendix, we discuss how to recover the gauge-invariant operators in non-Abelian lattice gauge theory from relational observables obtained through dressing to a frame.

As a remark, we stress that in our constructions we often invoke the fact that gauge transformations act on all bulk nodes independently. Thus, for any operator $A$
\begin{equation}
    \mathcal{G}(A)=\mathcal{G}_{\rm supp(A)}(A) =\frac{1}{\rm{Vol}(G^{N_A})}\int_{G^{N_A}}\dd\mathbf{g}_A U(\mathbf{g}_A) A U^\dag(\mathbf{g}_A)\,,
\end{equation}
where $N_A$ is the number of vertices on which $A$ has support. This is because the full G-twirl just reduces to the identity on all other tensor factors, reducing appropriately the overall volume pre-factors coming from twirling over the other vertices.  Similarly, for the physical projector $
    \Pi_{\rm phys}=\prod_{v\in\mathcal{V}_{\rm b}}\Pi_v
$,  with 
$
    \Pi_v:=\frac{1}{\rm{Vol}(G)}\int_G\dd{g}U_v(g)\
$. From now on, we adopt the notation $\frac{1}{\text{Vol}(G)}\int_G dg =: \int [dg]$, with the obvious extension for integrals at multiple vertices.

We need to distinguish between two cases, which we analyze separately: dressing to an 1) extrinsic or 2) intrinsic frame. In either case, we will construct observables of the type
\begin{align}
    O_{f_\textbf{S}|\textbf{R}}(\textbf{g}_{\p\Sigma}):=\mathcal{G}_{\p\Sigma}(\ket{\textbf{g}_{\p\Sigma}}\bra{\textbf{g}_{\p\Sigma}}_\textbf{R}\otimes f_\textbf{S})\,.
\end{align}
Here, $f_\textbf{S}$ is any observable in $\mathcal{B}(\mathcal{H}_\textbf{S})$, where $\mathcal{H}_\text{kin}^{\p\Sigma}=\mathcal{H}_\Phi\otimes \mathcal{H}_{\Tilde{\Phi}}\otimes \mathcal{H}_\textbf{S}$ such that $\mathcal{H}_\textbf{S}=\mathcal{H}_\Sigma^\text{phys}\otimes \mathcal{H}_{\p\Sigma\setminus\Tilde{\Phi}}\otimes \mathcal{H}_{\Bar{\Sigma}\setminus\Phi}^\text{phys} $ and $\mathcal{H}_\textbf{R}$ can correspond to either $\mathcal{H}_\Phi$ or $\mathcal{H}_{\Tilde{\Phi}}$. By assumption, $f_\textbf{S}$ is already invariant with respect to bulk gauge transformations either inside or outside the region. This is why only a $\mathcal{G}_{\p\Sigma}$-twirl is needed. So examples of $f_\textbf{S}$ are: open Wilson lines anchored on $\p\Sigma$, open Wilson lines in the complement anchored on $\mathcal{B}$, Wilson loops fully supported inside or outside the region, as well as linear and quadratic combinations of electric fields that can be either internal normals or tangential to the boundary.

Typically, $f_\textbf{S}$ will have support on either one or two boundary vertices. Examples of the latter are any open Wilson lines with both endpoints anchored on $\p\Sigma$. Examples of the former are left and right electric fields on edges hitting the boundary from inside or tangential to it. This means that from the $\mathcal{G}_{\p\Sigma}$-twirl on the whole $\p\Sigma$, only the incoherent average with respect to one or two nodes has a non-trivial effect. 

We also note that if $f_\textbf{S}$ is already invariant with respect to boundary gauge transformations, then $\mathcal{G}_{\p\Sigma}(\ket{\textbf{g}_{\p\Sigma}}\bra{\textbf{g}_{\p\Sigma}}_\textbf{R}\otimes f_\textbf{S})=\mathbb{1}_\textbf{R
}\otimes f_\textbf{S}$. Examples of that are the Casimir operators on any link or Wilson loops obtained by closing inside/outside open Wilson lines with link variables in $\mathcal{H}_{\p\Sigma\setminus\Tilde{\Phi}}$. More generally, operators may behave differently under the twirl: some remain invariant, others acquire a nontrivial relational dressing, and in some cases, the twirl yields a trivial multiple of the identity (including zero).

Below, we will consider the following concrete examples, that we will dress to both intrinsic and extrinsic frames: A) open Wilson line operators with both ends anchored on $\p\Sigma$, and B) left/right electric fields on edges either hitting or belonging to $\p\Sigma$ (i.e. whose vertex belongs to $\p\Sigma$). For case A), we take the open Wilson line $l$ of interest to be oriented from vertex $w$ to vertex $v$, where $v,w\in\p\Sigma$. We denote its link operator by $\hat{g}_l$, which stands for a composite operator obtained from (tensor) multiplying the link variables of the edges making up $l$, in a path-ordered fashion. For example, if the link contains the set $\mathcal{E}_l=\{e|e\in l \}$ of edges whose orientations happen to align to our choice of $l$ orientation, then $\hat{g}_l=\bigotimes_{e\in\mathcal{E}_l}\hat{g}_e$, with identity on all other factors, where there is an implicit contraction of representation indices. For case B), the relevant electric fields depend on the orientation of the corresponding edge $e$ relative to the vertex $v$ on $\p\Sigma$. We will simply denote it by $(\hat{\mathfrak{G}}^v_e)^a$ (see App.~\ref{app:kin}).

Let us denote the gauge action at vertex $v$ by the unitary $\mathcal{U}_v$. This can either correspond to a left or right action, depending on whether the relevant line/edge hits the vertex with inward or outward orientation. That means that, under gauge transformations (labelled by the vector $\textbf{h}_{\p\Sigma}=(h_v|v\in\p\Sigma)$) we have the following transformations
\begin{align}
    \mathcal{U}_v(h_v)\mathcal{U}_w(h_w) \;\hat{g}_l \;\mathcal{U}_v^\dagger(h_v)\mathcal{U}_w^\dagger(h_w) &= \rho(h_v^{-1})\,\hat{g}_l\,\rho(h_w)\\
    \mathcal{U}_e(h_v) \;(\hat{\mathfrak{G}}^v_e)^a \;\mathcal{U}_e(h_v) &=D(h_v^{-1})^{ab}\,(\hat{\mathfrak{G}}^v_e)^b\,,
\end{align}
where $D$ denotes the adjoint representation of the group $G$.

\paragraph*{1) Dressing to Extrinsic Frame:}

We take the extrinsic frame to be in configuration $\ket{\textbf{g}_{\p\Sigma}}_\Phi=\bigotimes_{v\in\p\Sigma}\ket{g_v}_{\Phi_v}$. Because we have chosen it to be oriented towards $\p\Sigma$ we have a left action at all nodes
$U(\textbf{h}_{\p\Sigma})\ket{\textbf{g}_{\p\Sigma}}_\Phi  = \ket{\textbf{h}_{\p\Sigma}\textbf{g}_{\p\Sigma}}_\Phi$. Further, since the operators we want to dress are only (at most) supported on the boundary vertices $v$ and $w$, we have that 
\begin{align}
\mathcal{G}_{\p\Sigma}(\ket{\textbf{g}_{\p\Sigma}}\bra{\textbf{g}_{\p\Sigma}}_\Phi\otimes f_\textbf{S})= \int [dh_v][dh_w] \ket{h_vg_v}\bra{h_vg_v}_{\Phi_v}\otimes \ket{h_wg_w}\bra{h_wg_w}_{\Phi_w}\otimes \mathcal{U}_v(h_v)\mathcal{U}_w(h_w) \;f_\textbf{S} \;\mathcal{U}_v^\dagger(h_v)\mathcal{U}_w^\dagger(h_w)\,,
\end{align}
with identity understood on all other factors. We perform the changes of variables $h_v\mapsto h_vg_v^{-1}$ and $h_w\mapsto h_wg_w^{-1}$ to get
\begin{align}
    \int [dh_v][dh_w] \ket{h_v}\bra{h_v}_{\Phi_v}\otimes \ket{h_w}\bra{h_w}_{\Phi_w}\otimes \mathcal{U}_v(h_vg_v^{-1})\mathcal{U}_w(h_wg_w^{-1}) \;f_\textbf{S} \;\mathcal{U}_v^\dagger(h_vg_v^{-1})\mathcal{U}_w^\dagger(h_wg_w^{-1})\,.
\end{align}
 
\paragraph*{1A. Open Wilson line:}

We take $f_\textbf{S}=\hat{g}_l$. We then have
\begin{align}
\label{eq:Open_W_lines_ext_dress}
    \mathcal{G}_{\p\Sigma}(\ket{\textbf{g}_{\p\Sigma}}\bra{\textbf{g}_{\p\Sigma}}_\Phi\otimes \hat{g}_l)&=\int [dh_v][dh_w] \ket{h_v}\bra{h_v}_{\Phi_v}\otimes \ket{h_w}\bra{h_w}_{\Phi_w}\otimes \rho(g_vh_v^{-1})\,\hat{g}_l\,\rho(h_wg_w^{-1})\notag\\
    &=\rho(g_v)\l\int [dh_v]\,\rho(h_v^{-1})\ket{h_v}\bra{h_v}_{\Phi_v}\r\otimes\hat{g}_l\otimes\l\int [dh_v]\,\rho(h_w)\ket{h_w}\bra{h_w}_{\Phi_w}\r\rho(g_w^{-1})\\
    &=\rho(g_v)\;\hat{g}_{\Phi_v}^\dagger\otimes \hat{g}_l\otimes\hat{g}_{\Phi_w}\;\rho(g_w^{-1})\,.\notag
\end{align}
The RHS corresponds to an open Wilson line anchored at $\mathcal{B}$,  combining the extrinsic Wilson lines and the open line $l$ in a consistent path-ordered fashion.

\paragraph*{1B. Left/Right electric field:}

We take $f_\textbf{S}=(\hat{\mathfrak{G}}^v_e)^a$. We then have
\begin{align}
    \mathcal{G}_{\p\Sigma}(\ket{\textbf{g}_{\p\Sigma}}\bra{\textbf{g}_{\p\Sigma}}_\Phi\otimes (\hat{\mathfrak{G}}^v_e)^a)&=\int [dh_v][dh_w] \ket{h_v}\bra{h_v}_{\Phi_v}\otimes \ket{h_w}\bra{h_w}_{\Phi_w}\otimes D(g_vh_v^{-1})^{ab}\,(\hat{\mathfrak{G}}^v_e)^b\\
    &=D(g_v)^{ac}\l\int [dh_v] \,D(h_v^{-1})^{cb}\ket{h_v}\bra{h_v}_{\Phi_v}\r\otimes (\hat{\mathfrak{G}}^v_e)^b \otimes \mathbb{1}_{\Phi_w}
\end{align}
We notice that the operator in brackets is nothing but (the dagger of) the link operator $\hat{g}_{\Phi_v}$, only that this time in the adjoint representation, instead of the $\rho$ representation. We denote it by $\l\hat{g}_\text{Adj}\r_{\Phi_v}$. We thus have
\begin{align}
\label{eq:L_ext_dress}
    \mathcal{G}_{\p\Sigma}(\ket{\textbf{g}_{\p\Sigma}}\bra{\textbf{g}_{\p\Sigma}}_\Phi\otimes (\hat{\mathfrak{G}}^v_e)^a)=(D(g_v)\l\hat{g}_\text{Adj}\r_{\Phi_v}^\dagger)^{ab}\otimes (\hat{\mathfrak{G}}^v_e)^b \otimes \mathbb{1}_{\Phi_w}\,.
\end{align}

We see that the link left and right electric fields get dressed by the color charge of the configuration of the Wilson line connecting the node $v$ to the asymptotic boundary. The dressed observables correspond to the bulk edge electric fields parallel transported to the asymptotic boundary along the Wilson lines.

Consider now the special case in which the edge $e$ lies on $\p\Sigma$. Then we can dress the color charge operator $\hat{E}^a_e:=\hat{L}^a_e-\hat{R}^a_e$, by dressing the electric fields at both ends
\begin{align}
    \mathcal{G}_{\p\Sigma}(\ket{\textbf{g}_{\p\Sigma}}\bra{\textbf{g}_{\p\Sigma}}_\Phi\otimes \hat{E}_e)&=( D(g_v)\l\hat{g}_\text{Adj}\r_{\Phi_v}^\dagger-\mathbb{1}_{\Phi_v})\otimes\hat{E}_e\otimes(\mathbb{1}_{\Phi_w}-\l\hat{g}_\text{Adj}\r_{\Phi_w}D(g_w^{-1}))\,,
\end{align}
where we remember that $\hat{E}_e$ is a vector of electric fields. This observable measures the total color charge of the configuration of the joint open Wilson line system, combining the two extrinsic Wilson lines with the edge $e$.

\paragraph*{2) Dressing to Intrinsic Frame:}

This is similar in spirit to the previous situation. However, the incompleteness of the intrinsic frame will play a role here. In particular, as discussed in the main body and in appendices \ref{app:factor}, \ref{app:GM}, the intrinsic frame has an isotropy group given by a $\p\Sigma$-uniform action
\begin{align}
    H = \{(h,...,h)|h\in G\}\,.
\end{align}
We will label elements of $H$ by $\mathfrak{h}$. Any element $\textbf{g}_{\p\Sigma}\in \mathcal{G}_{\p\Sigma}$ can be written as $\textbf{g}_{\p\Sigma}=\mathfrak{g}_x \mathfrak{h}$, for some $\mathfrak{g}_x\in\mathcal{G}_{\p\Sigma} $ and $\mathfrak{h}\in H$. Then, we say element $\textbf{g}_{\p\Sigma}\in x$, where $x\in \mathcal{G}_{\p\Sigma}/H$ labels the left cosets of $H$.
This also means that states in $\mathcal{H}_{\Tilde{\Phi}}$ are labeled by the cosets $\{\ket{\mathfrak{g}_x}_{\Tilde{\Phi}}|x\in \mathcal{G}_{\p\Sigma}/H\}$, where we can take any representative of $x$. For $\cN$ the anchor point of the intrinsic frame (e.g.\ north pole of the boundary $\p \Sigma$), we will choose the representative $\mathfrak{g}_x=\textbf{g}_{\p\Sigma}(g_\mathcal{N},...,g_\mathcal{N})^{-1}=(e_\mathcal{N},\mathring{\mathfrak{g}}_x)\,,\, \mathring{\mathfrak{g}}_x\in \mathcal{G}_{\p\Sigma\setminus\mathcal{N}}$. So observables relational to the intrinsic frame read
\begin{align}
    \mathcal{G}_{\p\Sigma}(\ket{\textbf{g}_{\p\Sigma}}\bra{\textbf{g}_{\p\Sigma}}_{\Tilde{\Phi}}\otimes f_\textbf{S})&=\int_{\mathcal{G}_{\p\Sigma}}[d\textbf{h}_{\p\Sigma}]\ket{\textbf{h}_{\p\Sigma}\textbf{g}_{\p\Sigma}}\bra{\textbf{h}_{\p\Sigma}\textbf{g}_{\p\Sigma}}_{\Tilde{\Phi}}\otimes U_\textbf{S}(\textbf{h}_{\p\Sigma}) f_\textbf{S}U_\textbf{S}^\dagger(\textbf{h}_{\p\Sigma})\\
    &=\int_{\mathcal{G}_{\p\Sigma}}[d\textbf{h}_{\p\Sigma}]\ket{\textbf{h}_{\p\Sigma}}\bra{\textbf{h}_{\p\Sigma}}_{\Tilde{\Phi}}\otimes U_\textbf{S}(\textbf{h}_{\p\Sigma}\textbf{g}_{\p\Sigma}^{-1}) f_\textbf{S}U_\textbf{S}^\dagger(\textbf{h}_{\p\Sigma}\textbf{g}_{\p\Sigma}^{-1})\\
    &=\int _{\mathcal{G}_{\p\Sigma\setminus\mathcal{N}}} [\mathcal{D} \mathring{\mathfrak{h}}_x]\ket{\mathring{\mathfrak{h}}_x}\bra{\mathring{\mathfrak{h}}_x}_{\Tilde{\Phi}} \otimes \int_{H}[d\mathfrak{h}] U_{\textbf{S}}(\mathfrak{h}_x\mathfrak{h}\textbf{g}_{\p\Sigma}^{-1})f_\textbf{S}U_{\textbf{S}}^\dagger(\mathfrak{h}_x\mathfrak{h}\textbf{g}_{\p\Sigma}^{-1})
\end{align}
In going to the second line, we performed the change of variables $\textbf{h}_{\p\Sigma}\mapsto\textbf{h}_{\p\Sigma}\textbf{g}_{\p\Sigma}^{-1}$. 
We note that even though the coset space $\mathcal{G}_{\p\Sigma}/H\cong \mathcal{G}_{\p\Sigma\setminus\mathcal{N}}$ as a manifold, it does \emph{not} have the structure of a group, since $H$ is not normal.  We denote its measure by $\mathcal{D}$ to remind the reader that it is only left-invariant. Nevertheless, we still have the fact that the coset states provide a resolution of the identity for $\mathcal{H}_{\Tilde{\Phi}}$. We now perform the change of variables $\mathfrak{h}\mapsto \textbf{g}_{\p\Sigma}^{-1}\mathfrak{h}\textbf{g}_{\p\Sigma}$, which maps the isotropy group to $H_{\textbf{g}_{\p\Sigma}}:=\textbf{g}_{\p\Sigma}H\textbf{g}_{\p\Sigma}^{-1}\cong H$ (with same Haar measure)
\begin{align}
    \mathcal{G}_{\p\Sigma}(\ket{\textbf{g}_{\p\Sigma}}\bra{\textbf{g}_{\p\Sigma}}_{\Tilde{\Phi}}\otimes f_\textbf{S})&=\int _{\mathcal{G}_{\p\Sigma\setminus\mathcal{N}}} [\mathcal{D} \mathring{\mathfrak{h}}_x]\ket{\mathring{\mathfrak{h}}_x}\bra{\mathring{\mathfrak{h}}_x}_{\Tilde{\Phi}} \otimes U_{\textbf{S}}(\mathfrak{h}_x\textbf{g}_{\p\Sigma}^{-1})\l \int_{H_{\textbf{g}_{\p\Sigma}}}[d\mathfrak{h}]U_{\textbf{S}}(\mathfrak{h})f_\textbf{S}U_{\textbf{S}}^\dagger(\mathfrak{h})\r U_{\textbf{S}}^\dagger(\mathfrak{h}_x\textbf{g}_{\p\Sigma}^{-1})
\end{align}
We note that $H_{\textbf{g}_{\p\Sigma}}$ is the isotropy group for the frame orientation $\ket{\textbf{g}_{\p\Sigma}}_{\Tilde{\Phi}}$.

The conclusion is that, due to the incompleteness of the frame, all observables dressed to it are  averaged over the isotropy group. So the algebra of relational observables for the intrinsic frame is \emph{coarser} than for the extrinsic one
\begin{align}
    \mathcal{G}_{\p\Sigma}(\ket{\textbf{g}_{\p\Sigma}}\bra{\textbf{g}_{\p\Sigma}}_{\Tilde{\Phi}}\otimes \mathcal{B}(\mathcal{H}_\textbf{S})) \cong \mathcal{G}_{\p\Sigma}(\ket{\textbf{g}_{\p\Sigma}}\bra{\textbf{g}_{\p\Sigma}}_{\Tilde{\Phi}}\otimes \mathfrak{H}_{\textbf{g}_{\p\Sigma}}\l\mathcal{B}(\mathcal{H}_\textbf{S})\r)\,,
\end{align}
where $\mathfrak{H}_{\textbf{g}_{\p\Sigma}}(\cdot)$ denotes the $H_{\textbf{g}_{\p\Sigma}}$-twirl. This will play a crucial role in the properties of the intrinsic algebra (see App. \ref{app:alg_content}). We note that $H_{\textbf{g}_{\p\Sigma}}=H_{\mathfrak{g}_x}$ for any representative $\mathfrak{g}_x$.

Going back to our concrete examples, where $f_\textbf{S}$ is supported only on vertices $v$ and $w$, we need to evaluate
\begin{align}
    \int [dh_v][dh_w] \ket{h_v}\bra{h_v}_{\Tilde{\Phi}_v}\otimes \ket{h_w}\bra{h_w}_{\Tilde{\Phi}_w} \otimes \mathcal{U}_v(h_vg_v^{-1})\mathcal{U}_w(h_wg_w^{-1}) \;\mathfrak{H}_{\textbf{g}_{\p\Sigma}}(f_\textbf{S}) \;\mathcal{U}_v^\dagger(h_vg_v^{-1})\mathcal{U}_w^\dagger(h_wg_w^{-1})\,,
\end{align}
where the isotropy group average is itself only supported on (at most) the two vertices $v$ and $w$.

\paragraph*{2A. Open Wilson line:} We take $f_\textbf{S}=\hat{g}_l$. We note that the action of the isotropy group $H_{\textbf{g}_{\p\Sigma}}\cong G$ on the two endpoints of the open link $l$ is precisely isomorphic to the action of the diagonal subgroup of the independent left/right  group $G_L\times G_R$. The group average with respect to it will project onto its singlet representation. Its isotropy average is 
\begin{align}
    \mathfrak{H}(\hat{g}_l)=\int_H [dh]\, \rho(h^{-1})\,\hat{g}_l\,\rho(h) = \Tr_\rho(\hat{g}_l) \mathbb{1}_\rho\,,
\end{align}
where $\Tr_\rho$ is the trace in the representation $\rho$.
$\hat{g}_l$ lives in the representation space $\mathbf{V}^L_\rho\otimes (\mathbf{V}^R_\rho)^*$ of $G_L\times G_R$ and so the integral projects onto the singlet sector of $\mathbf{V}_\rho\otimes(\mathbf{V}_\rho)^*$, which is given by the trace. Averaging over $H_{\textbf{g}_{\p\Sigma}}$ gives the same result with the appropriate modification
\begin{align}
    \mathfrak{H}_{\textbf{g}_{\p\Sigma}}(\hat{g}_l)=\Tr_\rho\l\rho(g_v^{-1})\,\hat{g}_l\,\rho(g_w)\r\rho(g_vg_w^{-1})\,.
\end{align}

Plugging this in, we get the relational observable 
\begin{align}
    \mathcal{G}_{\p\Sigma}(\ket{\textbf{g}_{\p\Sigma}}\bra{\textbf{g}_{\p\Sigma}}_{\Tilde{\Phi}}\otimes \hat{g}_l)&= \rho(g_vg_w^{-1})\,\Tr_\rho\left\{ \l\int [dh_v]\,\rho(h_v^{-1}) \ket{h_v}\bra{h_v}_{\Tilde{\Phi}_v} \r \otimes \hat{g}_l \otimes \l\int [dh_w]\,\rho(h_w)\ket{h_w}\bra{h_w}_{\Tilde{\Phi}_w}\r \right\}\\
    &=\rho(g_vg_w^{-1})\,\Tr_\rho\l\hat{g}_{\Tilde{\Phi}_v}^\dagger\otimes\hat{g}_l\otimes \hat{g}_{\Tilde{\Phi}_w} \r\,.\label{eq:intloop}
\end{align}
This is nothing but a Wilson loop consisting of the link $l$ and closed with the intrinsic frame, with a consistent path-orientation.

\paragraph*{2B. Left/Right electric field:}
We now take $f_\textbf{S}=(\hat{\mathfrak{G}}^v_e)^a$, which lives in the adjoint representation space of the structure group $G$ at $v$ and hence also of $H$. For non-Abelian groups, the adjoint has no support on the singlet, so its isotropy average  vanishes\footnote{For Abelian theories, the adjoint representation coincides with the singlet representation, so the electric field itself is invariant under the isotropy group $H$, and its dressed observable yields back the undressed one.}
\begin{align}
    \mathfrak{H}((\hat{\mathfrak{G}}^v_e)^a)=\int_H [dh]\, D(h^{-1})^{ab}\,(\hat{\mathfrak{G}}^v_e)^b = 0\,.
\end{align}
The same result holds for $\mathfrak{H}_{\textbf{g}_{\p\Sigma}}$.  The electric fields dressed with the intrinsic frame are trivial. This is to be expected, since in non-Abelian gauge theory, the generators are not gauge-invariant, but rather only the corresponding Casimirs are. Contrast this with the situation of extrinsic dressing, in which case the electric fields are parallel transported to the global boundary, where they are allowed to be charged under large gauge transformations.  If we restrict ourselves to bulk-supported degrees of freedom, the only way to obtain gauge-invariant observables from electric fields is by taking second order combinations in which the Lie algebra indices are contracted. We turn to explaining this now.

\paragraph*{2C. Casimirs:}

Consider a general quadratic operator of electric fields
\begin{align}
    (\hat{\mathfrak{G}}^v_e)^a(\hat{\mathfrak{G}}^w_{e'})^b\,,
\end{align}
where $e$ and $e'$ are some edges connected to (potentially distinct) vertices $v$ and $w$. This is not gauge-invariant. Let us take this as our $f_\textbf{S}$ and dress it to the intrinsic frame. So to start, we take $v$ and $w$ to lie on $\p\Sigma$. The isotropy average is now non-trivial, since the product of two adjoint representations always has support on the trivial representation
\begin{align}
    \mathfrak{H}((\hat{\mathfrak{G}}^v_e)^a(\hat{\mathfrak{G}}^w_{e'})^b)=\int_H [dh] D(h^{-1})^{ac}(\hat{\mathfrak{G}}^v_e)^cD(h^{-1})^{bd}(\hat{\mathfrak{G}}^w_{e'})^d=\delta^{ab}\delta^{cd}(\hat{\mathfrak{G}}^v_e)^c(\hat{\mathfrak{G}}^w_{e'})^d\,.
\end{align}
So the isotropy average contracts the Lie algebra indices. For the orientation-dependent average we get
\begin{align}
    \mathfrak{H}_{\textbf{g}_{\p\Sigma}}((\hat{\mathfrak{G}}^v_e)^a(\hat{\mathfrak{G}}^w_{e'})^b)= D(g_vg_w^{-1})^{ab}D(g_vg_w^{-1})^{cd}(\hat{\mathfrak{G}}^v_e)^c(\hat{\mathfrak{G}}^w_{e'})^d\,.
\end{align}

The relational observable becomes
\begin{align}
    \mathcal{G}_{\p\Sigma}(\ket{\textbf{g}_{\p\Sigma}}\bra{\textbf{g}_{\p\Sigma}}_{\Tilde{\Phi}}\otimes (\hat{\mathfrak{G}}^v_e)^a(\hat{\mathfrak{G}}^w_{e'})^b)&=D(g_vg_w^{-1})^{ab}\int[dh_v][dh_w]\ket{h_v}\bra{h_v}_{\Tilde{\Phi}_v}\otimes \ket{h_w}\bra{h_w}_{\Tilde{\Phi}_w}\otimes D(h_vh_w^{-1})^{cd}(\hat{\mathfrak{G}}^v_e)^c(\hat{\mathfrak{G}}^w_{e'})^d\\
    &=D(g_vg_w^{-1})^{ab}\l(\hat{g}_\text{Adj})_{\Tilde{\Phi}_v}\otimes (\hat{g}_\text{Adj})_{\Tilde{\Phi}_w}^\dagger\r^{cd}(\hat{\mathfrak{G}}^v_e)^c(\hat{\mathfrak{G}}^w_{e'})^d\,.\label{eq:c28}
\end{align}
This observable encodes the product of the two electric fields, having parallel-transported one of them towards the location of the other along the Wilson line on the intrinsic tree connecting the two vertices. Notice that for the case that $v=w$ we get
\begin{align}
    \mathcal{G}_{\p\Sigma}(\ket{\textbf{g}_{\p\Sigma}}\bra{\textbf{g}_{\p\Sigma}}_{\Tilde{\Phi}}\otimes (\hat{\mathfrak{G}}^v_e)^a(\hat{\mathfrak{G}}^v_{e'})^b)&=\delta^{ab}\delta^{cd}(\hat{\mathfrak{G}}^v_e)^c(\hat{\mathfrak{G}}^v_{e'})^d)\,,
\end{align}
for any two edges $e$, $e'$ hitting the same vertex. As expected, to multiply two electric fields at the same vertex, there is no need of parallel-transport. Nevertheless, the Lie algebra indices must be contracted in order to get a gauge-invariant observable.

More generally, the observable
\begin{align}\label{eq:c30}(\hat{\mathfrak{G}}^v_e)^a(\hat{g}_\text{Adj})_\ell^{ab}(\hat{\mathfrak{G}}^w_{e'})^b\,,
\end{align}
for any Wilson line $\ell$ connecting vertices $v$ and $w$, is a gauge-invariant one. The intrinsic frame just happens to pick a particular such path. But all paths lead to valid dressed observables. Crucially, these are distinct observables. This is a manifestation of the fact that parallel-transport is path-dependent, due to the presence of magnetic curvature. With this consideration, one can clearly relax the condition that $v$ and $w$ lie on $\p\Sigma$. The above prescription tells us how to dress products of electric fields on any two vertices. These considerations will be crucial for understand the various physical algebras, see App.~\ref{app:alg_content}.

\section{Gauge-invariant factorizations on the physical Hilbert space}\label{app:factor}

We start with $\Hkin = \bigotimes_{e} L^2(G)$ corresponding to the original lattice. We then perform the refactorization described in the main body and the end of App.~\ref{app:kin}
to define our  frames. This leads to a modified lattice, in which edges along the frame Wilson lines have been removed and replaced with node degrees of freedom. This is akin to a lattice gauge theory with ``matter'' insertions at those nodes (although usually matter lives in finite-dimensional representations). Here, the frames at node $v$ live in $L^2(G)_{\Phi_v}$ and $L^2(G)_{\Tilde{\Phi}_v}$, respectively. Because we have defined the frames with inward orientation, we have the following transformation properties
\begin{align}
    U_v(g_v)U_{v_0}(g_{v_0})\triangleright\ket{g}_{\Phi_v} &= \ket{g_v g g_{v_0}^{-1}}_{\Phi_v}\,,\\
    U_v(g_v)U_{v_0}(g_{v_0})\triangleright\ket{g}_{\Tilde{\Phi}_v} &= \ket{g_v g g_{v_0}^{-1}}_{\Tilde{\Phi}_v}\,,
\end{align}
where in the case of $\Phi$, $v_0$ denotes the anchor node of $\Phi_v$ on the global boundary $\mathcal{B}$, and, in the case of $\tilde\Phi$, $v_0$ denotes the root of the tree on $\p\Sigma$ defining the frame.

We define the notation $\Phi := \bigcup_{v\in\p\Sigma}\Phi_v$ and $\Tilde{\Phi}:=\bigcup_{v\in\p\Sigma, v\neq\mathcal{N}}\Tilde{\Phi}_v$ to denote the collection of subregion boundary frames. We further define $\text{sub}:= \Sigma \cup \l\p\Sigma\setminus\Tilde{\Phi}\r$ and $\text{compl}:=\Bar{\Sigma}\setminus\Phi$, to represent the subregion, respectively complement, degrees of freedom in the modified lattice, having removed the frame degrees of freedom on the subregion boundary. We then have $\Hkin= \mathcal{H}_\text{sub}\otimes\mathcal{H}_\text{compl}\otimes\mathcal{H}_{\Tilde{\Phi}}\otimes\mathcal{H}_\Phi$. 

Each edge or frame degree of freedom transforms under the product group $G_L\times G_R$ at the corresponding endpoints. For a node $v\in\p\Sigma$  we define the generator of gauge transformations to be: $U_v(h)=\bigotimes_{e\in\mathcal{E}_v^{\rm in}}U_e(h)\bigotimes_{e'\in\mathcal{E}_v^{\rm out}}V_{e'}(h)\bigotimes U_{\Tilde{\Phi}_v}(h)\bigotimes U_{\Phi_v}(h)$, where $\mathcal{E}_v^{\rm in/out}$ denote the sets of in- and outgoing edges at $v$. We define the boundary gauge group $\mathcal{G}_{\p\Sigma}$ to be the cartesian product of the group $G_v$ acting at each boundary node $v\in\p\Sigma$ independently, 
and we denote a joint gauge transformation on the subregion boundary by $U(\textbf{g}_{\p\Sigma}):=\bigotimes_{v\in\p\Sigma}U_v(g_v)$, with $\textbf{g}_{\p\Sigma}\in \mathcal{G}_{\p\Sigma}$. We can repackage this into $U(\textbf{g}_{\p\Sigma})=U_\text{sub}(\textbf{g}_{\p\Sigma})\otimes U_\text{compl}(\textbf{g}_{\p\Sigma})\otimes U_{\Tilde{\Phi}}(\textbf{g}_{\p\Sigma})\otimes U_\Phi(\textbf{g}_{\p\Sigma})$, with the definitions above. 

We first impose the Gauss constraint at every bulk node not belonging to the corner $\p\Sigma$. The corresponding group averaging projectors act locally on $\Hkin$, more precisely, separately on the $\mathcal{H}_\text{sub}$ and $\mathcal{H}_\text{compl}$ factors. This leads to the Hilbert space factors $\mathcal{H}_\text{sub}^\text{phys}$ and $\mathcal{H}_\text{compl}^\text{phys}$, respectively. The resulting total Hilbert space is still ``kinematical'' with respect to the group $\mathcal{G}_{\p\Sigma}$ which motivates the notation
\begin{align} \mathcal{H}_\text{kin}^{\p\Sigma}:=\mathcal{H}_\text{sub}^\text{phys}\otimes \mathcal{H}_\text{compl}^\text{phys}\otimes  \mathcal{H}_{\Tilde{\Phi}}\otimes\mathcal{H}_\Phi\,.
\end{align}

We now proceed to impose the Gauss constraint at the boundary nodes. To do that, we act with the joint coherent group averaging projector
\begin{align}
    \mathbf{\Pi}^\text{phys}_{\p\Sigma}:= \int [d\textbf{g}_{\p\Sigma}]\, U(\textbf{g}_{\p\Sigma})
\end{align}
on the kinematical states.  The following is an application of the formalism developed in \cite{delaHamette:2021oex}.

\paragraph*{Extrinsic perspective:}

Let us take the extrinsic frame as our frame field of interest on $\p\Sigma$ and define
\begin{align}
    \mathcal{H}_\textbf{S}:= \mathcal{H}^\text{phys}_\text{sub}\otimes \mathcal{H}^\text{phys}_\text{compl}\otimes \mathcal{H}_{\Tilde{\Phi}}\,.
\end{align}
Then, our kinematical system (wrt to the gauge group $\mathcal{G}_{\p\Sigma}$) can be written as
\begin{align}
    \mathcal{H}_\text{kin}^{\p\Sigma}:=\mathcal{H}_\Phi\otimes \mathcal{H}_\textbf{S}\,,\quad \text{with}\quad U(\textbf{g}_{\p\Sigma})=U_\Phi(\textbf{g}_{\p\Sigma})\otimes U_\textbf{S}(\textbf{g}_{\p\Sigma})\,.
\end{align}
The physical Hilbert space is simply $\Hphys = \mathbf{\Pi}^\text{phys}_{\p\Sigma}\l\mathcal{H}_\text{kin}^{\p\Sigma}\r$. 
One can now define a lattice version of the Page-Wootters (PW) conditioning procedure, which will induce a frame-dependent TPS on $\Hphys$ \cite{Hoehn:2023ehz,DeVuyst:2024uvd}. We define the PW conditioning map associated to the frame $\Phi$ as 
\begin{align}
    \mathcal{R}_\Phi(\textbf{g}_{\p\Sigma}):=\bra{\textbf{g}_{\p\Sigma}}_\Phi\otimes \mathbb{1}_\textbf{S}\,,\quad \text{where}\quad \ket{\textbf{g}_{\p\Sigma}}_\Phi:=\bigotimes_{v\in\p\Sigma}\ket{g_v}_{\Phi_v}\,.
\end{align}
The latter is merely a choice of frame orientation for $\Phi$ on $\p\Sigma$. The PW conditioning is, therefore, a map $\mathcal{R}_\Phi(\textbf{g}_{\p\Sigma}):\Hphys \to \mathcal{H}_\textbf{S}$. This map can be shown to be invertible on its image, with inverse \cite{delaHamette:2021oex}:
\begin{align}\label{eq:Rinv}
    \l\mathcal{R}_\Phi(\textbf{g}_{\p\Sigma})\r^{-1} = \mathbf{\Pi}^\text{phys}_{\p\Sigma} \l\ket{\textbf{g}_{\p\Sigma}}_\Phi\otimes \mathbb{1}_\textbf{S}\r\,.
\end{align}
Crucially, it was also shown in \cite{delaHamette:2021oex}, that for complete ideal frames, like $\Phi$, the image under $\mathcal{R}^\Phi(\textbf{g}_{\p\Sigma})$ is precisely $\mathcal{H}_\textbf{S}$. This means that the PW conditioning provides a Hilbert space isomorphism between the physical Hilbert space and the Hilbert space of the rest of the degrees of freedom, excluding the frame
\begin{align}
    \Hphys \cong \mathcal{H}_\textbf{S}= \mathcal{H}^\text{phys}_\text{sub}\otimes \mathcal{H}^\text{phys}_\text{compl}\otimes \mathcal{H}_{\Tilde{\Phi}}\,,
\end{align}
which is a tensor product of Hilbert spaces. Crucially, this defines a \emph{gauge-invariant} TPS on $\Hphys$, as noted in \cite{Hoehn:2023ehz}. Indeed, let us generalize the argument of \cite[Sec.~III]{Hoehn:2023ehz} to non-Abelian theories. A TPS on an abstract Hilbert space $\H$ is an equivalence class of isomorphisms $\mathbf{T}:\H\to\bigotimes_\alpha\H_\alpha$ such that $\mathbf{T}_1\sim\mathbf{T}_2$ if $\mathbf{T}_2(\mathbf{T}_1)^\dag$ is a product of multilocal unitaries $\otimes_\alpha U_\alpha$ (and permutations of subsystem factors of equal dimension). We have the unitary
\begin{equation}\label{eq:extfactor}
    \mathcal{R}_\Phi(\mathbf{g}_{\p\Sigma}):\Hphys\to \mathcal{H}^\text{phys}_\text{sub}\otimes \mathcal{H}^\text{phys}_\text{compl}\otimes \mathcal{H}_{\Tilde{\Phi}}\,.
\end{equation}
To establish that the PW reduction map $\mathcal{R}_\Phi(\mathbf{g}_{\p\Sigma})$, i.e.\ $\Phi$'s internal QRF perspective, is a (representative of a) gauge-invariant TPS on $\Hphys$, we need to show that it does not depend on $\Phi$'s orientation, i.e.\ on the gauge. Indeed, invoking the left covariance of the QRF orientation basis under gauge transformations, we have ${\mathcal{R}_\Phi(\mathbf{g}_{\p\Sigma})\sim\mathcal{R}_\Phi(\mathbf{g}'_{\p\Sigma})}$, $\forall\,\mathbf{g}_{\p\Sigma},\mathbf{g}'_{\p\Sigma}\in\mathcal{G}_{\p\Sigma}$, because, on the domain $\Hphys$, they are related by a multilocal unitary $\mathcal{R}_\Phi(\mathbf{g}_{\p\Sigma}'\mathbf{g}_{\p\Sigma})=\left(U_{\rm sub}^{\rm phys}(\mathbf{g}_{\p\Sigma}')\otimes U_{\rm compl}^{\rm phys}(\mathbf{g}'_{\p\Sigma})\otimes U_{\tilde\Phi}(\mathbf{g}'_{\p\Sigma})\right)\mathcal{R}_\Phi(\mathbf{g}_{\p\Sigma})$. Hence, the equivalence class $\mathcal{R}_\Phi:=\big[\mathcal{R}_\Phi(\mathbf{g}_{\p\Sigma})\big]$ is a TPS on $\Hphys$. This TPS can equivalently be defined in terms of commuting subalgebras of relational observables associated with the QRF \cite{Hoehn:2023ehz}. Non-ideal and/or incomplete QRFs map onto a subspace of the above TPS \cite{delaHamette:2021oex}.

So, from the perspective of the extrinsic frame, the physical Hilbert space \emph{does} factorize between the subregion and the complement. The important caveat is that the complement Hilbert space does not include the frame degrees of freedom. This is, thus, consistent with the statement that one cannot factorize between ``inside'' and ``outside'' if we were to include all the degrees of freedom on both sides. One can immediately also see that the TPS is frame-dependent. Had we chosen a different Wilson line system to define our $\Phi$, we would have ended up not only with a different $\mathcal{H}^\text{phys}_\text{compl}$ factor, but also the subsystem itself would contain different observables, as detailed in App. \ref{app:subrel}. This is a manifestation of subsystem relativity on the lattice \cite{Hoehn:2023ehz}.

Another interesting feature of this factorization is the extra factor $\mathcal{H}_{\Tilde{\Phi}}$. This will be shown in App. \ref{app:GM} to correspond to the corner Goldstone mode degrees of freedom appearing in continuum classical setups \cite{Araujo-Regado:2024dpr,Ball:2024hqe,Ball:2024gti}.

\paragraph*{Intrinsic Perspective:}

We now split our kinematical degrees of freedom in terms of the $\Tilde{\Phi}$ frame and the rest, re-adapting the above notation for \textbf{S}
\begin{align}\label{eq:intframepers}
    \mathcal{H}_\textbf{S}:= \mathcal{H}^\text{phys}_\text{sub}\otimes \mathcal{H}^\text{phys}_\text{compl}\otimes \mathcal{H}_\Phi\,,\quad \mathcal{H}_\text{kin}^{\p\Sigma}=\mathcal{H}_{\Tilde{\Phi}}\otimes\mathcal{H}_\textbf{S}\,.
\end{align}
As a manifestation of the intrinsic frame being incomplete, the PW conditioning map is now
\begin{align}
    \mathcal{R}_{\Tilde{\Phi}}(\mathring{\textbf{g}}_{\p\Sigma}):=\bra{\mathring{\textbf{g}}_{\p\Sigma}}_{\Tilde{\Phi}}\otimes \mathbb{1}_\textbf{S}\,,\quad \text{where}\quad \ket{\mathring{\textbf{g}}_{\p\Sigma}}_{\Tilde{\Phi}}:=\bigotimes_{v\in\p\Sigma, v\neq \mathcal{N}}\ket{g_v}_{\Tilde{\Phi}_v}\,.
\end{align}
The notation $\mathring{\textbf{g}}_{\p\Sigma}$ refers to an element of $\mathcal{G}_{\p\Sigma\setminus \mathcal{N}}$. In other words, the orientation labels of $\Tilde{\Phi}$ do not parameterize the whole group orbit of $\mathcal{G}_{\p\Sigma}$. This map is still invertible on its image, with the analogous definition for the inverse: $\l\mathcal{R}_{\Tilde{\Phi}}(\mathring{\textbf{g}}_{\p\Sigma})\r^{-1} = \mathbf{\Pi}^\text{phys}_{\p\Sigma} \l\ket{\mathring{\textbf{g}}_{\p\Sigma}}_{\Tilde{\Phi}}\otimes \mathbb{1}_\textbf{S}\r$. However, the image is a strict subspace of $\mathcal{H}_\textbf{S}$. In particular, as shown in \cite{delaHamette:2021oex}, the projector on $\mathcal{H}_\textbf{S}$ takes the form
\begin{align}
    \mathbf{P}^\text{phys}_\textbf{S}(\mathring{\textbf{g}}_{\p\Sigma})&:=\mathcal{R}_{\Tilde{\Phi}}(\mathring{\textbf{g}}_{\p\Sigma})\l\mathcal{R}_{\Tilde{\Phi}}(\mathring{\textbf{g}}_{\p\Sigma})\r^{-1}\\
    &=\l\bra{\mathring{\textbf{g}}_{\p\Sigma}}_{\Tilde{\Phi}}\otimes \mathbb{1}_\textbf{S}\r\mathbf{\Pi}^\text{phys}_{\p\Sigma}\l\ket{\mathring{\textbf{g}}_{\p\Sigma}}_{\Tilde{\Phi}}\otimes \mathbb{1}_\textbf{S}\r\\
    &= \int [d\textbf{g}_{\p\Sigma}']\l\bra{\mathring{\textbf{g}}_{\p\Sigma}}_{\Tilde{\Phi}}\otimes \mathbb{1}_\textbf{S}\r U(\textbf{g}_{\p\Sigma}')\l\ket{\mathring{\textbf{g}}_{\p\Sigma}}_{\Tilde{\Phi}}\otimes \mathbb{1}_\textbf{S}\r\\
    &=\int [dg_\mathcal{N}'][d \mathring{\textbf{g}}_{\p\Sigma}']\braket{\mathring{\textbf{g}}_{\p\Sigma}\,}{\,\mathring{\textbf{g}}_{\p\Sigma}'\,\mathring{\textbf{g}}_{\p\Sigma}g_\mathcal{N}'^{-1}}_{\Tilde{\Phi}}U_\textbf{S}(\textbf{g}_{\p\Sigma}')\,, 
\end{align}
where the notation for the inner product means $\braket{\mathring{\textbf{g}}_{\p\Sigma}\,}{\,\mathring{\textbf{g}}_{\p\Sigma}'\,\mathring{\textbf{g}}_{\p\Sigma}{g'}_\mathcal{N}^{-1}}_{\Tilde{\Phi}}=\prod_{v\in\p\Sigma, v\neq\mathcal{N}}\braket{g_v}{g'_v g_v {g'}_\mathcal{N}^{-1}}_{\Tilde{\Phi}_v}$. Using the perfect distinguishability of the frame coherent states, we get 
\begin{align}
\mathbf{P}^\text{phys}_\textbf{S}(\mathring{\textbf{g}}_{\p\Sigma}) = \int [dh] \,U_{\textbf{S}, \mathcal{N}}(h)\bigotimes_{v\in\p\Sigma, v\neq\mathcal{N}}U_{\textbf{S}, v}(g_v\,h\,g_v^{-1})\,,
\end{align}
where the unitaries in the expression above are implementing $G$-transformations on the system $\textbf{S}$, by $h$ at the anchor point $\mathcal{N}$ and by $g_v\,h\,g_v^{-1}$ at the other nodes $\p\Sigma\ni v\neq\mathcal{N}$.
This $\mathbf{P}^\text{phys}_\textbf{S}$ is thus an operator acting on $\mathcal{H}_\textbf{S}$ only. It is a projector, because $\l\mathbf{P}^\text{phys}_\textbf{S}(\mathring{\textbf{g}}_{\p\Sigma})\r^2=\mathbf{P}^\text{phys}_\textbf{S}(\mathring{\textbf{g}}_{\p\Sigma})$ and $\l\mathbf{P}^\text{phys}_\textbf{S}(\mathring{\textbf{g}}_{\p\Sigma})\r^\dagger=\mathbf{P}^\text{phys}_\textbf{S}(\mathring{\textbf{g}}_{\p\Sigma})$, where the proof is analogous to the one in \cite{delaHamette:2021oex}. Crucially, $\mathbf{P}^\text{phys}_\textbf{S}(\mathring{\textbf{g}}_{\p\Sigma})\neq \mathbb{1}_\textbf{S}$, so it is a non-trivial projector. At first this might appear somewhat surprising since the frame is ideal (and in mechanical settings, ideal frames lead to trivial projectors). However, the frame is incomplete and so it has a non-trivial isotropy subgroup  $H\subset\mathcal{G}_{\p\Sigma}$. Following \cite{delaHamette:2021oex}, when this happens, the physical Hilbert space in the frame perspective corresponds to the subspace of $\mathcal{H}_\textbf{S}$ that is invariant under $H$; the frame can only resolve what is invariant under its isotropy group. This is precisely what $\mathbf{P}^\text{phys}_\textbf{S}(\mathring{\textbf{g}}_{\p\Sigma})$ is doing. To see this, note that when $\Tilde{\Phi}$ is in the orientation $\ket{\mathring{\textbf{g}}_{\p\Sigma}}_{\Tilde{\Phi}}=\bigotimes_{v\in\p\Sigma, v\neq \mathcal{N}}\ket{g_v}_{\Tilde{\Phi}_v}$ we have that
\begin{align}
     \l U_{\Tilde{\Phi}, \mathcal{N}}(h)\bigotimes_{v\in\p\Sigma, v\neq\mathcal{N}}U_{\Tilde{\Phi}, v}(g_v\,h\,g_v^{-1})\r \ket{\mathring{\textbf{g}}_{\p\Sigma}}_{\Tilde{\Phi}}&= \bigotimes_{v\in\p\Sigma, v\neq\mathcal{N}}U_{\Tilde{\Phi}, \mathcal{N}}(h)U_{\Tilde{\Phi}, v}(g_v\,h\,g_v^{-1})\ket{g_v}_{\Tilde{\Phi}_v}\\
     &=\bigotimes_{v\in\p\Sigma, v\neq\mathcal{N}} \ket{(g_v\,h\,g_v^{-1})g_v h^{-1}}_{\Tilde{\Phi}_v}=\ket{\mathring{\textbf{g}}_{\p\Sigma}}_{\Tilde{\Phi}}\,,\quad \forall h\in G\,.
\end{align}
Thus,  $\mathbf{P}^\text{phys}_\textbf{S}(\mathring{\textbf{g}}_{\p\Sigma})$ is nothing but a coherent group average with respect to the orientation-dependent intrinsic isotropy group $H(\mathring{\textbf{g}}_{\p\Sigma}):=\{(e_\mathcal{N},\mathring{\textbf{g}}_{\p\Sigma})\cdot(h,...,h)\cdot(e_\mathcal{N},\mathring{\textbf{g}}_{\p\Sigma})^{-1}|\,h\in G\}$. This has a single group parameter and corresponds to a ``uniform'' gauge transformation on $\p\Sigma$. It is isomorphic to the structure group $H(\mathring{\textbf{g}}_{\p\Sigma})\cong G,\;\forall\,\mathring{\textbf{g}}_{\p\Sigma}\in\mathcal{G}_{\p\Sigma\setminus\mathcal{N}}$. For the identity frame orientation $\ket{\mathring{\textbf{e}}_{\p\Sigma}}_{\Tilde{\Phi}}$, the projector is exactly a $\p\Sigma$-uniform group action on \textbf{S}. We use the shorthand notation
\begin{align}
    \mathbf{P}^\text{phys}_\textbf{S}(\mathring{\textbf{g}}_{\p\Sigma}) =: \mathbf{P}^\text{inv}_\textbf{S}\l H(\mathring{\textbf{g}}_{\p\Sigma})\r
\end{align}
to denote projection onto the invariant subspace with respect to $H(\mathring{\textbf{g}}_{\p\Sigma})$.

This projection operator generally acts non-trivially on every factor of $\mathcal{H}_\textbf{S}$ and thus breaks the kinematical TPS. The physical Hilbert space therefore does not inherit any subregion/complement factorization:
\begin{align}
    \Hphys \cong \mathbf{P}^\text{inv}_\textbf{S}\l H(\mathring{\textbf{g}}_{\p\Sigma})\r\l\mathcal{H}_\textbf{S}\r\,.
\end{align}
As explained in \cite{delaHamette:2021oex}, as we change the orientation $\mathring{\textbf{g}}_{\p\Sigma}$ onto which we condition, the image subspace rotates inside $\mathcal{H}_\textbf{S}$. This is a consequence of the fact that the intrinsic frame lacks frame reorientations. Nevertheless, all these different projected subspaces are isomorphic.

We now note a couple of important differences between Abelian (including finite Abelian groups) and non-Abelian gauge theories. First, in the former case, the isotropy subgroup is actually orientation-independent. Second, even though the projector is a priori still non-trivial, if we have no matter, it acts trivially on almost all factors in $\mathcal{H}_\textbf{S}$. This is because states in $\mathcal{H}_\Sigma^\text{phys}$ and $\mathcal{H}_\text{compl}^\text{phys}$ are already invariant under $\p\Sigma$-uniform gauge transformations, as a consequence of having imposed the Gauss constraints at every node inside, respectively, outside. Both factors include only states of zero total electric charge. Further, $\mathcal{H}_{\p\Sigma\setminus\Tilde{\Phi}}$ is also already invariant under this uniform action. So, in the end, the projector only acts non-trivially on the extrinsic frame factor $\mathcal{H}_\Phi$. This means that it acts locally on the kinematical TPS and so we recover a factorization on the physical Hilbert space
\begin{align}\label{eq:Abelintfact}
    \text{Abelian case:}\quad \Hphys \cong \mathcal{H}_\text{sub}^\text{phys}\otimes \mathcal{H}_\text{compl}^\text{phys}\otimes \textbf{P}^\text{inv}_\Phi(H)\l\mathcal{H}_\Phi\r\,.
\end{align}

On the other hand, in non-Abelian theories (even without matter), it is not true that having imposed the Gauss constraint at every node inside and outside leads to states of zero total electric flux in and out of $\p\Sigma$. This is a manifestation of the fact that gluons carry color charge, while photons are electrically neutral. More concretely, this is because, on each edge, $\hat{L}^a\neq\hat{R}^a$, but are rather related via \eqref{eq:l=r}. This makes it such that the different electric fields do not cancel when adding together the bulk Gauss constraints, unlike in the Abelian case. Therefore, both $\mathcal{H}_\Sigma^\text{phys}$ and $\mathcal{H}_\text{compl}^\text{phys}$ include states with arbitrary total color charge. In addition, states in $\mathcal{H}_{\p\Sigma\setminus\Tilde{\Phi}}$ are not invariant under $\p\Sigma$-uniform transformations either, since both the holonomies and electric field operators transform non-trivially. Therefore, the isotropy projector has a genuinely non-local action on all factors of $\mathcal{H}_\textbf{S}$, spoiling a clean factorization between regional degrees of freedom and the complement.

\section{The Subregional Goldstone Mode}
\label{app:GM}
Let us now return to the extrinsic frame perspective, which we saw in Eq.~\eqref{eq:extfactor} yields a neat factorization between the in- and outside of the region, as well as between the intrinsic frame, the corner complement of the intrinsic frame, and the bulk for the regional factor. We will now revisit the latter factorization of the regional factor $\H^{\rm phys}_{\rm sub}\otimes\H_{\tilde\Phi}$ and show that one can do a further refactorization that isolates a Goldstone mode (GM) for the region, mimicking the continuum discussion in \cite{Araujo-Regado:2024dpr,Ball:2024gti,Ball:2024hqe}. The key idea is to reorganize the regional degrees of freedom by dressing, as much as possible, to the intrinsic frame $\tilde\Phi$ and only dressing to the extrinsic one $\tilde\Phi$ itself. This has the effect of isolating the action of frame reorientations onto the GM sector, while leaving the other intrinsically dressed degrees of freedom unaffected by reorientations. 
This procedure yields a simple refactorization for Abelian theories. For non-Abelian theories, however, this reorganization of the regional degrees of freedom is only possible on a subspace of the physical Hilbert space. This is again due to the incompleteness of $\tilde\Phi$ and its non-trivial isotropy group. These observations will prove useful for the main body. We will first make a few preparatory technical remarks before proceeding to explicitly construct the GM sector.

\paragraph{Frame Reorientations:} Frame reorientations are physical symmetries that change the QRF's orientations, while leaving its complement, the system, invariant \cite{delaHamette:2021oex}. This changes the relation between the QRF and the system. In order to commute with gauge transformations, the reorientations have to act from the ``other side'' on the orientation labels than the gauge transformations. In the present context, applied to an extrinsic edge QRF $\Phi$, this means its reorientations act on its frame coherent states by regular right multiplication, while acting as the identity on the subregion (but may act non-trivially on $\bar\Sigma\setminus\Phi$). One way to realize them is by large gauge transformations, see the discussion in the main body. On $\mathcal{H}_\text{kin}^{\p\Sigma}$ they take the form
\begin{align}
  {V}_\Phi(\mathbf{h}_{\p\Sigma})\otimes \mathbb{1}_\text{sub}\otimes \mathbb{1}_{\Tilde{\Phi}}\otimes \mathcal{U}(\mathbf{h}_{\p\Sigma})_\text{compl}\,,\quad \text{such that}\quad {V}_\Phi(\mathbf{h}_{\p\Sigma})\ket{\textbf{g}_{\p\Sigma}}_\Phi=\ket{\textbf{g}_{\p\Sigma}\mathbf{h}_{\p\Sigma}^{-1}}_\Phi\,,
\end{align}
with $\mathcal{U}(\mathbf{h}_{\p\Sigma})_\text{compl}$ some (possibly complicated) complement-supported unitary commuting with all gauge generators there. This complement action is determined by ${V}_\Phi(\mathbf{h}_{\p\Sigma})$ itself. The frame reorientation group is thus $\mathbb{G}_{\p\Sigma}\cong \mathcal{G}_{\p\Sigma}$.

Under the PW conditioning, the frame reorientation operators get mapped to operators on the reduced Hilbert space
\begin{align}
    &\mathcal{R}_\Phi(\textbf{g}_{\p\Sigma})\l{V}_\Phi(\mathbf{h}_{\p\Sigma})\otimes \mathbb{1}_\text{sub}\otimes \mathbb{1}_{\Tilde{\Phi}}\otimes \mathcal{U}(\mathbf{h}_{\p\Sigma})_\text{compl}\r\l\mathcal{R}_\Phi(\textbf{g}_{\p\Sigma})\r^{\dag}\\
    &=U_\text{sub}(\Bar{\textbf{h}}_{\p\Sigma})\otimes U_{\Tilde{\Phi}}(\Bar{\textbf{h}}_{\p\Sigma})\otimes U_\text{compl}(\Bar{\textbf{h}}_{\p\Sigma})\,\mathcal{U}(\mathbf{h}_{\p\Sigma})_\text{compl}\,,\quad \text{where} \quad \Bar{\textbf{h}}_{\p\Sigma}:=\textbf{g}_{\p\Sigma}\textbf{h}_{\p\Sigma}\textbf{g}_{\p\Sigma}^{-1}\,.
\end{align}
So we see that, up to the gauge-invariant complement unitary, on the reduced Hilbert space, frame reorientations ``look like'' gauge transformations on $\p\Sigma$ acting on all other degrees of freedom except the frame $\Phi$. The group is still $\textbf{g}_{\p\Sigma}\mathbb{G}_{\p\Sigma}\textbf{g}_{\p\Sigma}^{-1}= \mathbb{G}_{\p\Sigma}$ and we label the action above on $\mathcal{R}_\Phi(\Hphys)$ by 
\begin{align} \label{eq:fr}
    \mathbb{U}(\textbf{h}_{\p\Sigma})=\mathbb{U}_\text{sub}(\textbf{h}_{\p\Sigma})\otimes \mathbb{U}_{\Tilde{\Phi}}(\textbf{h}_{\p\Sigma})\otimes\mathbb{U}_\text{compl}(\textbf{h}_{\p\Sigma})\,\,,
\end{align}
where for $\mathbb{U}_\text{compl}(\textbf{h}_{\p\Sigma})$ it is understood that extra complement unitaries may be included.

\paragraph{Reorientation Isotropy Subgroup of $\Tilde{\Phi}$:} 

In the extrinsic perspective, frame reorientations act as in Eq.~\eqref{eq:fr}. We, thus, land on a similar structure to the one referring to gauge transformations. We have a tensor product action of the symmetry group $\mathbb{G}_{\p\Sigma}$, with the difference now that we do not want to project onto the $\mathbb{G}_{\p\Sigma}$-invariant subspace, as these are physical symmetries. Nevertheless, the factor $\mathcal{H}_{\Tilde{\Phi}}$ acts as a \emph{physical}, i.e.\ gauge-invariant reference frame for this group action (recall that we are in $\Phi$'s perspective, so, upon inverting the reduction map, $\tilde\Phi$ becomes dressed by $\Phi$, thus being gauge-invariant). It is an ideal, \emph{incomplete} frame for $\mathbb{G}_{\p\Sigma}$, with isotropy group given by 
\begin{align} \label{eq:isotropy}
    \mathbb{H} = \{\mathfrak{h}:=(h,...,h)\,|\, h\in G\}\,.
\end{align}
(We use the notation $\mathbb{H}$ to distinguish it from the gauge isotropy group $H$.)
As before this is the representative of the set $\{\mathbb{H}_{\mathring{\textbf{g}}_{\p\Sigma}}:= \mathring{\textbf{g}}_{\p\Sigma}\mathbb{H}\mathring{\textbf{g}}_{\p\Sigma}^{-1}\}$ that leaves the coherent state $\ket{\mathring{\textbf{e}}_{\p\Sigma}}_{\Tilde{\Phi}}$ invariant. $\mathbb{H}$ corresponds to a $\p\Sigma$-uniform frame reorientation.

Paralleling the discussion in App. \ref{app:relational_obs}, the coherent states of $\mathcal{H}_{\Tilde{\Phi}}$ label points in the coset space $\mathbb{G}_{\p\Sigma}/\mathbb{H}$ as follows
\begin{align}
    \mathbb{G}_{\p\Sigma}/\mathbb{H}=\{\textbf{g}_{\p\Sigma}\mathbb{H}\} \quad \leftrightarrow \quad \mathcal{H}_{\Tilde{\Phi}}=\text{span}_\mathbb{C}\{\ket{\mathfrak{g}_{\p\Sigma}}_{\Tilde{\Phi}}|\mathfrak{g}_{\p\Sigma}:=\textbf{g}_{\p\Sigma}(g_\mathcal{N}^{-1},...,g_\mathcal{N}^{-1})=:(e_\mathcal{N},\mathring{\mathfrak{g}}_{\p\Sigma})\}\,,
\end{align}
where $\mathring{\mathfrak{g}}_{\p\Sigma}\in \mathbb{G}_{\p\Sigma\setminus\mathcal{N}}$. Any other representative of the coset will be of the form $\textbf{g}_{\p\Sigma}'=\mathfrak{g}_{\p\Sigma}\mathfrak{h}$, for some $\mathfrak{h}\in\mathbb{H}$ and so yields the same coherent state: $\ket{\textbf{g}_{\p\Sigma}'}_{\Tilde{\Phi}}=\ket{\mathfrak{g}_{\p\Sigma}\mathfrak{h}}_{\Tilde{\Phi}}=\ket{\mathfrak{g}_{\p\Sigma}}_{\Tilde{\Phi}}$. In what follows, we will use these coset labels ($\mathfrak{g}_{\p\Sigma}=(e_\mathcal{N},\mathring{\mathfrak{g}}_{\p\Sigma})$) for the coherent states of $\mathcal{H}_{\Tilde{\Phi}}$. When we perform integrals over $\mathbb{G}_{\p\Sigma}\setminus\mathbb{H}$ we need to remember that this is not a group manifold, so, in particular, the right-invariance of the measure does not hold. That being said, as a manifold, it is isomorphic to $\mathbb{G}_{\p\Sigma\setminus\mathcal{N}}\subset \mathbb{G}_{\p\Sigma}$.

\paragraph{Disentangling Re-factorizations:}
We are now ready to introduce the re-factorization map that isolates the frame reorientation action. It is called a ``disentangling'' re-factorization to contrast it with the original factorization in which the frame reorientation action is entangled across the various subregional systems.

As a warm-up, we first define the maps for a mechanical system, following \cite{delaHamette:2021oex} (also used in \cite{Carrozza:2024smc,Hoehn:2019fsy,Chataignier:2024eil,Hoehn:2021flk,Hoehn:2020epv,Vanrietvelde:2018dit,Vanrietvelde:2018pgb,Giacomini:2021gei,Castro-Ruiz:2021vnq}), and then extend the construction to the lattice. In the process, we wish to make clear that the obstruction to a clean GM factorization comes entirely from the existence of the isotropy group.

Consider a kinematical system of the form $\Hkin = \mathcal{H}_R\otimes \mathcal{H}_S$ with a tensor product action for some group $G$.
Starting with the case of complete, ideal frames, we introduce the unitary \emph{trivialization map} which re-factorizes $\Hkin$ as
\begin{align}
    \mathcal{T}_R:\mathcal{H}_R\otimes \mathcal{H}_S \to \mathcal{H}_{\Tilde{R}}\otimes\mathcal{H}_{\Tilde{S}}\,,
\end{align}
in such a way that the symmetry action gets mapped to
\begin{align}\label{eq:refact}
    \mathcal{T}_R\l U_R(g)\otimes U_S(g)\r\mathcal{T}_R^\dagger=U_R(g)\otimes \mathbb{1}_S\,.
\end{align}
For ideal frames, the map takes the form
\begin{align}
    \mathcal{T}_R&:=\int_G[dg]\ket{g}\!\bra{g}_R\otimes U_S^\dagger(g)\\
    \mathcal{T}_R^\dagger &:= \int_G[dg] \ket{g}\!\bra{g}_R\otimes U_S(g)\,.
\end{align}
These are non-local unitaries implementing a change of TPS. They perform a conditional unitary on the system. This has the effect of transferring all the $G$-noninvariant data into the frame $R$, while repackaging the joint degrees of freedom into a new invariant system, denoted by $\Tilde{S}$. In other words, it disentangles the group action between $R$ and $S$. The action on states in $\mathcal{H}_R\otimes \mathcal{H}_S$ is
\begin{align}
    \mathcal{T}_R(\ket{g}_R\otimes\ket{\psi}_S)=\ket{g}_R\otimes U_S^\dagger(g)\ket{\psi}_S=:\ket{g}_{\Tilde{R}}\otimes \ket{\Tilde{\psi}}_{\Tilde{S}}\,
\end{align}
where we define the state $\ket{\Tilde{\psi}}_{\Tilde{S}}$ to be the (due to Eq.~\eqref{eq:refact}) $G$-invariant state to which $\ket{\psi}_S$ gets mapped under this non-local dressing with $R$. 

In this case it turns out that $\mathcal{H}_{\Tilde{S}}\cong \mathcal{H}_S$, which is expected since we know the joint $G$-invariant subspace of $\Hkin$ to be $\mathcal{H}^G_{RS}\cong \mathcal{H}_S$. This will no longer hold when the frame $R$ is incomplete. We note that the unitary $\mathfrak{U}$ used to re-factorize the kinematical Hilbert space to build Wilson line frames in App.~\ref{app_kinrefactor} is nothing more than a composition of such trivialization maps.

The story changes when the frame $R$ has an isotropy subgroup $H$, which is the relevant case for us here. The operators $\mathcal{T}_R$ and $\mathcal{T}_R^\dagger$, defined  above, fail to have the desired properties. To see this we explicitly separate the effect of the isotropy group by splitting the integral over $G$ into an integral over $H$ and another over the coset space (cf.~\cite[Lem.~1]{delaHamette:2021oex}), to get
\begin{align}
    &\mathcal{T}_R =\int_X[dx]\l\ket{\mathfrak{g}_x}\!\bra{\mathfrak{g}_x}_R\otimes P_S^H U_S^\dagger(\mathfrak{g}_x)\r\\ 
    &\mathcal{T}_R^\dagger=\int_X[dx]\l\ket{\mathfrak{g}_x}\!\bra{\mathfrak{g}_x}_R\otimes U_S(\mathfrak{g}_x)P_S^H\r\,,
\end{align}
where $P_S^H$ projects orthogonally onto the $H$-invariant subspace of $S$ and $\mathfrak{g}_x$ is a coset representative for the group element $g$: $\ket{g}_R=\ket{\mathfrak{g}_x}_R$. From here we can show that the following relations hold 
\begin{align}
    \mathcal{T}_R\l U_R(g)\otimes U_S(g)\r\mathcal{T}_R^\dagger&=U_R(g)\otimes P_S^H\\
    \mathcal{T}_R\mathcal{T}_R^\dagger&=\mathbb{1}_R\otimes P_S^H\\
    \mathcal{T}_R^\dagger\mathcal{T}_R &=\int_X [dx] \ket{\mathfrak{g}_x}\bra{\mathfrak{g}_x}_R\otimes P_S^{H_{\mathfrak{g}_x}}\,, 
\end{align}
where $P_S^{H_{\mathfrak{g}_x}}=U_S(\mathfrak{g}_x)P_S^HU_S^\dagger(\mathfrak{g}_x)$ is the projector onto the isotropy subspace corresponding to orientation state $\ket{\mathfrak{g}_x}_R$. One can show that $(\mathcal{T}_R^\dagger\mathcal{T}_R)^2=\mathcal{T}_R^\dagger\mathcal{T}_R$ and so this is a projector on $\Hkin$. In particular, it is a conditional projector, where the projected subspace of $\mathcal{H}_S$ rotates inside $\mathcal{H}_S$ depending on the frame orientation. So these maps fail even to define a change of TPS, since they are not unitary, let alone implementing the re-factorization of interest.

Nevertheless, the first equation in particular is suggestive of the required modification. It says that if we want to perform a disentangling re-factorization of the type described above, which shifts all the $G$-noninvariant data onto $\Tilde{R}$, we need to first project $\Hkin$ onto the subspace $\mathcal{H}_R\otimes P_S^H(\mathcal{H}_S)$. This leads to the natural definitions of modified trivialization maps for \emph{incomplete}, ideal frames
\begin{align}
    \Tilde{\mathcal{T}}_R:=\mathcal{T}_R(\mathbb{1}_R\otimes P_S^H)&=\int_X[dx]\l\ket{\mathfrak{g}_x}\bra{\mathfrak{g}_x}_R\otimes  \mathfrak{H}\l U_S(\mathfrak{g}_x)\r^\dagger \r(\mathbb{1}_R\otimes P_S^H)\\
    \Tilde{\mathcal{T}}_R^\dagger:=(\mathbb{1}_R\otimes P_S^H)\mathcal{T}_R^\dagger&=\int_X[dx]\l\ket{\mathfrak{g}_x}\bra{\mathfrak{g}_x}_R\otimes  \mathfrak{H}\l U_S(\mathfrak{g}_x)\r\r(\mathbb{1}_R\otimes P_S^H)
\end{align}
where we defined $\mathfrak{H}(\cdot)=\int_H [\dd{h}] U_S(h)(\cdot)U_S^\dag(h)$  to represent the $H$-twirl and used the facts that $P_S^H U_S(\mathfrak{g}_x)P_S^H=\mathfrak{H}\l U_S(\mathfrak{g}_x)\r P_S^H$ and $\mathfrak{H}\l U_S^\dagger(\mathfrak{g}_x)\r=\mathfrak{H}\l U_S(\mathfrak{g}_x)\r^\dagger$.
We see that $\Tilde{\mathcal{T}}_R$ and $\Tilde{\mathcal{T}}_R^\dagger$ perform isotropy averaged conditional unitaries on the system. These satisfy the properties
\begin{align}
    \Tilde{\mathcal{T}}_R\l U_R(g)\otimes U_S(g)\r\Tilde{\mathcal{T}}_R^\dagger&=U_R(g)\otimes P_S^H\\
    \Tilde{\mathcal{T}}_R^\dagger\Tilde{\mathcal{T}}_R&=\mathbb{1}_R\otimes P_S^H\\
    \Tilde{\mathcal{T}}_R\Tilde{\mathcal{T}}_R^\dagger&=\mathbb{1}_R\otimes P_S^H\,,
\end{align}
which tells us that they are unitary on the space $\mathcal{H}_R\otimes P_S^H(\mathcal{H}_S)$. Furthermore, on this subspace of $\Hkin$ the $G$-action gets transferred entirely onto $R$, as desired. All in all, we have a disentangling change of TPS on $\mathcal{H}_R\otimes P_S^H(\mathcal{H}_S)$
\begin{align}
    \Tilde{\mathcal{T}}_R:\mathcal{H}_R\otimes P_S^H(\mathcal{H}_S)\to \mathcal{H}_{\Tilde{R}}\otimes \mathcal{H}_{\Tilde{S}}\,,
\end{align}
with the action on states given by
\begin{align}
    \Tilde{\mathcal{T}}_R(\ket{\mathfrak{g}_x}_R\otimes P_S^H\ket{\psi}_S)=\ket{\mathfrak{g}_x}_R\otimes \mathfrak{H}\l U_S(\mathfrak{g}_x)\r^\dagger P_S^H\ket{\psi}_S=:\ket{\mathfrak{g}_x}_{\Tilde{R}}\otimes\ket{\Tilde{\psi}}_{\Tilde{S}}\,,
\end{align}
where $\ket{\Tilde{\psi}}_{\Tilde{S}}$ denotes the $G$-invariant state to which $P_S^H\ket{\psi}_S$ gets mapped under the non-local dressing to $R$. As alluded to earlier, for incomplete frames we get that $\mathcal{H}_{\Tilde{S}}\cong P_S^H(\mathcal{H}_S)$.

We note that when the frame has trivial isotropy group $H$, all the formulae reduce to the ones for complete frames.

The conclusion is that the incompleteness of the frame $R$ prevents a re-factorization into $G$-invariant and $G$-variant factors at the level of the full $\Hkin$, but only allows for one on the subspace $\mathcal{H}_R\otimes P_S^H(\mathcal{H}_S)$. The reason is that not all the degrees of freedom in $\mathcal{H}_S$ can be dressed to the incomplete frame $R$. The ones that cannot are precisely given by the orthogonal space $P_S^H(\mathcal{H}_S)^\perp$. We thus have 
\begin{align}
    \Hkin = \mathcal{H}_R\otimes \mathcal{H}_S = \mathcal{H}_R\otimes (\underbrace{P_S^H(\mathcal{H}_S)}_{\text{can dress to }R}\oplus \underbrace{P_S^H(\mathcal{H}_S)^\perp}_{\text{cannot dress to }R})\,.
\end{align}
\\~
We now bring all of these notions together to define the  \emph{subregional Goldstone mode} (GM).

In the reduced perspective of the extrinsic frame the physical Hilbert space includes the following factors associated to the subregion: $\H^{\rm phys}_{\rm sub}\otimes\H_{\tilde\Phi}$. This acts like a ``kinematical'' space for the frame reorientation group $\mathbb{G}_{\p\Sigma}$ of the extrinsic frame, as explained before. So the different quantities in the mechanical model above map naturally to the present situation
\begin{align}
    G \mapsto \mathbb{G}_{\p\Sigma}\,,\quad H\mapsto \mathbb{H}\,,\quad \mathcal{H}_R\mapsto \mathcal{H}_{\Tilde{\Phi}}\,,\quad \mathcal{H}_S\mapsto \H^{\rm phys}_{\rm sub}\,.
\end{align}

We define the natural extension of the trivialization map on $\mathcal{R}_\Phi(\Hphys)$ which re-factorizes the subregional degrees of freedom, while leaving the complement untouched
\begin{align}\label{eq:GM trivial}
    \Tilde{\mathcal{T}}_{\Tilde{\Phi}}:=\int_{\mathbb{G}_{\p\Sigma\setminus\mathcal{N}}}[\mathcal{D}\mathring{\mathfrak{g}}_{\p\Sigma}] \l\ket{\mathfrak{g}_{\p\Sigma}}\!\bra{\mathfrak{g}_{\p\Sigma}}_{\Tilde{\Phi}}\otimes \mathfrak{H}(\mathbb{U}_\text{sub}(\mathfrak{g}_{\p\Sigma}))^\dagger\otimes\mathbb{1}_\text{compl}\r(\mathbb{1}_{\Tilde{\Phi}}\otimes P_\text{sub}^\mathbb{H}\otimes \mathbb{1}_\text{compl})\,,
\end{align}
with the analogous definition for $\Tilde{\mathcal{T}}_{\Tilde{\Phi}}^\dagger$. The frame reorientations get mapped to
\begin{align}\label{eq:disentanglement}
    \Tilde{\mathcal{T}}_{\Tilde{\Phi}}\l\mathbb{U}_{\Tilde{\Phi}}(\textbf{g}_{\p\Sigma})\otimes \mathbb{U}_\text{sub}(\textbf{g}_{\p\Sigma}) \otimes\mathbb{U}_\text{compl}(\textbf{g}_{\p\Sigma})\r\Tilde{\mathcal{T}}_{\Tilde{\Phi}}^\dagger=\mathbb{U}_{\Tilde{\Phi}}(\textbf{g}_{\p\Sigma})\otimes \mathbb{1}_\text{sub}\otimes\mathbb{U}_\text{compl}(\textbf{g}_{\p\Sigma})\,.
\end{align}
As desired, under the map, the subregional degrees of freedom have been intrinsically dressed and are now frame reorientation invariant. The action of $\mathbb{G}_{\p\Sigma}$ only affects the frame and the complement. For this reason, we make the following identification with the spaces of the mechanical model above
\begin{align}
    \mathcal{H}_{\Tilde{R}}\mapsto \mathcal{H}_\text{GM}\cong \mathcal{H}_{\Tilde{\Phi}}\,,\quad \mathcal{H}_{\Tilde{S}}\mapsto \mathcal{H}_\text{sub}^\text{int}\cong P_\text{sub}^\mathbb{H}(\mathcal{H}_\text{sub}^\text{phys})\,.
\end{align}
and we have the follwing change of TPS on a subspace of $\mathcal{R}_\Phi(\Hphys)$
\begin{align}
    \Tilde{\mathcal{T}}_{\Tilde{\Phi}}:\mathcal{H}_{\Tilde{\Phi}}\otimes P_\text{sub}^\mathbb{H}(\mathcal{H}_\text{sub}^\text{phys})\otimes \mathcal{H}_\text{compl}^\text{phys}\to \mathcal{H}_\text{GM}\otimes\mathcal{H}_\text{sub}^\text{int}\otimes \mathcal{H}_\text{compl}^\text{phys}\,.
\end{align}
On this subspace of $\Hphys$, we can cleanly separate the subregional degrees of freedom into those that are frame reorientation invariant (and, hence, intrinsically dressed), and those that parameterize the orbit of frame reorientations, and so act as Goldstone mode (GM) degrees of freedom for the corner symmetries. Notice that, as a direct consequence of the incompleteness of the frame, the GM factor can only parameterize frame roerientations modulo $\p\Sigma$-uniform ones. This is the best one can achieve. We have that
\begin{align}
    \mathcal{H}_\text{GM}\cong L^2\l\mathbb{G}_{\p\Sigma}/\mathbb{H}\r\,.
\end{align} 

In the extrinsic perspective, the physical Hilbert space can thus be written as\footnote{We point out that, here, we have obtained a bulk/GM factorization between the inside of the region ($\mathcal{H}^\text{int}_\text{sub}$) and its corner ($\mathcal{H}_\text{GM}$) in a ``first quantize, then reduce'' approach, in contrast with the ``first reduce, then quantize'' continuum setup of \cite{Ball:2024gti,Ball:2024hqe}. }
\begin{align}\label{eq:GM}
    \Hphys\cong \l\l\mathcal{H}_\text{GM}\otimes\mathcal{H}_\text{sub}^\text{int}\r\oplus \l\mathcal{H}_{\Tilde{\Phi}}\otimes P_\text{sub}^\mathbb{H}(\mathcal{H}_\text{sub}^\text{phys})^\perp\r\r\otimes \mathcal{H}_\text{compl}^\text{phys} 
\end{align}

We can understand the various factors from our study of relational observables in App.~\ref{app:relational_obs}. There we saw that to dress any observable to the intrinsic frame $\Tilde{\Phi}$ a coarse-graining twirl over the isotropy group is implemented as a consequence of its incompleteness (note that $\mathbb{H}\cong H$ and $\mathbb{G}_{\p\Sigma}\cong \mathcal{G}_{\p\Sigma}$ so the computations follow in direct correspondence). If we focus on observables acting on the subregional factors, $\mathcal{A}_\text{sub}:=\mathcal{B}(\mathcal{H}_\text{sub}^\text{phys})$, we have an orthogonal decomposition (with respect to the Hilbert-Schmidt inner product) induced by the isotropy group 
\begin{align}
    \mathcal{A}_\text{sub} = \mathcal{A}_\text{sub}^\mathbb{H} \oplus (\mathcal{A}_\text{sub}^\mathbb{H})^\perp,
\end{align}
where $\mathcal{A}_\text{sub}^\mathbb{H}:=\mathfrak{H}(\mathcal{A}_\text{sub})$. So, if we attempt to dress the algebra to the intrinsic frame by performing the $\mathcal{G}_{\p\Sigma}$-twirl we get
\begin{align}
    \mathcal{A}_{\text{sub}|\Tilde{\Phi}}:=\mathcal{G}_{\p\Sigma}(\ket{\textbf{g}_{\p\Sigma}}\bra{\textbf{g}_{\p\Sigma}}_{\Tilde{\Phi}}\otimes \mathcal{A}_\text{sub})=\mathcal{G}_{\p\Sigma}(\ket{\textbf{g}_{\p\Sigma}}\bra{\textbf{g}_{\p\Sigma}}_{\Tilde{\Phi}}\otimes \mathcal{A}_\text{sub}^\mathbb{H})\,.
\end{align}
Only $\mathbb{H}$-invariant observables can be dressed intrinsically.
The representation of $\mathcal{A}_{\text{sub}|\Tilde{\Phi}}$ on the physical Hilbert space in the extrinsic perspective is then given by 
\begin{align}
    \mathcal{A}_{(\text{sub}|\Tilde{\Phi})|\Phi}^\text{phys}:=\mathcal{A}_{\text{sub}|\Tilde{\Phi}}P_\text{sub}^\mathbb{H}\,.
\end{align}
We refer the reader to App.~\ref{app:alg_content} for a more expanded description of this procedure. For now, we just want to highlight the conclusion that the intrinsically-dressed subregional algebra is thus naturally represented on the space $\mathcal{H}^\text{int}_\text{sub}$ that emerges out of the GM-refactorization above. 

This also clarifies what the factor $P_\text{sub}^\mathbb{H}(\mathcal{H}_\text{sub}^\text{phys})^\perp$ stands for. These are the subregional degrees of freedom that the intrinsic frame cannot ``see''. In other words, observables on this factor are trivial under an intrinsic relational dressing. We have seen several such examples already in App.~\ref{app:relational_obs}. For instance, any linear combination of subregional electric fields corresponds to a trivial intrinsic relational observable. Nevertheless, these operator \emph{can} be ``seen'' by the extrinsic frame. As we showed, these same operators lead to non-trivial relational observables under extrinsic dressing. We see that the Hilbert space factorization in \eqref{eq:GM} is a natural manifestation of the incompleteness of the intrinsic frame.

\paragraph{Goldstone Mode Observables:} The observables on $\mathcal{H}_\text{GM}$ can be mapped to $\Hphys$ by first applying the inverse of the unitary trivialization map of Eq. \eqref{eq:GM trivial} and then dressing with the extrinsic frame, following App.~\ref{app:relational_obs}.

The Wilson line operator $\hat{g}_{\Tilde{\Phi}_v}$ is invariant under the first operation since $\Tilde{\mathcal{T}}_{\Tilde{\Phi}}$ is diagonal in the same basis as $\hat{g}_{\Tilde{\Phi}_v}$. Now, we need to remember that the latter transforms at two nodes, $v$ and $\mathcal{N}$
\begin{align}
    \mathcal{U}_v(h_v)\mathcal{U}_\mathcal{N}(h_\mathcal{N}) \;\hat{g}_{\Tilde{\Phi}_v} \;\mathcal{U}_v^\dagger(h_v)\mathcal{U}_\mathcal{N}^\dagger(h_\mathcal{N}) &= \rho(h_v^{-1})\,\hat{g}_{\Tilde{\Phi}_v}\,\rho(h_\mathcal{N})\,.
\end{align}
Then, we have the relational observable field, describing the intrinsic edge mode field relative to the extrinsic one (conditional on it being in orientation $\mathbf{g}_{\p\Sigma}$), given by observables of the form
\begin{align}
   O_{\tilde\Phi_v|\Phi}(\mathbf{g}_{\p\Sigma}):=\mathcal{G}_{\p\Sigma}(\ket{\textbf{g}_{\p\Sigma}}\bra{\textbf{g}_{\p\Sigma}}_\Phi\otimes \hat{g}_{\Tilde{\Phi}_v})&=\int [dh_v][dh_\mathcal{N}] \ket{h_v}\bra{h_v}_{\Phi_v}\otimes \ket{h_\mathcal{N}}\bra{h_\mathcal{N}}_{\Phi_\mathcal{N}}\otimes \rho(g_vh_v^{-1})\,\hat{g}_{\Tilde{\Phi}_v}\,\rho(h_\mathcal{N}g_\mathcal{N}^{-1})\\
    &=\rho(g_v)(\hat{g}_{\Phi_v}^\dagger\otimes \hat{g}_{\Tilde{\Phi}_v}\otimes \hat{g}_{\Phi_\mathcal{N}})\rho(g_\mathcal{N}^{-1})\,,\label{eq:GMfield}
\end{align}
which represents an open Wilson line anchored with both ends on $\mathcal{B}$ passing through the north pole $\mathcal{N}$ of $\p\Sigma$ (in the main body also denoted as $v_0$). This is our lattice GM field, cf.~\cite{Araujo-Regado:2024dpr} for the classical continuum counterpart for Abelian theories.

Unlike for the link variable above, the conjugate variable gets modified under the change of TPS defining the GM factor. In the final TPS, the conjugate electric field is nothing but the generator of $\mathbb{U}_{\text{GM},v}(h_v)$ at each vertex. We see from \eqref{eq:disentanglement} that this gets mapped to
\begin{align}
    \mathbb{U}_{\text{GM},v}(h_v) \mapsto \mathbb{U}_{\Tilde{\Phi}_v}(h_v)\otimes \mathbb{U}_{\text{sub},v}(h_v)\otimes \mathbb{1}_{\text{compl},v}
\end{align}
in the original TPS. So, on $\mathcal{R}_\Phi(\Hphys)$ the generator is given by $\hat{\mathfrak{G}}^a_{\text{GM},v}:=\sum_{e\in\mathcal{E}_v\cap\Sigma}(\hat{\mathfrak{G}}_e^v)^a$, in other words the total Gau{\ss} generator with support on the subregion and corner. We now dress this with the extrinsic frame. Re-using the result from App.~\ref{app:relational_obs}
\begin{align}
    \mathcal{G}_{\p\Sigma}(\ket{\textbf{g}_{\p\Sigma}}\bra{\textbf{g}_{\p\Sigma}}_{\Phi}\otimes \hat{\mathfrak{G}}^a_{\text{GM},v})&=(D(g_v)\l\hat{g}_\text{Adj}\r_{\Phi_v}^\dagger)^{ab}\otimes \hat{\mathfrak{G}}^b_{\text{GM},v}\,. 
\end{align}
Finally, we note that on $\Hphys$, using the Gau{\ss} law at the vertex $v$ we can re-write the total Gau{\ss} generator in the subregion in terms of the total electric flux hitting $\p\Sigma$ from outside, including the extrinsic frame $\Phi$: $\sum_{e\in\mathcal{E}_v\cap\Sigma}(\hat{\mathfrak{G}}_e^v)^a= \sum_{e\in\mathcal{E}_v\cap(\Bar{\Sigma}\cup\Phi)}(\hat{\mathfrak{G}}_e^v)^a=: (\hat{\mathfrak{G}}^v_\perp)^a$. Thus, at the level of the physical algebra, the conjugate electric field to the GM at $v$ corresponds to the extrinsically dressed total outward electric flux at $v$, thus recovering the same result obtained classically in the continuum (see \cite{Araujo-Regado:2024dpr,Ball:2024gti})
\begin{align}
    \mathcal{G}_{\p\Sigma}(\ket{\textbf{g}_{\p\Sigma}}\bra{\textbf{g}_{\p\Sigma}}_{\Phi}\otimes \hat{\mathfrak{G}}^a_{\text{GM},v})\Pi_\text{phys} = (D(g_v)\l\hat{g}_\text{Adj}\r_{\Phi_v}^\dagger)^{ab}\otimes (\hat{\mathfrak{G}}^v_\perp)^b\,.
\end{align}
\\~
\indent We conclude by explaining how the structure simplifies in the Abelian case.  In that case, the projector $P^\mathbb{H}_\text{sub}= \mathbb{1}_\text{sub}$ because all states in $\mathcal{H}_\text{sub}^\text{phys}$ carry zero global electric charge, as a consequence of the Gau{\ss} constraint having been imposed at every node inside. This has the drastic effect of allowing for the GM re-factorization to work on the full $\Hphys$ for Abelian theories, cf.~\cite{Araujo-Regado:2024dpr,Ball:2024hqe} for the continuum counterpart. We expect this difference to manifest itself in the continuum also, even at the classical level, see \cite{Ball:2024gti,jesse}.

 \section{Extrinsic and intrinsic relational algebras}\label{app:alg_content}
In this appendix, we elaborate on the structure of the extrinsic and intrinsic relational algebras
$\mathcal{A}^{\rm phys}_{\Sigma|\Phi}$ and $\mathcal{A}^{\rm phys}_{\Sigma\setminus\tilde\Phi|\tilde\Phi}$, respectively, building upon the results of the previous appendices. As a preliminary remark, we recall that, by definition, both of these algebras consist solely of bounded operators. In particular, the extrinsic algebra is, via the reduction maps $\mathcal{R}_\Phi$, isomorphic to the algebra of bounded operators $\mathcal{A}^{\rm phys}_{\Sigma|\Phi} \simeq \cB (\H^\phys_\Sigma$), as described in the main text. While Wilson loop operators are themselves bounded, this is not generally the case for electric field operators, nor their quadratic combinations. Therefore, whenever we refer in the following to electric operators as part of the generators of the algebra, we implicitly mean bounded functions of these operators. This, in particular, encompasses the electric operators in finite groups, which only exist in ``exponentiated form''. A concrete example in the Abelian $U(1)$ case is provided by exponentiated electric field operators  $e^{i \theta \hat E}$, which, for $\theta \in [0,2\pi]$, spans the bounded functions of $\hat E$.

Let us also recall some QRF considerations before applying them to the lattice. Recall the definition of relational observables in Eq.~\eqref{eq:relobs}. Formally, we may write 
\begin{equation}\label{relalgkin}
    \mathcal{A}^{\rm kin}_{S|R}:=O_{\mathcal{A}_S|R}=\mathcal{G}\left(\ket{\mathbf{g}}\!\bra{\mathbf{g}}_R\otimes\mathcal{A}_S\right)\,,\qquad \mathcal{A}_S=\mathcal{B}(\H_S)\,,
\end{equation}
for the algebra of kinematical relational observables describing $S$ relative to $R$. Note that this algebra is independent of the QRF orientation $\mathbf{g}\in\mathcal{G}$, as may be checked with a change of integration variable in the $\mathcal{G}$-twirl and invoking the left and right invariance of the Haar measure. 

We are interested in the representation of this algebra on $\Hphys$, i.e.\ $\mathcal{A}^{\rm phys}_{S|R}:=\mathcal{A}^{\rm kin}_{S|R}\,\Pi_{\rm phys}$, where $\Pi_{\rm phys}$ is the \emph{coherent} group averaging projector. Using 
\begin{equation} \label{eq:Gpi}
\mathcal{G}(\bullet)\,\Pi_{\rm phys}=\hat\Pi_{\rm phys}(\bullet)\,,
\end{equation}
with $\hat\Pi_{\rm phys}(\bullet):=\Pi_{\rm phys}(\bullet)\Pi_{\rm phys}$ denoting conjugation by the projector, we have
\begin{equation}\label{eq:relphys}
    \mathcal{A}^{\rm phys}_{S|R}=\hat\Pi_{\rm phys}\left(\ket{\mathbf{g}}\!\bra{\mathbf{g}}_R\otimes\mathcal{A}_S\right)=\mathcal{R}^\dag_R(\mathbf{g})\,\mathcal{A}_S\,\mathcal{R}_R(\mathbf{g})\,,
\end{equation}
where the last equality holds for complete ideal QRFs $R$ and follows from \cite[Thm.~2]{delaHamette:2021oex} (cf.~Eq.~\eqref{eq:obsred}). Since $\mathcal{R}_R$ is unitary, this means that relational subsystem algebras commute $\big[\mathcal{A}_{S_\alpha|R}^{\rm phys},\mathcal{A}_{S_\beta|R}^{\rm phys}\big]=0$ for complete ideal QRFs, when $\mathcal{A}_S=\bigotimes_\alpha\mathcal{A}_{S_\alpha}$ admits a further subsystem partition \cite{Hoehn:2023ehz,delaHamette:2021oex,AliAhmad:2021adn}. 

Crucially, since $\mathcal{R}_R$ is unitary, we have from Eq.~\eqref{eq:relphys} that, for complete ideal QRFs $R$, the relational observable algebra $\mathcal{A}_{S|R}^{\rm phys}$ is a \emph{factor} if and only if $\mathcal{A}_S$ is one. Recall that a factor algebra is an algebra with trivial center. This means that the only subalgebra that commutes with \emph{all} elements of the algebra is the identity component, i.e.\ $\mathbb{C}$-numbers. In particular, when $\mathcal{A}_S=\mathcal{B}(\H_S)$, which is a von Neumann Type I factor, we have that $\mathcal{A}_{S|R}^{\rm phys}$ is a Type I factor also. Our extrinsic relational algebras will be Type I factors, and we will see how to obtain algebras with non-trivial centers from the intersections of algebras dressed with different frames. Note that, by contrast, $\mathcal{A}^{\rm kin}_{S|R}$ in Eq.~\eqref{relalgkin} is \emph{not} a factor; this is because it contains the subalgebra $\mathcal{G}(\mathcal{A}_S)$, which has the Casimir elements of the gauge generators in $S$ in its center. This will be explained in detail in Apps.~\ref{app:subrel} and~\ref{app_electriccenter}. Similarly, for incomplete frames, due to the average of the isotropy group, the relational observable algebra fails to be a factor when $G$ is non-Abelian.

Finally, we should also stress that the subsystem $S$ being dressed slightly differs between extrinsic and intrinsic algebra. For the former, $S$ coincides with the whole subregion $S =\Sigma = \mathring{\Sigma} \cup \p\Sigma$, while, for the intrinsic case, we are dressing only the boundary links excluded from the frame tree $S =\Sigma\setminus\tilde\Phi = \mathring{\Sigma} \cup \p\Sigma\setminus\tilde\Phi$.

We denote the representation of relational observables \eqref{eq:relobs} on $\Hphys$ by
\begin{equation}
    \mathcal{O}_{f_S|R}(\mathbf{g}):=O_{f_S|R}(\mathbf{g})\,\Pi_{\rm phys}=\hat\Pi_{\rm phys}\left(\ket{\mathbf{g}}\!\bra{\mathbf{g}}_R\otimes f_S\right)\,,
\end{equation}
and note that frame reorientations translate the relational observable orientation label \cite{delaHamette:2021oex} (likewise for $O_{f_S|R}(\mathbf{g})$):
\begin{equation}
    V_R(\mathbf{g}')\,\mathcal{O}_{f_S|R}(\mathbf{g})\,V_R^\dag(\mathbf{g}')=\mathcal{O}_{f_S|R}(\mathbf{g}\mathbf{g}'^{-1})\,.
\end{equation}

For simplicity, in the rest of this appendix, we will give explicit formulas for relational observables in the identity frame orientation $\mathcal{O}_{f_S|R}(\mathbf{e})$; all others can be obtained via reorientations.

\subsection*{Extrinsic relational algebras}

The extrinsic relational algebra $\mathcal{A}^{\rm phys}_{\Sigma|\Phi}$ consists of all relational observables, according to Eq.~\eqref{eq:relobs}, dressing any kinematical degrees of freedom in $\Sigma=\mathring{\Sigma}\cup\p\Sigma$ w.r.t.\ the chosen extrinsic edge mode frame $\Phi$, and projecting with $\Pi_{\rm phys}$ onto the physical Hilbert space. It is generated by:
\begin{itemize}
    \item Regional Wilson lines anchored to the global boundary via the extrinsic Wilson lines of $\Phi$. The frame indeed is completing the open lines with two endpoints on $\p \Sigma$, as explained in App. \ref{app:relational_obs}, Eq. \eqref{eq:Open_W_lines_ext_dress}. 
    \item Linear combinations of electric fields, dressed with a single line attached to the global boundary through the extrinsic frame Wilson lines, as explained in App.~\ref{app:relational_obs}, Eq.~\eqref{eq:L_ext_dress}. If the electric field lies inside the subregion, we need to parallel transport it to $\p\Sigma$ (with any path) before dressing with the extrinsic frame. These operators look like (cf.~Eq.~\eqref{eq:L_ext_dress} for frame in identity orientation) $ (\hat{\mathfrak{G}}^v_e)^a (\l\hat{g}_\text{Adj}\r_{\ell})^{ab}(\l\hat{g}_\text{Adj}\r_{\Phi_w})^{bc} 
    $, with $w\in \p\Sigma$,  $v\in \Sigma$ and $\ell$ any path connecting the two (from $w$ to $v$), with support on $\Sigma$. In terms of reduction map these are obtained as $\cO_{(\mathfrak{G}^v_e)^a\, \l\hat{g}_\text{Adj}\r_{\ell}^{ab}|\Phi}(\mathbf{e})$.

\end{itemize}

All of these generators are charged under frame reorientations: open Wilson lines transform in a representation $\rho$ of $G$, while electric field operators transform in the adjoint representation. Nonetheless, one can construct various reorientation-invariant operators (these have exclusively support in $\Sigma$) by contracting the open indices of the charged generators appropriately. This follows from the fact that $\mathcal{O}_{\cdot|\Phi}(\mathbf{e}):\mathcal{B}\l\H^{\rm phys}_\Sigma\r\to \mathcal{A}^{\rm phys}_{\Sigma|\Phi}$ is a unital $*$-homomorphism \cite[Thm.~1]{delaHamette:2021oex}.

For example, consider two open Wilson lines $\ell_1, \ell_2$ in the same representation $\rho$, both connecting points $v$ and $w$ on the corner $\p\Sigma$. The operator
\be
\Tr \cO_{\hat g_{\ell_1} | \Phi} (\mathbf{e}) \cO_{ \hat g_{\ell_2}^{-1} | \Phi} (\mathbf{e}) =  \Tr \l (\hat{g}_{\Tilde{\Phi}_v}^\dagger\otimes \hat g_{\ell_1} \otimes \hat{g}_{\Tilde{\Phi}_w} )(\hat{g}_{\Tilde{\Phi}_v}\otimes g_{\ell_2}^{-1}\otimes \hat{g}_{\Tilde{\Phi}_w}^\dagger) \r = \Tr \l \hat g_{\ell_1} g_{\ell_2}^{-1} \r
\ee
is invariant under frame reorientations, and corresponds to a closed Wilson loop along the path $\ell_1 \circ \ell_2^{-1}$. 

Similarly, we can get quadratic combinations of electric fields connected with an internal path. For example, given two boundary fields at nodes $v, w\in \p \Sigma$ and a path $\ell\in \Sigma$ between them, we have 
\be
\cO_{(\hat{\mathfrak{G}}^v_e )^a | \Phi} (\mathbf{e})
\cO_{({\hat{g}_{\rm Adj}}_\ell)^{ab} | \Phi} (\mathbf{e})
\cO_{(\hat{\mathfrak{G}}^w_{e'})^b | \Phi} (\mathbf{e})
= (\hat{\mathfrak{G}}^v_e )^a ({\hat{g}_{\rm Adj}})_\ell^{ab} (\hat{\mathfrak{G}}^w_{e'})^b\,.
\ee
Both of these examples act on the \textit{disentangled} subsystem introduced in the previous App.~\ref{app:GM}. Their invariance is not surprising as they belong to the internal gauge-invariant algebra $\mathcal{A}^\mathcal{G}_\Sigma$, which is unaffected by dressing. These examples show explicitly how operators on the \textit{disentangled} subsystem can be recovered as suitably contracted combinations of charged generators from the dressed algebra.

Among the open Wilson lines in the extrinsic algebra, there are ones connecting edges belonging to the intrinsic frame with the asymptotic boundary. These correspond to the relational observables measuring one frame relative to the other and are the Goldstone modes $\mathcal{O}_{\tilde\Phi_v|\Phi}$ in Eq.~\eqref{eq:GMfield} (projected to $\Hphys$).  

Finally, we emphasize that frame reorientation generators themselves can be obtained from linear combinations of dressed electric fields at the corner.
Using the Gau{\ss} law at a boundary node $v\in \p \Sigma$, we get, on $\Hphys$,
    \be
    \sum_{e \in \p\Sigma \cap \cE_v} (\hat{\mathfrak{G}}^v_e)^a = - \sum_{e \in \bar \Sigma \cap \cE_v} (\hat {\mathfrak{G}}^v_e)^a\,.
    \ee
    The right-hand side contains the electric fields normal to $\p \Sigma$, tangential to the frame  Wilson line, that act as a left vector field on $\H_{\Phi_v}$. In particular, if the boundary node has only one incoming edge from the complement, then $\sum_{e \in \bar \Sigma \cap \cE_v} (\mathfrak{G}^v_e)^a = \hat L_{\Phi_v}$. Now, in the extrinsic algebra, we have the parallel transported version of the operator in the left hand side of the equation above, all the way to the asymptotic boundary:
    \be
    \sum_{e \in \p\Sigma \cap \cE_v} \cO_{(\hat {\mathfrak{G}}^v_e)^a|\Phi}(\mathbf{e})=
    \sum_{e \in \Sigma \cap \cE_v} (\hat{\mathfrak{G}}^v_e)^a \l\hat{g}_\text{Adj}\r_{\Phi_v}^{ab} = - \sum_{e \in \bar \Sigma \cap \cE_v} (\hat{\mathfrak{G}}^v_e)^a \l\hat{g}_\text{Adj}\r_{\Phi_v}^{ab}\,.\label{eq:reorientgen}
    \ee
    The right-hand side contains $\hat L_{\Phi_v}^a \l\hat{g}_\text{Adj}\r_{\Phi_v}^{ab}$, that, by Eq. \eqref{eq:l=r}, is equal to the generator of right transformations on the frame, corresponding to frame reorientations, i.e.\ the electric corner symmetry group $\mathbb{G}_{\p\Sigma}$. This shows that the dressing of what looks like the generator of boundary gauge transformations on the boundary $\p \Sigma$ in the relational perspective, acts as a frame reorientation on the frame in the physical perspective.

In the Abelian case, all the dressings of $E$ are trivial and linear combinations of any electric field with support on the subregion (or normal to $\p \Sigma$, tangential to the frame) are part of the generators of the algebra.

\begin{minipage}{\textwidth}
    \centering
    \vspace{5pt}
    \captionsetup{hypcap=false} 
    \includegraphics[width=0.5\linewidth]{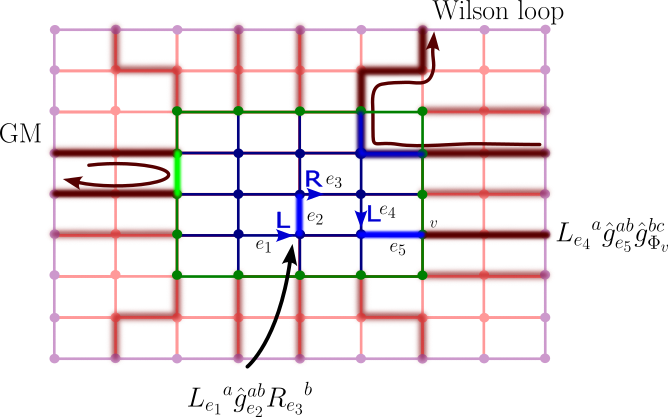}
    \captionof{figure}{\small Schematic illustration of some operators in the extrinsic algebra: Wilson loop attached to the global boundary (top right corner), with the particular example of one GM operator, dressing the intrinsic frame relative to the extrinsic one (left side); one example of a quadratic combination of $L_{e_1}$ and $R_{e_3}$ connected by the edge $e_2$ (in the bulk) and finally the electric field $L_{e_4}$, parallel transported to $\p \Sigma$ with $e_5$ and then dressed with $\Phi$.}  
\label{fig:ext frame dress}
\end{minipage}\\

\subsection*{Intrinsic relational algebra}

The intrinsic algebra has a similar overall structure, with the key difference that corner degrees of freedom must be treated separately for the frame tree $\tilde\Phi$ and its complement on $\p \Sigma$.  
As we saw in Eq.~\eqref{eq:intframepers}, \emph{after} imposing gauge constraints everywhere, \emph{except} on $\p\Sigma$, the complement of $\tilde\Phi$ among the surviving degrees of freedom is 
\begin{equation}
    \H_{\textbf{S}}=\H^{\rm phys}_{\mathring{\Sigma}}\otimes\H_{\p\Sigma\setminus\tilde\Phi}\otimes\H_\Phi\otimes\H_{\bar\Sigma\setminus\Phi}^{\rm phys}\,,
\end{equation}
where again the superscript `phys' indicates that gauge constraints are imposed everywhere, except on the corner. We have also seen that, after imposing constraints also on the corner, $\tilde\Phi$'s perspective is given by
\begin{equation}\label{eq:intfac2}
  \mathcal{R}_{\tilde\Phi}(\mathring{\mathbf{g}}_{\p\Sigma})\left(\Hphys\right)=\mathbf{P}^\text{phys}_\textbf{S}(\mathring{\textbf{g}}_{\p\Sigma})\left(\H_{\textbf{S}}\right)
\end{equation}
with $\mathbf{P}^\text{phys}_\textbf{S}(\mathring{\textbf{g}}_{\p\Sigma})$ the projector onto the subspace invariant under $\tilde\Phi$'s isotropy group $H$, which, in non-Abelian theories, acts non-locally on all four factors of Eq.~\eqref{eq:intfac2}. $\mathcal{R}_{\tilde\Phi}$ is unitary on this image.

Let now $S_\alpha$ be any of $\mathring{\Sigma},\p\Sigma\setminus\tilde\Phi,\Phi,\bar\Sigma\setminus\Phi$ and $\mathcal{A}_{S_\alpha}:=\mathcal{B}(\H^{\rm phys}_{S_\alpha})$ (with the `phys' dropped for $\p\Sigma\setminus\tilde\Phi,\Phi$). Invoking the generalization of Eq.~\eqref{eq:relphys} to incomplete QRFs \cite[Thm.~2]{delaHamette:2021oex} and adapting Eq.~\eqref{eq:Gpi}, we have that the algebra that $\tilde\Phi$ associates with $S_\alpha$ is given by
\begin{equation}\label{eq:isoalgebra}
   \mathcal{A}_{S_\alpha|\tilde\Phi}:= \hat{\mathbf{P}}^\text{phys}_\textbf{S}(\mathring{\textbf{g}}_{\p\Sigma})\left(\mathcal{A}_{S_\alpha}\right)=\mathcal{A}_{S_\alpha}^H\,\mathbf{P}^\text{phys}_\textbf{S}(\mathring{\textbf{g}}_{\p\Sigma})=\mathbf{P}^\text{phys}_\textbf{S}(\mathring{\textbf{g}}_{\p\Sigma})\,\mathcal{A}_{S_\alpha}^H\,,
\end{equation}
with $\mathcal{A}_{S_\alpha}^H=\mathfrak{H}(\mathcal{A}_{S_\alpha})\subset\mathcal{A}_{S_\alpha}$ the $H$-invariant subalgebra, produced by the incoherent group average over $\tilde\Phi$'s isotropy group,
\begin{equation}
    \mathfrak{H}(\bullet):=\frac{1}{\rm{Vol}(H)}\int_H\dd{h} \otimes_{v\in\p\Sigma}U_{\textbf{S},v}(h_{g_v})(\bullet)\otimes_{v\in\p\Sigma}U^\dag_{\textbf{S},v}(h_{g_v})\,,\qquad\qquad h_{g_v}:=g_vhg_v^{-1}\,,\quad g_{v_0}=e\,.
\end{equation}
That is, it is the isotropy-group-averaged algebra of $S_\alpha$; an incomplete QRF can only resolve what is invariant under its own isotropy group \cite{delaHamette:2021oex,Chataignier:2024eil}. The last equality in Eq.~\eqref{eq:isoalgebra} holds because $\mathbf{P}^\text{phys}_\textbf{S}(\mathring{\textbf{g}}_{\p\Sigma})$ is a coherent average over this isotropy group. In particular, since the incoherent twirl acts trivially on uninvolved tensor factors, we further have
\begin{equation}
    \Big[\mathfrak{H}(\mathcal{A}_{S_\alpha}),\mathfrak{H}(\mathcal{A}_{S_\beta})\Big]=0\,,\qquad \alpha\neq\beta\,,
\end{equation}
and so
\begin{equation}
   \Big[\mathcal{A}_{S_\alpha|\tilde\Phi},\mathcal{A}_{S_\beta|\tilde\Phi}\Big]=0\,,\qquad \alpha\neq\beta\,. 
\end{equation}
That is, the kinematical commutation, and hence distinction, between the $S_\alpha$-algebras survives in $\tilde\Phi$'s perspective, despite its incompleteness and the projector.
This works because our intrinsic frame $\tilde\Phi$ is still ideal, i.e.\ has perfectly distinguishable orientation states. For non-ideal QRFs, the projected algebra will in general no longer be an algebra;\footnote{The projector will then be a smeared, i.e.\ `fuzzy' group average \cite{delaHamette:2021oex,DeVuyst:2024uvd,Hoehn:2019fsy}.} e.g., projected elements from kinematically distinct algebras may fail to commute upon imposing gauge symmetries \cite{AliAhmad:2021adn}.

Crucially, however, the algebras $\mathcal{A}_{S_\alpha|\tilde\Phi}$ fail to be factors. This is because the incoherent twirl $\mathfrak{H}$ introduces a superselection across irreps of $H$ into the algebra. The center is then comprised of the Casimir elements of the irrep, as detailed below in App.~\ref{app_electriccenter}. In particular, as the $\mathfrak{H}$-twirl acts globally, we have
\begin{equation}
    \left(\mathcal{A}_{S_{\alpha}}\otimes\mathcal{A}_{S_\beta}\right)^H\supsetneq \mathcal{A}_{S_\alpha}^H\otimes\mathcal{A}_{S_\beta}^H\,,\qquad\quad \alpha\neq \beta\,.
\end{equation}
The right hand side misses some $H$-invariant observables with joint support on $S_\alpha,S_\beta$.

As we have seen in App.~\ref{app:factor}, the exception is given by Abelian theories, where $H$ acts non-trivially only on $S_\alpha=\Phi$, i.e.\ the extrinsic frame, and so also $\hat{\mathbf{P}}^\text{phys}_\textbf{S}(\mathring{\textbf{g}}_{\p\Sigma})$ and $\mathfrak{H}$ act non-trivially only on that factor. In this case, all of $\mathcal{A}_{S_\alpha|\tilde\Phi}$ remain Type I factors.\footnote{For $S_\alpha=\Phi$, it is $\mathcal{A}_{\Phi|\tilde\Phi}:=\mathcal{B}\left({\mathbf{P}}^\text{phys}_\textbf{S}(\mathring{\textbf{g}}_{\p\Sigma})\H_\Phi\right)$.} 

We may now use the unitarity of $\mathcal{R}_{\tilde\Phi}$ to lift these algebras back into the physical Hilbert space, and it follows from \cite[Thm.~2]{delaHamette:2021oex} that
\begin{equation}\label{eq:Aphys2}
    \mathcal{A}_{\rm phys}=\mathcal{R}^\dag_{\tilde\Phi}(\mathring{\mathbf{g}}_{\p\Sigma})\left(\bigotimes_{S_\alpha=\mathring{\Sigma},\p\Sigma\setminus\tilde\Phi,\Phi,\bar\Sigma\setminus\Phi}\mathcal{A}_{S_\alpha}\right)^H\mathcal{R}_{\tilde\Phi}(\mathring{\mathbf{g}}_{\p\Sigma})\,,
\end{equation}
is the total algebra of relational observables relative to $\tilde\Phi$, which, by construction, is a Type I factor, and
\begin{equation}\label{eq:intframealg}
  \mathcal{A}^{\rm phys}_{S_\alpha|\tilde\Phi}:=\mathcal{R}^\dag_{\tilde\Phi}(\mathring{\mathbf{g}}_{\p\Sigma})\left(\mathcal{A}_{S_\alpha}^H\right)\mathcal{R}_{\tilde\Phi}(\mathring{\mathbf{g}}_{\p\Sigma})\,,
\end{equation}
which (except in the Abelian case) is not a factor, due to the unitarity of $\mathcal{R}_{\tilde\Phi}$. For this reason, the tensor product on the r.h.s.\ of Eq.~\eqref{eq:Aphys2} cannot be pulled out of the conjugation. Note that the nonlocal coherent projector ${\mathbf{P}}^\text{phys}_\textbf{S}(\mathring{\textbf{g}}_{\p\Sigma})$ drops out in the conjugation by $\mathcal{R}_{\tilde\Phi}$, as it automatically maps onto its image, cf.~App.~\ref{app:factor}. For the regional observable algebra that $\tilde\Phi$ assigns to $\Sigma$, we have
\begin{equation}
    \mathcal{A}_{\Sigma\setminus\tilde\Phi|\tilde\Phi}^{\rm phys}:=\mathcal{R}^\dag_{\tilde\Phi}(\mathring{\mathbf{g}}_{\p\Sigma})\left(\mathcal{A}_{S_{\mathring{\Sigma}}}\otimes\mathcal{A}_{\p\Sigma\setminus\tilde\Phi}\right)^H\mathcal{R}_{\tilde\Phi}(\mathring{\mathbf{g}}_{\p\Sigma})\supsetneq\mathcal{R}^\dag_{\tilde\Phi}(\mathring{\mathbf{g}}_{\p\Sigma})\left(\mathcal{A}^H_{S_{\mathring{\Sigma}}}\otimes\mathcal{A}^H_{\p\Sigma\setminus\tilde\Phi}\right)\mathcal{R}_{\tilde\Phi}(\mathring{\mathbf{g}}_{\p\Sigma})\,.\label{eq:regionalinta}
\end{equation}
It is thus important to consider the two subsystems $\mathring\Sigma$ and $\p \Sigma \backslash\tilde\Phi$ together if our aim is to have an intrinsic relational description of the subregion $\Sigma$. This will have consequences on the subsystem relativity discussed in the following appendix, because the relativity is already present at the kinematical level and the frames always partially overlap, in contrast to the usual treatment as in \cite{AliAhmad:2021adn,Hoehn:2023ehz,delaHamette:2021oex,DeVuyst:2024pop,DeVuyst:2024uvd,Castro-Ruiz:2021vnq,AliAhmad:2024wja} and what we will discuss for the extrinsic case.

For a given $\tilde \Phi$, the algebra $\mathcal{A}^{\rm phys}_{\Sigma\setminus\tilde\Phi|\tilde\Phi}$ contains all the relational observables, according to Eq.~\eqref{eq:relobs}, dressing kinematical degrees of freedom in $\Sigma\setminus\tilde\Phi$ with $\tilde\Phi$. Thus it contains: 

\begin{itemize}
    \item Wilson loops entirely supported on $\Sigma = \mathring{\Sigma} \cup \p\Sigma$. If these have no support on the tree, meaning they are completely internal to $\Sigma$, or closed on the boundary by using only complementary links to the frame tree, they are already gauge invariant in $S=\Sigma\setminus\tilde\Phi$, on the other hand, they are obtained through the frame dressing in Eq.~\eqref{eq:intloop} if part of them is supported on the edges defining $\tilde\Phi$.
    \item Quadratic combinations of electric fields $(\hat {\mathfrak{G}}_e^v)^a$, acting on internal or boundary frame-complement edges, i.e. $v \in \Sigma$, $e \in \Sigma \backslash \tilde \Phi$. If acting at two different nodes, these are connected with any possible path supported on $\Sigma = \mathring{\Sigma} \cup \p\Sigma$, as in Eq.~\eqref{eq:c30}. If the path is supported only on $\Sigma \backslash \tilde \Phi$, they are already part of the subsystem algebra $\cA^\phys_S$. Otherwise, the segments lying in $\tilde \Phi$ appear, in the correct path-ordering fashion, when dressing and twirling on the boundary, cf.~Eq.~\eqref{eq:c28}.
\end{itemize}
Notably, no linear combinations of $L$ and $R$ electric fields acting on $\mathbf{S}$ are allowed in the non-Abelian case, but they are all present if $G$ is Abelian.

Finally, we can identify the central element with the dressed Casimir of the H-twirl generators. The undressed generators of constant gauge transformations on $\Sigma\setminus\tilde \Phi$ are given by the sum of Gau{\ss} generators on the corner $\p\Sigma$ restricted  to $\mathbf{S}$, meaning (for intrinsic frame in identity orientation) 
\be
U_{\textbf{S},v}(h) = e^{- i \theta^a \hat H^a }\,,\q\wth\,,\;\;\;   h = e^{i \theta^a T^a}\,,\;\;\; \hat H^a = \sum_{v \in \p \Sigma}\,\, \sum_{e  \in \mathcal{E}_v\cap \Sigma\setminus\tilde\Phi}  (\hat{\mathfrak{G}}_e^v)^a \,, 
\ee
representing the total electric field tangential and interior transversal to the corner in $\mathbf{S}$. Using the results in App.~\ref{app:relational_obs}, we immediately get that the dressed version of the corresponding Casimir is
\be\label{eq:intcenter}
 \mathcal{Z}_{\tilde{\Phi}}= \sum_{v,w \in \p \Sigma}\,\, \sum_{e  \in \mathcal{E}_v\cap \Sigma\setminus\tilde\Phi} \sum_{e'  \in \mathcal{E}_w\cap \Sigma\setminus\tilde\Phi} \l (\hat{\mathfrak{G}}_e^v)^a({\hat{g}_{\rm Adj}})_{\ell(\tilde \Phi)}^{ab}  (\hat{\mathfrak{G}}_{e'}^w)^b \r
\ee
where $\hat g_{\ell(\tilde \Phi)}$ is the unique Wilson line connecting the two nodes, entirely supported on the frame tree. Note that for Abelian theories, this operator is trivial. In that case, the dressing is trivial. Moreover, the contribution from the tangential electric fields vanishes because we add the $L$ and $R$ parts of those edges (for Abelian theories, these coincide) separately and with relative sign. We are thus left with the square of the sum of the interior normals to $\p \Sigma$, which vanishes by the Gau{\ss} law.

 \section{Subsystem relativity on the lattice for extrinsic edge mode QRFs}\label{app:subrel}

In this appendix, we adapt the property of \emph{subsystem relativity} \cite{AliAhmad:2021adn,Hoehn:2023ehz,delaHamette:2021oex,DeVuyst:2024pop,DeVuyst:2024uvd,Castro-Ruiz:2021vnq,AliAhmad:2024wja} to lattice gauge theories and refine some corresponding observations.

We say that two extrinsic frame choices $\Phi,\Phi'$ are non-overlapping if the edges of their Wilson lines do not overlap. Whenever the lattice $\mathcal{L}$ is sufficiently large, one can find such non-overlapping extrinsic QRFs. The following result is an adaptation and refinement of \cite[Thm.~1]{Hoehn:2023ehz} and results in \cite{DeVuyst:2024pop,DeVuyst:2024uvd} to the present non-Abelian lattice case.

\begin{prop}[\textbf{Extrinsic relational algebra intersections}]\label{prop:1}
Let $\Phi,\Phi'$ be two non-overlapping extrinsic QRFs for $\p\Sigma$. Furthermore, let $S_\alpha$ be any of the regional subsystems $\mathring{\Sigma},\p\Sigma\setminus\tilde\Phi,\tilde\Phi$, or any of their unions. Then,
\begin{equation}
\mathcal{A}^{\rm phys}_{S_\alpha|\Phi}\cap\mathcal{A}^{\rm phys}_{S_\alpha|\Phi'}=\hat\Pi^{\rm phys}_{\p\Sigma}\left(\mathcal{A}_{S_\alpha}\right)=\mathcal{A}^{\mathcal{G}_{\p\Sigma}}_{S_\alpha}\Pi^{\rm phys}_{\p\Sigma}=\mathbb{G}_{\p\Sigma}\left(\mathcal{A}^{\rm phys}_{S_\alpha|\Phi}\right)=\mathbb{G}'_{\p\Sigma}\left(\mathcal{A}^{\rm phys}_{S_\alpha|\Phi'}\right)\,,
\end{equation}
where $\mathbb{G}_{\p\Sigma}(\bullet),\mathbb{G}_{\p\Sigma}'(\bullet)$ denote the $G$-twirls over the reorientation groups $\mathbb{G}_{\p\Sigma},\mathbb{G}_{\p\Sigma}'$ of $\Phi,\Phi'$, respectively.
\end{prop}

\begin{proof}
 $\Phi,\Phi'$ non-overlapping means that their Hilbert spaces factorize, $\H_{\rm kin}\simeq\cdots\otimes\H_\Phi\otimes\H_{\Phi'}\otimes\cdots$, and that, while they may share anchor points on $\mathcal{B}$, the generators of their reorientations are independent. Thus, by construction, $\big[V_{\Phi'}(\mathbf{g}_{\p\Sigma}'),\mathcal{O}_{f_{S_\alpha}|\Phi}(\mathbf{g}_{\p\Sigma})\big]=0$, for all $\mathbf{g}_{\p\Sigma},\mathbf{g}_{\p\Sigma}'\in G^{V_{\p\Sigma}}$ because $\mathbb{G}_{\p\Sigma}\simeq\mathbb{G}'_{\p\Sigma}\simeq G^{V_{\p\Sigma}}$. That is, relational observables of $S_\alpha$ relative to $\Phi$ commute with reorientations of $\Phi'$, and vice versa, with $\Phi,\Phi'$ exchanged. Hence, operators in $\mathcal{A}^{\rm phys}_{S_\alpha|\Phi}\cap\mathcal{A}^{\rm phys}_{S_\alpha|\Phi'}$ must commute with reorientations of both extrinsic frames. In particular, for $\mathcal{O}_{f_{S_\alpha}|\Phi}(\mathbf{g}_{\p\Sigma})\in\mathcal{A}^{\rm phys}_{S_\alpha|\Phi}\cap\mathcal{A}^{\rm phys}_{S_\alpha|\Phi'}$, we have
 \begin{equation}
  \mathcal{O}_{f_{S_\alpha}|\Phi}(\mathbf{g}_{\p\Sigma})=V_{\Phi}(\mathbf{g}_{\p\Sigma}')\,\mathcal{O}_{f_{S_\alpha}|\Phi}(\mathbf{g}_{\p\Sigma})\,V_\Phi^\dag(\mathbf{g}'_{\p\Sigma})=\mathcal{O}_{f_{S_{\alpha}}|\Phi}(\mathbf{g}_{\p\Sigma}{\mathbf{g}'}_{\p\Sigma}^{-1})=\mathcal{O}_{U^\dag_{S_\alpha}(\tilde{\mathbf{g}}_{\p\Sigma})f_{S_\alpha}U_{S_\alpha}(\tilde{\mathbf{g}}_{\p\Sigma})|\Phi}(\mathbf{g}_{\p\Sigma})\,,
 \end{equation}
 with $\tilde{\mathbf{g}}_{\p\Sigma}=\mathbf{g}_{\p\Sigma}{\mathbf{g}'}_{\p\Sigma}^{-1}\mathbf{g}_{\p\Sigma}^{-1}$. Applying Eq.~\eqref{eq:obsred} to the present case implies $f_{S_\alpha}=U^\dag_{S_\alpha}(\tilde{\mathbf{g}}_{\p\Sigma})f_{S_\alpha}U_{S_\alpha}(\tilde{\mathbf{g}}_{\p\Sigma})$, and because this holds for all $\mathbf{g}_{\p\Sigma},\mathbf{g}'_{\p\Sigma}\in G^{V_{\p\Sigma}}$, it follows that $\big[f_{S_{\alpha}},U_{S_{\alpha}}(\mathbf{g}_{\p\Sigma})\big]=0$ for all $\mathbf{g}_{\p\Sigma}\in G^{V_{\p\Sigma}}$ and so $f_{S_{\alpha}}\in\mathcal{A}_{S_{\alpha}}^{\mathcal{G}_{\p\Sigma}}$. For such $f_{S_\alpha}$, we have $\mathcal{O}_{f_{S_\alpha}|\Phi}(\mathbf{g}_{\p\Sigma})=f_{S_\alpha}\Pi^{\rm phys}_{\p\Sigma}$, using that $\Phi$ is complete, so that $1/\rm{Vol}(\mathcal{G}_{\p\Sigma})\int_{\mathcal{G}_{\p\Sigma}}\dd\mathbf{g}\ket{\mathbf{g}}\!\bra{\mathbf{g}}_\Phi=\mathds1_{\Phi}$. Conversely, it is clear that $\mathcal{O}_{f_{S_\alpha}|\Phi}(\mathbf{g}_{\p\Sigma})\in\mathcal{A}^{\rm phys}_{S_\alpha|\Phi}\cap\mathcal{A}^{\rm phys}_{S_\alpha|\Phi'} $ for all $f_{S_\alpha}\in\mathcal{A}^{\mathcal{G}_{\p\Sigma}}_{S_\alpha}$ (and likewise for $\Phi'$). This establishes the first two equalities.

 To establish also the last two equalities, note that, for all $f_{S_\alpha}\in\mathcal{A}_{S_\alpha}$ and all $\mathbf{g}_{\p\Sigma}\in\mathcal{G}_{\p\Sigma}$,
 \begin{equation}
     \mathbb{G}_{\p\Sigma}\left(\mathcal{O}_{f_{S_\alpha}|\Phi}(\mathbf{g}_{\p\Sigma})\right)=\hat\Pi^{\rm phys}_{\p\Sigma}\left(\left(\frac{1}{\rm{Vol}(G^{V_{\p\Sigma}})}\int_{\mathbb{G}_{\p\Sigma}}\dd\mathbf{h}\ket{\mathbf{g}_{\p\Sigma}\mathbf{h}^{-1}}\!\bra{\mathbf{g}_{\p\Sigma}\mathbf{h}^{-1}}_\Phi\right)\otimes f_{S_{\alpha}}\right)=\hat\Pi^{\rm phys}_{\p\Sigma}\left(\mathds1_\Phi\otimes f_{S_\alpha}\right)\,,
 \end{equation}
 where we made use of the left and right invariance of the Haar measure for compact groups and of the completeness of $\Phi$ (and suppressed uninvolved tensor factors). The same reasoning applies to $\Phi'$.
\end{proof}

This is \emph{subsystem relativity} \cite{Hoehn:2023ehz,AliAhmad:2021adn,delaHamette:2021oex,Castro-Ruiz:2021vnq,DeVuyst:2024pop,DeVuyst:2024uvd} on the lattice: the two extrinsic frames $\Phi,\Phi'$ describe the subsystem $S_\alpha$ with intersecting, yet distinct subalgebras of the full algebra $\mathcal{A}_{\rm phys}=\mathcal{B}(\Hphys)$ on the physical Hilbert space. The two extrinsic QRFs only share $S_\alpha$ observables that are already gauge-invariant without dressing to either of $\Phi,\Phi'$ (e.g.\ Wilson loops in the bulk of $\Sigma$). But clearly, there are many relational observables dressing $S_\alpha$ to $\Phi$ (resp.\ $\Phi'$) that depend non-trivially on $\Phi$ (resp.\ $\Phi'$) and thus transform non-trivially under its reorientations and are not contained in the intersection (e.g.\ any Wilson line in $\Sigma$ with two endpoints on $\p\Sigma$ that upon dressing with $\Phi,\Phi'$ extends to a small gauge-invariant Wilson line with its two endpoints on $\mathcal{B}$). Hence, $\mathcal{A}^{\rm phys}_{S_\alpha|\Phi}\cap\mathcal{A}^{\rm phys}_{S_\alpha|\Phi'}=\mathcal{A}^{\mathcal{G}_{\p\Sigma}}_{S_\alpha}\,\Pi_{\p\Sigma}^{\rm phys}\subsetneq\mathcal{A}^{\rm phys}_{S_\alpha|\Phi}$ (and likewise for $\Phi'$). Indeed, the two relational algebras relative to $\Phi,\Phi'$ differ precisely in their cross-boundary gauge-invariant data, as illustrated in Fig.~\ref{fig:algebras&sub_rel}. That is, they each capture \emph{different relational information between $\Sigma$ and its complement $\bar{\Sigma}$} (cf.~\cite{Araujo-Regado:2024dpr} for an analogous classical discussion).

This has immediate consequences, establishing that the QRF-defined TPSs on $\Hphys$ indeed depend on the choice of QRF\,---\,with repercussion for relational entanglement entropies. To see this, denote by
\begin{equation}\label{eq:QRFtransf}
 V_{\Phi'\to\Phi}(\mathbf{g}_{\p\Sigma}',\mathbf{g}_{\p\Sigma}):=\mathcal{R}_\Phi(\mathbf{g}_{\p\Sigma})\,\mathcal{R}^\dag_{\Phi'}(\mathbf{g}'_{\p\Sigma})=\frac{1}{\rm{Vol}(G^{V_{\p\Sigma}})}\int_{\mathcal{G}_{\p\Sigma}}\dd\mathbf{h}\ket{\mathbf{h}\mathbf{g}'_{\p\Sigma}}_{\Phi'}\otimes\bra{\mathbf{h}^{-1}\mathbf{g}_{\p\Sigma}}_\Phi\otimes U_{\textbf{S}}(\mathbf{h}) 
\end{equation}
the QRF transformation from $\Phi'$ in orientation $\mathbf{g}'_{\p\Sigma}$ to QRF $\Phi$ in orientation $\mathbf{g}_{\p\Sigma}$, cf.~\cite[Thm.~4]{delaHamette:2021oex}.

\begin{corol}
 The intersection of the relational observable algebras relative to extrinsic frames $\Phi,\Phi'$ reduces as follows into $\Phi$'s perspective (likewise for $\Phi'$'s perspective):
 \begin{equation}
        \mathcal{R}_{\Phi}(\mathbf{g}_{\p\Sigma})\,\left(\mathcal{A}^{\rm phys}_{S_\alpha|\Phi}\cap\mathcal{A}^{\rm phys}_{S_\alpha|\Phi'}\right)\,\mathcal{R}_\Phi^\dag(\mathbf{g}_{\p\Sigma})=\mathcal{A}_{S_\alpha}^{\mathcal{G}_{\p\Sigma}}=\mathcal{G}_{\p\Sigma}(\mathcal{A}_{S_\alpha})\,.
    \end{equation}
In particular, under QRF transformations\footnote{We have implicitly deleted a $\mathds1_{\Phi}$ in $\mathcal{A}^{\rm kin}_{S_\alpha|\Phi'}$ compared to the kinematical representation.} 
\begin{equation}\label{eq:QRFtransformedAS}
    V_{\Phi'\to\Phi}(\mathbf{g}_{\p\Sigma}',\mathbf{g}_{\p\Sigma})\,\mathcal{A}_{S_\alpha}\,V^\dag_{\Phi'\to\Phi}(\mathbf{g}_{\p\Sigma}',\mathbf{g}_{\p\Sigma})=\mathcal{G}_{\p\Sigma}\left(\ket{\mathbf{g}'_{\p\Sigma}}\!\bra{\mathbf{g}'_{\p\Sigma}}_{\Phi'}\otimes\mathcal{A}_{S_\alpha}\right)=\mathcal{A}^{\rm kin}_{S_\alpha|\Phi'}\supsetneq\mathcal{G}_{\p\Sigma}(\mathcal{A}_{S_\alpha})\,.
\end{equation}
\end{corol}

\begin{proof}
The first part of the statement follows from Proposition~\ref{prop:1}, the observation that $\mathcal{A}_{S_{\alpha}}^{\mathcal{G}_{\p\Sigma}}$ commutes with the reduction maps, and $\bra{\mathbf{g}_{\p\Sigma}}\Pi^{\rm phys}_{\p\Sigma}\ket{\mathbf{g}_{\p\Sigma}}_\Phi=\mathds1_{\textbf{S}}$ \cite{delaHamette:2021oex} (cf.~\ref{eq:Rinv}). Eq.~\eqref{eq:QRFtransformedAS} follows from the definition in Eq.~\eqref{eq:QRFtransf} of the QRF transformation, adapting Eq.~\eqref{eq:relphys} to the present case, and noting that, again, $\mathcal{A}^{\rm kin}_{S_\alpha|\Phi'}$ commutes with the reduction maps and $\bra{\mathbf{g}_{\p\Sigma}}\Pi^{\rm phys}_{\p\Sigma}\ket{\mathbf{g}_{\p\Sigma}}_\Phi=\mathds1_{\textbf{S}}$.
\end{proof}

In other words, QRF-transforming what extrinsic frame $\Phi'$ ``sees'' as the system $S_\alpha$, namely the observable algebra with non-trivial support only on $\H_{S_{\alpha}}$, becomes in extrinsic frame $\Phi$'s perspective a composite system, built from $S_\alpha$ and the old frame $\Phi'$, manifested in the composite algebra $\mathcal{A}^{\rm kin}_{S_\alpha|\Phi'}$. The overlap of that algebra with what $\Phi$ ``sees'' as the system $S_\alpha$, namely $\mathcal{R}_{\Phi}(\mathbf{g}_{\p\Sigma})\left(\mathcal{A}^{\rm phys}_{S_\alpha|\Phi}\right)\mathcal{R}_\Phi^\dag(\mathbf{g}_{\p\Sigma})=\mathcal{A}_{S_\alpha}$, is exactly the already invariant $S_\alpha$ subalgebra $\mathcal{A}^{\mathcal{G}_{\p\Sigma}}_{S_\alpha}\subsetneq\mathcal{A}_{S_\alpha}$. 

Now recall from App.~\ref{app:factor} that both extrinsic frame perspectives $\mathcal{R}_\Phi,\mathcal{R}_{\Phi'}$ are gauge-invariant TPSs on $\Hphys$. Recall also from Eqs.~\eqref{eq:obsred} and~\eqref{eq:relphys} that the manifestly gauge-invariant relational observable algebras $\mathcal{A}^{\rm phys}_{S_\alpha|\Phi}$ are \emph{local} with respect to the TPS $\mathcal{R}_\Phi$ (likewise for $\Phi'$), in particular that
\begin{equation}
    \mathcal{R}_\Phi(\mathbf{g}_{\p\Sigma})\mathcal{A}_{\rm phys} \mathcal{R}^\dag_\Phi(\mathbf{g}_{\p\Sigma})=\bigotimes_{S_\alpha=\mathring{\Sigma},\p\Sigma\setminus\tilde\Phi,\tilde\Phi,\bar{\Sigma}\setminus\Phi}\mathcal{A}_{S_\alpha}=\mathcal{A}_{\textbf{S}}=\mathcal{B}(\H_{\textbf{S}})\,,
\end{equation}
where $\mathcal{R}_\Phi$ is unitary. Thus, the algebras commute for different $\alpha$ and together generate all of $\mathcal{A}_{\rm phys}=\bigvee_{S_\alpha=\mathring{\Sigma},\p\Sigma\setminus\tilde\Phi,\tilde\Phi,\bar\Sigma\setminus\Phi}\mathcal{A}^{\rm phys}_{S_\alpha|\Phi}$ (each $\mathcal{A}_{S_\alpha}=\mathcal{B}(\H_{S_\alpha})$ is a Type I factor). Hence, the relational algebras $\mathcal{A}^{\rm phys}_{S_\alpha|\Phi}$ \emph{generate} the TPS $\mathcal{R}_\Phi$ on $\Hphys$, and likewise for $\Phi'$ \cite{Hoehn:2023ehz,AliAhmad:2021adn} (see also \cite{Zanardi:2001zz,Zanardi:2004zz,Cotler:2017abq}). 

Putting the above together, we arrive at the conclusion:

\begin{corol}[\textbf{QRF transformations are changes of physical TPS}] 
The QRF transformation in Eq.~\eqref{eq:QRFtransf}, as a controlled non-local unitary, is nothing but \emph{a change of TPS on $\H_{\rm phys}$}, cf.~\cite[Sec.~III]{Hoehn:2023ehz}.
\end{corol}
What $\Phi'$ ``sees'' as the factor $\H_{S_\alpha}$ is not what $\Phi$ ``sees'' as $\H_{S_\alpha}$. Finally, as on a sufficiently big lattice $\mathcal{L}$ one can always find non-overlapping extrinsic frames $\Phi,\Phi'$, we have the following repercussion of Proposition~\ref{prop:1}:
\begin{corol}\label{cor_3}[\textbf{Center algebras from intersections of factors}]
    Consider the union of the regional subsystems $\Sigma:=\bigcup_{S_\alpha=\mathring{\Sigma},\p\Sigma\setminus\tilde\Phi,\tilde\Phi}S_\alpha$. This is what any extrinsic frame ``sees'' as the subsystem associated with region $\Sigma$. On any sufficiently large lattice $\mathcal{L}$, we have
    \begin{equation}
        \mathcal{A}_E:= \hat\Pi_{\p\Sigma}^{\rm phys}\left(\mathcal{A}_{\Sigma}\right)=\mathcal{A}^{\mathcal{G}_{\p\Sigma}}_{\Sigma}\Pi_{\p\Sigma}^{\rm phys}=\mathbb{G}_{\p\Sigma}\left(\mathcal{A}_{\Sigma|\Phi}^{\rm phys}\right)\,,
    \end{equation}
where the last equality holds for the reorientation $G$-twirl of any $\Phi$.
\end{corol}
Thus, the intersection of all possible relational algebras describing $\Sigma$ relative to all possible extrinsic QRF choices is the algebra of all regional degrees of freedom that is already gauge-invariant. In App.~\ref{app:global_decoherence}, we will see that this algebra contains a non-trivial center and, in fact, coincides with the \emph{electric center algebras} of \cite{Casini:2013rba,Soni:2015yga}. 

\section{Intersection of intrinsic frames and the magnetic center}
\label{app:intersection_int_frame}

So far, we have only considered subsystem relativity for extrinsic QRFs. However, it also applies to intrinsic QRFs $\tilde\Phi$. In that case, subsystem relativity for $\Sigma$ emerges already at the kinematical level: because the reference frame is built from degrees of freedom internal to the region $\Sigma$, each choice of frame $\tilde\Phi$ selects a different complement $\Sigma \setminus \tilde\Phi$ within $\Sigma$ to be treated as the subsystem. This corresponds to the refactorization of the Hilbert space, $\mathcal{H}^\text{kin}_\Sigma \simeq \mathcal{H}_{\tilde\Phi} \otimes \mathcal{H}_{\Sigma \setminus \tilde\Phi}$, which is frame-dependent.

Crucially, this subsystem relativity persists even after imposing the constraints, at the physical level. It is easy to verify that, for two different intrinsic frames $\tilde\Phi$ and $\tilde\Phi'$, the corresponding physical algebras assigned to the region $\Sigma$, namely $\cA^{\phys}_{\Sigma\setminus\tilde \Phi|\tilde \Phi}$ and $\mathcal{A}^{\rm phys}_{\Sigma\setminus\tilde \Phi'|\tilde \Phi'}$, do not completely overlap. For instance, the first algebra contains the Casimir operator associated with electric fields on links omitted from the tree defining $\tilde\Phi$, while the second contains those associated with a different set of links, omitted by $\tilde\Phi'$, and these do not fully overlap unless the two frames coincide.

Analogously to the extrinsic case, we can ask the question of what is the common intersection of all intrinsic algebras. From the previous App.~\ref{app:alg_content}, we immediately see that all the Wilson loops with support in $\Sigma$ are contained in any intrinsic algebra. On the other hand, the quadratic combinations of $(\hat{\mathfrak{G}}^v_e)^a$ are built with operators acting on frame-complement edges ($e\in \Sigma\setminus\tilde \Phi$), and the collection of them is frame dependent. However, the set of choices of path $\ell$ used to connect two vertices $v$ and $w$ in  $(\hat{\mathfrak{G}}^v_e )^a ({\hat{g}_{\rm Adj}})_\ell^{ab} (\hat{\mathfrak{G}}^w_{e'})^b$ is not frame-dependent, as the path $\ell$ can lie anywhere in $\Sigma$, either on the frame or its complement.\footnote{{If $\ell$ lies in the complement of $\tilde\Phi$, then it is part of the subalgebra of $\mathcal{B}(\H_\mathbf{S})$, which is already invariant pre-dressing with $\tilde\Phi$.}}

Combining these insights, we can characterize the generators of the intersection of all intrinsic algebras as: \begin{itemize}
    \item[(i)]  Wilson loops supported on $\Sigma$, including its boundary,\footnote{{Owing to the Mandelstam constraints \cite{Giles:1981ej,Gambini:1986ew,Loll:1991mh}, these are not all independent. We thank Bianca Dittrich for pointing this out.}} 
    \item[(ii)] Quadratic combinations of internal electric fields, dressed with any path in $\Sigma$, i.e.\ $(\hat{\mathfrak{G}}^v_e )^a ({\hat{g}_{\rm Adj}})_\ell^{ab} (\hat{\mathfrak{G}}^w_{e'})^b$, with $v,w\in \Sigma$, $e\in \mathring{\Sigma}$ and $\ell\in\Sigma$ connecting the two vertices. 
\end{itemize}
That is, all \emph{tangential} electric field observables drop out.

It is worth emphasizing that the generator of the center $\mathcal{Z}_{\tilde \Phi}$ in Eq.~\eqref{eq:intcenter} varies across different frames, both because of the choice of path and the selection of electric fields contributing to the flux. Nonetheless, the intersection algebra 
\begin{equation} 
\mathcal{A}_M:=\bigcap_{\text{{intrinsic QRFs }}\tilde\Phi}\mathcal{A}^{\rm phys}_{\Sigma\setminus\tilde\Phi|\tilde\Phi}
\end{equation}
\emph{also} has a center that contains both \emph{magnetic} and (only in the non-Abelian case) \emph{electric components}. This is our new proposal for a magnetic center algebra.

The \emph{magnetic part} of the center is generated by Wilson loops ($\Tr_\rho \hat{g}_W$) supported on $\partial\Sigma$. These commute with all observables in the algebra, since loops commute with each other, there are no electric field operators on boundary edges, and
\be
[\Tr_\rho \hat{g}_W, (\hat{\mathfrak{G}}^v_e)^a] =0 \q \text{if} \;\;e \notin W\,.
\ee

In the case in which the corner is one-dimensional, we will argue that the \textit{electric central element} is given by the total electric flux through the corner, extending the finite $G$ case in \cite{Delcamp:2016eya} to general compact Lie groups. This flux is projected onto the unique Wilson loop $W$ around the corner, yielding 
\begin{equation}
\label{eq:ZE_app}
Z_{1d}=\sum_{v\in\p\Sigma}\Tr_\rho\l \hat{g}_{W(v)} T_\rho^a\r\sum_{e\in\mathcal{E}_v\cap\mathring{\Sigma}}\l\mathfrak{G}^v_e\r^a=\sum_{v\in \p\Sigma}\Tr_\rho\l\hat{g}_{W(v_0)} \l\hat{g}_{\rm 
Adj}(v_0,v)\r^{ab}T^b_\rho\r\sum_{e\in\mathcal{E}_v\cap\mathring{\Sigma}}\l\mathfrak{G}^v_e\r^a\,,
\end{equation}
where $W(v)$ denotes the Wilson loop holonomy in $\p\Sigma$ from $v$ to $v$ before taking the trace and $\hat{g}_{\rm Adj}(v_0,v)$ denotes the parallel transporter in adjoint representation from $v$ to some reference node $v_0$ on $\p\Sigma$. In the second equality we made use of the identity 
\begin{equation}\label{eq:useful}
     \l\hat{g}_{\rm 
Adj}(v_0,v)\r^{ab}T^b_\rho=\hat{g}(v_0,v)T^a_\rho\hat{g}^\dag(v_0,v)\,.
\end{equation}

Let us show explicitly that it commutes with all elements in $\mathcal{A}_M$. First of all, it is clear that $Z_{1d}$ commutes with all regional Wilson loops that either lie entirely in the corner, or entirely in the bulk. In the former case, this happens because $Z_{1d}$ does not contain any tangential electric fields. In the latter case, the reason is that, at each $v$, $Z_{1d}$ sums over all the relevant parts in the Gau{\ss} law for the bulk Wilson loop (which may touch $v\in\p\Sigma$ without being tangential to $\p\Sigma$), which thus commutes with it. The same is true for all observables involving quadratic combinations of bulk electric fields connected by Wilson lines that are also purely bulk supported. 

The nontrivial commutation checks involve the following types of observables:
\begin{itemize}
\item[(i)] Wilson loops that are tangential to the boundary between vertices $v_1,v_2\in\p\Sigma$ and otherwise dive into the bulk;
\item [(ii)] Electric field observables $\l\mathfrak{G}^v_e\r^a\l\hat{g}_{\rm Adj}\r^{ab}_\ell \l\mathfrak{G}^w_{e'}\r^b$, such that $\ell$ is a Wilson line connecting $v,w\in\p\Sigma$ purely within $\p\Sigma$ and the electric fields are interior normal to the corner;
\item[(iii)] Electric field observables $\l\mathfrak{G}^v_e\r^a\l\hat{g}_{\rm Adj}\r^{ab}_\ell \l\mathfrak{G}^w_{e'}\r^b$, such that $\ell$ is a Wilson line with segments tangential to the corner, as well as segments in the bulk.
\end{itemize}

We begin with case (i). Consider a Wilson loop 
$W'$ that touches the corner between $v_1,v_2$ and is otherwise in the bulk. Using the commutators in Eq.~\eqref{eq:RL_G_commutators}, one finds
\be
[Z_{1d}, \Tr_{\rho} \hat{g}_{W'}] = \Tr_\rho (g_{W(v_1)} T^a)\Tr_\rho(\hat{g}_{W'(v_1)}T^a) - \Tr_\rho(\hat{g}_{W(v_2)}T^a) \Tr_\rho(\hat{g}_{W'(v_2)}T^a)\label{H5}
\ee
To analyze this expression, choose $v_0=v_1$ as the basepoint for the definition of $Z_{1d}$ as in the right-hand side of \eqref{eq:ZE_app}. Then the second term becomes
$$\Tr_\rho(g_{W(v_2)}T^a) \Tr_\rho(g_{W'(v_2)}T^a) = \Tr_\rho(g_{W(v_1)} T^b)\, \l\hat{g}_{\rm Adj}(v_1,v_2)\r^{ab} \l\hat{g}_{\rm Adj}(v_1,v_2)\r^{ac}\, \Tr_\rho(g_{W'(v_1)} T^c)\,, $$
But 
$\l\hat{g}_{\rm Adj}(v_1,v_2)\r^{ab} \l\hat{g}_{\rm Adj}(v_1,v_2)\r^{ac} = (\hat{g}_{\rm Adj}(v_1,v_2)^T)^{ba} \l\hat{g}_{\rm Adj}(v_1,v_2)\r^{ac} = \delta^{bc}$, 
where, in the last step, we used the orthogonality of the adjoint representation as inner product in the Lie algebra. Thus, Eq.~\eqref{H5} vanishes.

Next, let us consider (ii). We have
\begin{equation}
    \big[\l\mathfrak{G}^w_e\r^a\l\hat{g}_{\rm Adj}\r^{ab}_\ell \l\mathfrak{G}^{w'}_{e'}\r^b,Z_{1d}\big] = \sum_{v\in W}\Tr_\rho\l\hat{g}_{W(v)} T^c_\rho\r\l\hat{g}_{\rm Adj}\r^{ab}_\ell\sum_{\bar{e}\in\mathcal{E}_v\cap\mathring{\Sigma}}\big[ \l\mathfrak{G}^{w}_e\r^a \l\mathfrak{G}^{w'}_{e'}\r^b,\l\mathfrak{G}^v_{\bar{e}}\r^c\big]\,.
\end{equation}
Clearly, when $w=w'$, the observable commutes with $Z_{1d}$. In this case, $\l\mathfrak{G}^w_e\r^a\l\mathfrak{G}^{w}_{e'}\r^a$ is gauge-invariant and thus commutes with the relevant Gau{\ss} law contribution inside $Z_{1d}$. 
For a 1d boundary $w\neq w'$ lie both necessarily inside $W$. Then, using Eqs.~\eqref{commutators} and~\eqref{eq:Gauss_Law_generators},
\begin{align}
   \big[\l\mathfrak{G}^w_e\r^a\l\hat{g}_{\rm Adj}\r^{ab}_\ell \l\mathfrak{G}^{w'}_{e'}\r^b,Z_{1d}\big] &=-i\Tr_\rho\l\hat{g}_{W(w)}T^c_\rho\r \l\hat{g}_{\rm Adj}\r^{ab}_\ell f^{acd}\l\mathfrak{G}^w_e\r^d\l\mathfrak{G}^{w'}_{e'}\r^b\nonumber\\
   &\quad-i\Tr_\rho\l\hat{g}_{W(w')}T^c_\rho\r \l\hat{g}_{\rm Adj}\r^{ab}_\ell f^{bcd}\l\mathfrak{G}^w_e\r^a\l\mathfrak{G}^{w'}_{e'}\r^d\nonumber\\
   &=\Tr_\rho\l\hat{g}_{W(w)}\big[T_\rho^a,T^d_\rho\big]\r\l\hat{g}_{\rm Adj}\r^{ab}_\ell\l\mathfrak{G}^w_e\r^d\l\mathfrak{G}^{w'}_{e'}\r^b\nonumber\\
   &\quad +\Tr_\rho\l\hat{g}_{W(w')}\big[T^b_\rho,T^d_\rho\big]\r\l\hat{g}_{\rm Adj}\r^{ab}_\ell\l\mathfrak{G}^w_e\r^a\l\mathfrak{G}^{w'}_{e'}\r^d\nonumber\\
   &=\Tr_\rho\l\hat{g}_{W(w)}\Big[\l\mathfrak{G}^{w'}_{e'}\r^b\hat{g}^\dag_\ell T^b_\rho\hat{g}_\ell,\l\mathfrak{G}^{w}_{e}\r^d T^d_\rho\Big]\r\nonumber\\
   &\quad+\Tr_\rho\l\hat{g}_{W(w)}\Big[\l\mathfrak{G}^{w}_{e}\r^a T^a_\rho,\l\mathfrak{G}^{w'}_{e'}\r^d\hat{g}^\dag_\ell T^d_\rho\hat{g}_\ell\Big]\r=0\,,
\end{align}
where in the last equality we made use of Eq.~\eqref{eq:useful} and, invoking that $\hat{g}_\ell$ is the parallel transport between $w,w'$,
\begin{equation}
    \Tr_\rho\l\hat{g}_{W(w')}\big[\hat{g}_\ell T^a_\rho\hat{g}_\ell^\dag,T^d_\rho\big]\r = \Tr_\rho\l\hat{g}_{W(w)}\big[T^a_\rho,\hat{g}^\dag_\ell T^d_\rho\hat{g}_\ell\big]\r\,.
\end{equation}

Now, to study the third case (iii), we will simply need the following commutation relation, for two edges $e$, $e'$, both incoming to the boundary from the bulk at two nodes $v$, $w$ connected by a path along $\p \Sigma$:
\be
   \big[\l\mathfrak{G}^v_e\r^a\l\hat{g}_{\rm Adj}\r^{ab}_{\ell}\l\hat{g}_{\rm Adj}\r^{bc}_{e'} , Z_{1d}\big] =&-i\Tr_\rho\l\hat{g}_{W(v)}T^d_\rho\r  f^{ade}\l\mathfrak{G}^v_e\r^e \l\hat{g}_{\rm Adj}\r^{ab}_{\ell} \l\hat{g}_{\rm Adj}\r^{bc}_{e'}\nonumber\\
   &+i\Tr_\rho\l\hat{g}_{W(w)}T^d_\rho\r  \l\mathfrak{G}^v_e\r^a \l\hat{g}_{\rm Adj}\r^{ab}_{\ell} (T^d_{\rm Adj})^{be} \l\hat{g}_{\rm Adj}\r^{ec}_{e'}\nonumber\\
   =&\Tr_\rho\l\hat{g}_{W(v)} [T^a_\rho,T^e_\rho ]\r \l\mathfrak{G}^v_e\r^e \l\hat{g}_{\rm Adj}\r^{ab}_{\ell} \l\hat{g}_{\rm Adj}\r^{bc}_{e'}\nonumber\\
   &+\Tr_\rho\l\hat{g}_{W(w)} [T^b_\rho, T^e_\rho]\r  \l\mathfrak{G}^v_e\r^a \l\hat{g}_{\rm Adj}\r^{ab}_{\ell} \l\hat{g}_{\rm Adj}\r^{ec}_{e'}=0\nonumber\\
=&\Tr_\rho\l\hat{g}_{W(v)}\Big[\l\hat{g}_{\rm Adj}\r^{bc}_{e'}\hat{g}^\dag_\ell T^b_\rho\hat{g}_\ell,\l\mathfrak{G}^{v}_{e}\r^d T^d_\rho\Big]\r\nonumber\\
   &+\Tr_\rho\l\hat{g}_{W(v)}\Big[\l\mathfrak{G}^{w}_{e}\r^a T^a_\rho,\l\hat{g}_{\rm Adj}\r^{ec}_{e'}\hat{g}^\dag_\ell T^e_\rho\hat{g}_\ell\Big]\r=0\,,  
\ee
where we have used the fact that the matrix elements of the adjoint representation are given by the structure constants of the algebra $T_{\rm Adj}^{abc} =f^{abc}$, and the same properties of the parallel transporter as in the previous case. 

Now all possibilities in (iii), with at least one electric field touching $\p \Sigma$ and part of the Wilson line in the bulk are built out of the previous operators together with pieces that trivially commute with $Z_{1d}$. $\mathcal{A}_M$ is generated by the observables we investigated, hence $Z_{1d}$ commutes with all elements in $\cA_M$.

Notice that in the Abelian case, as for Eq.~\eqref{eq:intcenter}, we have that the total electric flux $\sum_{e\in\mathcal{E}_v\cap\mathring{\Sigma}}\l\mathfrak{G}^v_e\r^a$ in Eq.~\eqref{eq:ZE_app}  vanishes by the Gau{\ss} law. Thus, in the Abelian case, $\mathcal{A}_M$ only contains corner loops, reproducing the magnetic center definition of \cite{Casini:2013rba} (which was formulated for finite groups).

As we argue in Sec.~\ref{sec:alg} in the main text, we conjecture that $\mathcal{A}_M$ also contains electric parts for non-Abelian theories in higher dimensions. But we leave confirming this an open question.

\section{Superselection sectors from the incoherent G-twirl}\label{app_electriccenter}

Here we demonstrate how the action of the $G$-twirl operation on an algebra leads to a new algebra with superselection sectors labeled by irreps of $G$. In the case considered below, the irreps are labeled by a discrete set, as appropriate for electric irrep labels in the case of theories with finite or compact structure group. In our argument, it does not matter whether $G$ is gauge or physical. The results will directly apply to different cases, namely to the electric center $\mathcal{A}_E$ via Corollary~\ref{cor_3} and the intrinsic relational algebras for non-Abelian theories via Eq.~\eqref{eq:intcenter} and the $H$-twirl in Eq.~\eqref{eq:regionalinta}.

Consider a Hilbert space $\mathcal{H}$ which furnishes a unitary representation of some compact Lie group $G$. We can write
\begin{align}
    \mathcal{H}\cong\bigoplus_{[r]}\l\mathcal{H}_r\otimes \mathcal{H}_{rm}\r\,,
\end{align}
where $[r]$ is the set of isomorphism classes of irreps of $G$, $\mathcal{H}_r$ is the $r$-representation space and $\mathcal{H}_{rm}$ is the corresponding multiplicity space. We will label an orthonormal basis for $\mathcal{H}$ by $\{\ket{r,i,\alpha}\}$, where $\{i\}$ and $\{\alpha\}$ represent an orthonormal basis for $\mathcal{H}_r$ and $\mathcal{H}_{rm}$, respectively. The group action on $\mathcal{H}$ decomposes into irreps as
\begin{align}
    U(g) = \bigoplus_{[r]} U_r(g)\otimes \mathbb{1}_{rm}\,,
\end{align}
where $U_r$ is the group representation in irrep $r$. This irrep action is defined in terms of representation matrices
\begin{align}
    U_r(g) \ket{r,i}=\l\rho_r(g)\r^i_j\ket{r,j}\,.
\end{align}
We note that $\l\l\rho_r(g)\r^i_j\r^\dagger=\l\rho_r(g)\r^j_i{}^*$ in our conventions.

We now take some algebra of bounded operators on $\mathcal{H}$: $\mathcal{A}\subseteq\mathcal{B}(\mathcal{H})$. A generic operator written in the basis above takes the form
\begin{align}
    \hat{a}=\sum_{\substack{r,i,\alpha\\r',i',\alpha'}}\,a^{rr'}_{ii'\alpha\alpha'}\ket{r,i,\alpha}\bra{r',i',\alpha'}\,
\end{align}
for some complex coefficients $a^{rr'}_{ii'\alpha\alpha'}$. We now compute the G-twirl of $\hat{a}$ 
\begin{align}
    \mathfrak{G}(\hat{a})&:=\int [dg] \,U(g)\,\hat{a}\,U(g)^\dagger\\
    &=\int [dg] \sum_{\substack{r,i,\alpha\\r',i',\alpha'}}a^{rr'}_{ii'\alpha\alpha'}\,\l U_r(g)\ket{r,i}\bra{r',i'}U_{r'}(g)^\dagger\r\otimes \ket{r,\alpha}\bra{r',\alpha'}\\
    &=\sum_{\substack{r,i,\alpha\\r',i',\alpha'}}a^{rr'}_{ii'\alpha\alpha'}\, \sum_{j,j'}\l\int [dg] \l\rho_r(g)\r^i_j\l\rho_{r'}(g)\r^{j'}_{i'}{}^*\r \ket{r,j}\bra{r',j'}\otimes \ket{r,\alpha
    }\bra{r',\alpha'}\,.
\end{align}
We now use the irrep orthogonality relation
\begin{align}
    \int [dg] \l\rho_r(g)\r^i_j\l\rho_{r'}(g)\r^{j'}_{i'}{}^* = \frac{1}{d_r}\delta_{rr'}\delta^i_{i'}\delta^{j'}_j\,,
\end{align}
where $d_r$ is the irrep dimension. So, 
\begin{align}
    \mathfrak{G}(\hat{a})&=\sum_{r,i,\alpha,\beta} a^{rr}_{ii\alpha\beta} \frac{1}{d_r}\underbrace{\l\sum_j\ket{r,j}\bra{r,j}\r}_{\mathbb{1}_r}\otimes \ket{r,\alpha}\bra{r,\beta}\\
    &=\sum_r \frac{1}{d_r}\mathbb{1}_r\otimes \sum_{\alpha,\beta} \underbrace{\l\sum_i a^{rr}_{ii\alpha\beta}\r}_{=:\,(\mathfrak{a}_{rm})_{\alpha\beta}}\ket{r,\alpha}\bra{r,\beta}\\
    &=\sum_r\l\frac{1}{d_r}\mathbb{1}_r\otimes \mathfrak{a}_{rm}\r\,,
\end{align}
where we have defined the operator $\mathfrak{a}_{rm}\in \mathcal{B}(\mathcal{H}_{rm})$ acting on the multiplicity spaces alone. We note that we can think of the diagonal blocks (wrt to $r$) of the matrix coefficients of $\hat{a}$ as defining an operator $\hat{a}_r\in\mathcal{B}(\mathcal{H}_r\otimes \mathcal{H}_{rm})$ as $\hat{a}_r := \sum_{i,j,\alpha,\beta}\, a^{rr}_{ij\alpha\beta}\ket{r,i,\alpha}\bra{r,j,\beta}$. Then, we have that $\mathfrak{a}_{rm}=\Tr_{\mathcal{H}_r}(\hat{a}_r)$. As expected, we see that $\Tr_\mathcal{H}\mathfrak{G}(\hat{a})=\Tr_\mathcal{H}(\hat{a})=\sum_{r,i,\alpha}\,a^{rr}_{ii\alpha\alpha}$. The G-twirl is a CPTP map.

In conclusion, we have that
\begin{align}
    \mathfrak{G}(\mathcal{A})=\bigoplus_{[r]}\l\frac{1}{d_r}\mathbb{1}_r\otimes \Tr_{\mathcal{H}_r}\l\mathcal{A}\big|_r\r\r\,,
\end{align}
where by $\mathcal{A}\big|_r$ we mean the  restriction of $\mathcal{A}$ onto the diagonal block $r$.  We see that the superselection sectors are precisely the irrep classes of the group $G$. 

This naturally leads to the existence of an algebraic center, defined as $\mathcal{Z}:=\mathfrak{G}(\mathcal{A})\cap\mathfrak{G}(\mathcal{A})'$, where $'$ denotes the commutant in $\mathcal{B}(\mathcal{H})$. The bigger the starting algebra $\mathcal{A}$, the bigger $\mathfrak{G}(\mathcal{A})$. In particular, if we take $\mathcal{A}=\mathcal{B}(\mathcal{H})$, then we have that $\mathfrak{G}(\mathcal{A})=\bigoplus_{[r]}\l\frac{1}{d_r}\mathbb{1}_r\otimes \mathcal{B}(\mathcal{H}_{rm})\r$. Demanding commutativity in this ``maximal'' G-twirled algebra will always lead to the ``minimal'' center ($\mathcal{Z}_\text{min}$), which we now proceed to characterize. The condition that $[\mathfrak{z},\mathfrak{a}]=0\,,\forall \mathfrak{z}\in\mathcal{Z}\,,\, \mathfrak{a}\in\mathfrak{G}(\mathcal{A})$ leads to $[\mathfrak{z}_{rm},\mathfrak{a}_{rm}]=0\,,\forall r$. Because $\mathcal{B}(\mathcal{H}_{rm})$ is a factor, we have that $\mathfrak{z}_{rm}\propto\mathbb{1}_{rm}$. Thus, elements in $\mathcal{Z}_\text{min}$ are fully diagonal in the irrep basis and are labeled by a set of complex numbers: $\mathfrak{z}=\oplus_{[r]}\lambda_{[r]}\mathbb{1}_{[r]}$, $\lambda_{[r]}\in\mathbb{C}$. So, within each sector the center algebra is generated by the Casimirs of $\mathfrak{g}$ and these are contained in any center, regardless of the starting algebra $\cA$.

A subset of operators of interest to us are the states on the algebra: $\varrho\in \mathcal{S}(\mathcal{A})$. Under the G-twirl they become
\begin{align}
    \mathfrak{G}(\varrho)=\bigoplus_{[r]}\l\frac{1}{d_r}\mathbb{1}_r\otimes \varrho_{rm}\r=    \bigoplus_{[r]}p_r\l\frac{1}{d_r}\mathbb{1}_r\otimes \hat{\varrho}_{rm}\r\,
\end{align}
where we have defined $p_r:=\Tr_{\mathcal{H}_{rm}}(\varrho_{rm})$, satisfying $\sum_r p_r =\Tr_\mathcal{H}(\varrho)=1$. This means that $\hat{\varrho}_{rm}$ is such that $\Tr_{\mathcal{H}_{rm}}(\hat{\varrho}_{rm})=1$, and so it is a state on $\mathcal{H}_{rm}$. The von Neumann entropy of the G-twirled state is
\begin{align}\label{eq:entropyformula}
    S_\text{vN}(\mathfrak{G}(\varrho)) = -\sum_{[r]}p_r\log p_r + \sum_{[r]}p_r \l \log d_r + S_\text{vN}(\hat{\varrho}_{rm})\r\,.
\end{align}

A few comments are in order. First, applying the present discussion to the $\mathbb{G}_{\p\Sigma}$-twirl from any extrinsic algebra to $\mathcal{A}_E$ (cf.~Proposition~\ref{prop:1} and Corollary~\ref{cor_3}) yields indeed an algebra with center. The center is precisely given by all the Casimirs of the physical electric corner group $\mathbb{G}_{\p\Sigma}$; as shown in Eq.~\eqref{eq:reorientgen}, these are given by the $\Phi$-dressed electric fields normal to $\p\Sigma$. The center is thus comprised of the normal electric fields squared and this coincides with the electric center proposal in \cite{Soni:2015yga,Casini:2013rba}. The electric corner twirl
block diagonalizes operators, such as density matrices, in electric corner charge superselection sectors and essentially removes the GM. Information about the latter is now only contained in the probabilities $p_r$. See also App.~\ref{app:alternative} for further discussion of this in the Abelian case. 

Notably, our computation of the electric center entanglement entropy following from Eq.~\eqref{eq:entropyformula} and reported in Table~\ref{tab:1} \emph{differs} in the non-Abelian case from the one in \cite{Soni:2015yga,Delcamp:2016eya} (see also \cite{Bianchi:2024aim}), which misses the $\log d_r$-contribution (which vanishes in Abelian theories). In fact, our result in Eq.~\eqref{eq:entropyformula} agrees with the extended Hilbert space computation \cite{Ghosh:2015iwa,Soni:2015yga,Delcamp:2016eya,VanAcoleyen:2015ccp,Donnelly:2011hn,Donnelly:2014gva} instead. Our computation does not involve an auxiliary extended Hilbert space and is algebraic as the center calculation in \cite{Soni:2015yga,Delcamp:2016eya,Bianchi:2024aim}, leading however to a finer structure.

Finally, the entropy computation in Eq.~\eqref{eq:entropyformula} also directly applies to any intrinsic relational algebra, where the Casimirs are given by the $\tilde\Phi$-dressed square of the total electric flux tranversal to the corner and tangential to the complement of $\tilde\Phi$, see Eq.~\eqref{eq:intcenter}. In $\tilde\Phi$-perspective, i.e.\ upon applying the PW reduction $\mathcal{R}_{\tilde\Phi}$, this becomes directly the ``naked'' square of the sum of these electric fields.

\section{Pontryagin duality in Abelian theories}\label{app_pontryagin}

Let us briefly review basics of Pontryagin duality for Abelian gauge theories, which we will make use of subsequently.

The Pontryagin dual of a compact Abelian group $G$ is the group $\hat{G}=\rm{Hom}(G,\rm{U}(1))$ of homomorphisms from $G$ into the unit circle, i.e.\ the group of irreducible \emph{characters} $\chi$ of $G$. The group law is given by pointwise multiplication, $\chi\eta:g\to\chi(g)\eta(g)$ for all $\chi,\eta\in\hat{G}$ and $g\in G$, and we have the canonical pairing $G\times\hat{G}\to\rm{U}(1)$ given by $(g,\chi)\mapsto\chi(g)$. Pontryagin duality says that the double dual $\hat{\hat{G}}\simeq G$ is canonically isomorphic to $G$ itself. This is the basis for extending Fourier analysis to general Abelian groups. In particular, Plancherel's theorem implies that we have $L^2(G)\simeq L^2(\hat{G})$. For finite $G$, $\hat{G}$ is finite too with $|G|=|\hat{G}|$, and for $G$ continuous and compact, $\hat{G}$ is discrete with infinite cardinality. For instance, the Pontryagin dual of $G=\rm{U}(1)$ is $\mathbb{Z}$. In what follows, $\hat{G}$ is thus discrete.

Let us now build a unitary representation $\hat U$ for $\hat G$ on $\H=L^2(G)$ for later purpose by slightly generalizing the procedure in \cite[Prop.~4.7]{Carrozza:2024smc}, which holds for finite Abelian groups. Indeed, setting\footnote{We write $\chi$ as a superscript, given that for compact $G$ it takes value in a discrete set.}
\begin{equation}\label{eq:Uchidef}
    \hat{U}^\chi:=\frac{1}{\rm{Vol}(G)}\int_G\dd{g}\chi(g)\ket{g}\!\bra{g}
\end{equation}
defines a \emph{dual} unitary representation of $\hat G$ on $\H$, according to \cite[Def.~D.1]{Carrozza:2024smc}. This means that the $\{\hat U^\chi\}_{\chi\in\hat G}$ obey the following properties. First, using the orthonormality of the group basis, they comprise a group homomorphism:
\begin{equation}\label{eq:Uchi2}
    \hat{U}^\chi\hat{U}^\eta=\frac{1}{\rm{Vol}(G)}\int_G\dd{g}\chi(g)\eta(g)\ket{g}\!\bra{g}=\frac{1}{\rm{Vol}(G)}\int_G\dd{g}\chi\eta(g)\ket{g}\!\bra{g}=\hat{U}^{\chi\eta}\,.
\end{equation}
From this it also follows that $\hat{U}^\chi\left(\hat U^\chi\right)^\dag=\hat U^\chi\hat{U}^{\bar{\chi}}=\hat U^1=1=\left(\hat U\right)^\dag\hat U^\chi$, using the resolution of the identity. Hence, $\{\hat U^\chi\}_{\chi\in\hat G}$ is a unitary representation of $\hat G$ on $\H$. Second, they are \emph{dual} to the unitary $G$ representation $U$ on $\H$, i.e.\ for any $g\in G$ and $\chi\in\hat G$,
\begin{equation}
    U(g)\,\hat U^\chi=\frac{1}{\rm{Vol}(G)}\int_G\dd{g}'\chi(g')\ket{gg'}\!\bra{g'}=\frac{1}{\rm{Vol}(G)}\int_G\dd{g}'\chi(g^{-1})\chi(g')\ket{g'}\!\bra{g'} U(g)=\bar{\chi}(g)\,\hat U^\chi \,U(g)\,,\label{dualrep}
\end{equation}
where we used that $\chi(gg')=\chi(g)\chi(g')$ and $\chi(g^{-1})=\bar\chi(g)$. This relation mimics the Weyl duality relations for position and momentum translations (which are each other's Pontryagin dual).

We may now use the characters to build a dual orthonormal basis \cite{Carrozza:2024smc} by Fourier transform:
\begin{equation}
    \ket{\chi}:=\frac{1}{\rm{Vol}(G)}\int_G\dd{g}\chi(g)\ket{g}\,.
\end{equation}
Using the orthogonality relation
\begin{equation}
    \frac{1}{\rm{Vol}(G)}\int_G\dd{g}\bar{\chi}(g)\eta(g)=\delta_{\chi,\eta}
\end{equation}
it may be checked that indeed $\braket{\chi}{\eta}=\delta_{\chi,\eta}$ and
\begin{equation}\label{dualresid}
    \frac{1}{\rm{Vol}(G)}\sum_{\chi\in\hat{G}}\ket{\chi}\!\bra{\chi}=\mathds1\,.
\end{equation}Furthermore,
\begin{equation}\label{uchi}
    U(g)\ket{\chi}=\bar{\chi}(g)\ket{\chi}\,,\qquad\qquad \hat U^\eta\ket{\chi}=\ket{\eta\chi}\,,
\end{equation}
so the dual basis states are eigenstates of the gauge transformations with eigenvalues given by the (inverse) characters. (For Abelian theories, we drop the distinction between left and right multiplication.) Multiplying Eq.~\eqref{dualresid} with $U(g)$ and invoking the left equation in Eq.~\eqref{uchi} yields
\begin{equation}
    U(g)=\frac{1}{\rm{Vol}(G)}\sum_{\chi\in\hat{G}}\bar{\chi}(g)\ket{\chi}\!\bra{\chi}\,.
\end{equation}

Using the dual orthogonality relation
\begin{equation}\label{Eq:dualortho}
 \frac{1}{\rm{Vol}(G)} \sum_\chi\bar{\chi}(g)\chi(g')=\delta(g,g')\,,
\end{equation}
one finds the inverse Fourier transform
\begin{equation}
    \ket{g}=\frac{1}{\rm{Vol}(G)}\sum_{\chi\in\hat{G}}\bar{\chi}(g)\ket{\chi}\,
\end{equation}
and
\begin{equation}\label{eq:Ueta}
    \hat U^\eta\ket{g}=\eta(g)\ket{g}\,.
\end{equation}

Let us now show that Wilson loops give rise to a \emph{dual} unitary representation of the Pontryagin dual $\hat{G}$ according to \cite[Def.~D.1]{Carrozza:2024smc}. The starting point is the observation that Wilson loops define characters for the structure group $G$. Indeed, a closed Wilson line operator $\hat{g}_W$ from some vertex $v$ to itself takes value in the representation $\rho$ of $G$
\begin{equation}
    \hat{g}_W:=\frac{1}{\rm{Vol}(G)}\int_G\dd{g}\rho(g)\ket{g}\!\bra{g}_W\,,
\end{equation}
where we suppress the $\rho$-representation indices and $\{\ket{g}_W\}_{g\in G}$ constitute an orthonormal basis on the Hilbert space $\H_W\simeq L^2(G)$ on which the closed Wilson line is supported. Obtaining $\H_W$ from the local link Hilbert spaces involves a nonlocal refactorization of the kind explained in the main body. Thus, the corresponding Wilson loop operator reads
\begin{equation}\label{eq:loopcharacter}
   \Tr_\rho\hat{g}_W=\frac{1}{\rm{Vol}(G)}\int_G\dd{g}\chi_\rho(g)\ket{g}\!\bra{g}_W\,,
\end{equation}
where
\begin{equation}
\chi_\rho(g)=\Tr_\rho\rho(g)
\end{equation}
is the character of representation $\rho$.

For a general compact group $G$, the set of characters does not form a group. However, for the case of Abelian $G$, as considered here, it does form a group. In particular, for different irreps $\rho$, this constitutes $\hat{G}$.
Dropping for notational simplicity henceforth the $\rho$ subscript on the characters and labeling them with $\chi,\eta,...$ as before, we have that the Wilson loops
\begin{equation}
    \hat{U}^\chi:=\Tr_\rho\hat{g}_W=\frac{1}{\rm{Vol}(G)}\int_G\dd{g}\chi(g)\ket{g}\!\bra{g}_W
\end{equation}
define a \emph{dual} representation of $\hat{G}$ thanks to Eqs.~\eqref{eq:Uchidef}--\eqref{dualrep}.

These relations will prove useful below.

\section{An alternative derivation of the electric center entropy formula for Abelian theories}\label{app:alternative}

Consider a lattice gauge theory with compact or finite Abelian structure group $G$. We will briefly rederive the entropy formula Eq.~\eqref{eq:entropyformula} for the electric center algebra from our observation that the physical Hilbert space neatly factorizes relative to the extrinsic frame $\Phi$. This is to illustrate a key difference to the magnetic center entropy computation in the next section, which suffers from an infinite ambiguity when $G$ is continuous.

We saw in App.~\ref{app:factor} that for Abelian theories
\begin{equation}
\Hphys\simeq\H_{\mathring{\Sigma}|\tilde\Phi}^{\rm phys}\otimes\H^{\rm phys}_{\p\Sigma\setminus\tilde\Phi|\tilde\Phi}\otimes\H_{\rm GM}\otimes\H_{\bar\Sigma\setminus\Phi}^{\rm phys}\,,
\end{equation}
the regional factor $\H^{\rm phys}_{\Sigma}$ factorizes into the first three factors above. When computing entanglement entropies for a density operator supported on the first three factors, we may use the standard Hilbert space trace on $\H^{\rm phys}_\Sigma$, which thus factorizes across these three factors.

For a general regional density operator, we may write
\begin{equation}
    \rho^{\rm phys}_{\Sigma}=\sum_i\,p_i\rho^i_B\otimes\rho^i_{\rm GM}
\end{equation}
with $\sum_i p_i=1$ and the $\rho^i_B,\rho^i_{\rm GM}$ normalized, where $B=\Sigma\setminus\tilde\Phi|\tilde\Phi=\mathring{\Sigma}|\tilde\Phi\otimes\p\Sigma\setminus\tilde\Phi|\tilde\Phi$ for shorthand. Written in this factorization, the electric center algebra takes the form
\begin{equation}
    \mathcal{A}_E=\mathcal{B}(\mathcal{H}^{\rm phys}_B)\otimes\big\langle U_{\rm GM}(g)\,\big|\,g\in \mathbb{G}_{\p\Sigma}\simeq G^{V_{\p\Sigma}}\big\rangle
\end{equation}
because the $\{U_{\rm GM}(g)\}_{g\in \mathbb{G}_{\p\Sigma}}$, by Fourier analysis, generate all bounded functions of the exterior normal electric field operators on $\H_{\rm GM}\simeq L^2(G^{V_{\p\Sigma}})$. 

$\mathcal{A}_{\rm GM}^E:=\big\langle U_{\rm GM}(g)\,\big|\,g\in \mathbb{G}_{\p\Sigma}\big\rangle$, being Abelian, is a Type I von Neumann algebra. But because it is not a factor, it does not possess a unique trace \emph{a priori} (in each superselection sector, we could scale the trace differently). In the present case, this is not an issue, however. Indeed, note for contradistinction with the magnetic center case below that $\mathcal{A}_{\rm GM}^E$ still contains operators on $\H_{\rm GM}$ that are trace-class w.r.t.\ the standard Hilbert space trace. Thanks to Eq.~\eqref{uchi}, every element $a_{\rm GM}$ of $\mathcal{A}^E_{\rm GM}$ is diagonal in the dual character basis and therefore of the form
\begin{equation}
    a_{\rm GM}=\sum_{\chi\in\hat{G}^{V_{\p\Sigma}} }a(\chi)\ket{\chi}\!\bra{\chi}
\end{equation}
with bounded function $a(\chi)$. Trace-class operators are then those with $\sum_\chi|a(\chi)|<\infty$.
The standard trace on $\H_{\rm GM}$ is the trace on the larger algebra $\mathcal{B}(\H_{\rm GM})$ and is thus induced on $\mathcal{A}^E_{\rm GM}$. It is the natural choice in this case, breaking the large ambiguity in defining a trace on it.
We are therefore entitled to simply use the standard trace on the physical Hilbert space in what follows.

Now, if $\rho^{\rm phys}_\Sigma\in\mathcal{A}_E$, we must have that 
\begin{equation}
    \rho^i_{\rm GM}=\sum_{\mathbf{\chi}\in\hat G^{V_{\p\Sigma}}} f^i(\mathbf{\chi})\ket{\mathbf{\chi}}\!\bra{\mathbf{\chi}}
\end{equation}
with $0\leq f^i(\chi)\leq1$ normalized $\sum_\chi f^i(\chi)=1$, and so
\begin{equation}
    \rho^{\rm phys}_\Sigma=\sum_{\chi\in\hat{G}^{V_{\p\Sigma}}}\rho_B(\chi)\otimes\ket{\chi}\!\bra{\chi}\,,\qquad\qquad\rho_B(\chi):=\sum_i p_i f^i(\chi)\rho^i_B\,.
\end{equation}
Hence,
\begin{equation}
  \rho_\Sigma^{\rm phys}\log\rho^{\rm phys}_\Sigma=\sum_{\chi\in\hat{G}^{V_{\p\Sigma}}}\rho_B(\chi)\log\rho_B(\chi)\otimes\ket{\chi}\!\bra{\chi}
\end{equation}
and, using that the trace factorizes,
\begin{equation}
    -\Tr\rho^{\rm phys}_\Sigma\log\rho^{\rm phys}_\Sigma=-\sum_{\chi\in\hat{G}^{V_{\p\Sigma}}}\rho_B(\chi)\log\rho_B(\chi)\,.
\end{equation}
Note that
\begin{equation}
    \Tr_B\rho_B(\chi)=\sum_i p_i f^i(\chi)=:p_\chi
\end{equation}
is subnormalized, since $0\leq p_\chi\leq1$ and $\sum_{\chi\in\hat{G}^{V_{\p\Sigma}}}p_\chi=1$. Normalizing $\rho^\chi_B:=\rho_B(\chi)/p_\chi$ whenever $p_\chi>0$, we have
\begin{equation}\label{eq:electricentropy}
    S_{\rm vN}(\rho^{\rm phys}_\Sigma)=-\sum_{\chi\in\hat{G}^{V_{\p\Sigma}}}p_\chi\log p_\chi+\sum_{\chi\in\hat{G}^{V_{\p\Sigma}}} p_\chi S_{\rm vN}(\rho^\chi_B)\,,
\end{equation}
recovering Eq.~\eqref{eq:entropyformula} for the electric center in the Abelian case ($d_r=1$), upon identifying $\chi=r$ and the multiplicity spaces $\H_{rm}$ with the Hilbert space $\H_B^{\rm phys}$.

\section{Ambiguities in the entanglement entropy for magnetic center algebras: Abelian theories}\label{app_magneticambiguity}

Let us now repeat the exercise of the previous appendix for magnetic center algebra in the Abelian case (we discuss the non-Abelian case in App.~\ref{app_nonAbmagnetic} below). We will see that this is only well-defined for \emph{finite} structure group $G$. When $G$ is continuous, the magnetic center algebra does \emph{not} inherit the standard Hilbert space trace, nor that of any of the algebras of which $\mathcal{A}_M$ is a subalgebra by Eq.~\eqref{eq:alghierarchy}. While a trace can still be defined, there is a large ambiguity that cannot be broken by a natural induced trace, as in the electric center case.

The starting point is the perspective of some intrinsic QRF, from which the magnetic center can be obtained (recall that $\mathcal{A}_M\subsetneq\mathcal{A}^{\rm phys}_{\Sigma\setminus\tilde\Phi|\tilde\Phi}$, cf.~Eq.~\eqref{eq:alghierarchy}). From App.~\ref{app:factor}, we have that, for Abelian $G$, the physical Hilbert space relative to any intrinsic frame $\tilde\Phi$ also factorizes according to Eq.~\eqref{eq:Abelintfact}
\begin{equation}
\Hphys\simeq\H_{\mathring{\Sigma}}^{\rm phys}\otimes\H_{\p\Sigma\setminus\tilde\Phi}\otimes\textbf{P}^\text{inv}_\Phi(\mathbb{H})\l\mathcal{H}_\Phi\r\otimes\H^{\rm phys}_{\bar\Sigma\setminus\Phi}\,.
\end{equation}
We are interested in the first two factors on the r.h.s.\ on which $\mathcal{A}_M$ has support.
For a regional density operator in $\tilde\Phi$'s perspective we may write
\begin{equation}
    \rho^{\rm phys}_{\Sigma\setminus\tilde\Phi}=\sum_i p_i\,\rho^i_{\mathring{\Sigma}}\otimes\rho^i_{C}
\end{equation}
with $0\leq p_i$ and $\sum_i p_i=1$ and $\rho^i_{\mathring{\Sigma}},\rho^i_{C}$ normalized. For shorthand, we introduce only in this and the following appendix the notation $C=\p\Sigma\setminus\tilde\Phi$ for the complement of the intrinsic frame $\tilde\Phi$ in the corner $\p\Sigma$.

In this factorization, the magnetic center algebra reads (recall that in the Abelian case there is no electric contribution)
\begin{equation}
    \mathcal{A}_M=\mathcal{B}(\H^{\rm phys}_{\mathring{\Sigma}})\otimes\big\langle \hat U^\chi\,\big|\,\chi\in\hat G^{E_C}\big\rangle\,,
\end{equation}
where $E_C:=|\mathcal{E}\cap\p\Sigma|-(V_{\p\Sigma}-1)$ is the number of edges in the complement of the tree defining $\tilde\Phi$ in the corner $\p\Sigma$.
Similarly to the electric case, $\mathcal{A}^M_{C}:=\big\langle \hat U^\chi\,\big|\,\chi\in\hat G^{E_C}\big\rangle$ is an Abelian von Neumann algebra of Type I, generated by bounded functions of the corner Wilson loops $\Tr\hat g_W$ on $\H_{C}\simeq L^2(G^{E_C})\simeq L^2(\hat{G}^{E_C})$,\footnote{We thus count $E_C$ independent Wilson loops in $\p\Sigma$ generating $\mathcal{A}_C^M$ (taking different irreps into account). Note that when the dimension of $\p\Sigma$ is $D\geq2$, there will generally be additional Wilson loops in the corner, involving multiple (or only) edges in the complement of the tree. However, these will not be independent and the redundancy among Wilson loop variables is encoded in Mandelstam constraints, e.g.\ see \cite{Giles:1981ej,Gambini:1986ew,Loll:1991mh}. (We thank Bianca Dittrich for pointing this out.) Our construction thus automatically takes these into account.} and thus not a factor. Hence, it again possesses a highly non-unique trace. In contradistinction with the electric center case of the previous appendix, $\mathcal{A}_{C}^M$  now only contains operators that are trace-class w.r.t.\ the standard Hilbert space trace on $\H_{C}$, provided $G$ is finite. Only in that case may we use the standard Hilbert space trace to break the large ambiguity. 

To see this, note that Eq.~\eqref{eq:Ueta} implies that every element $a_{C}$ in $\mathcal{A}_{C}^M$ is diagonal in the original group basis (recall from App.~\ref{app_pontryagin} that the group elements $g\in G$ label the characters of the Pontryagin dual group $\hat G$) and therefore of the form
\begin{equation}\label{eq:magneticelement}
    a_C=\frac{1}{\rm{Vol}(G)^{E_C}}\int_{G^{E_C}}\dd{\mathbf g}a(\mathbf g)\ket{\mathbf g}\!\bra{\mathbf g}
\end{equation}
with bounded $a(\mathbf g)$. Hence,
\begin{equation}\label{eq:nomagnetictrace}
    \Tr_C a_C = \frac{1}{\rm{Vol}(G)^{E_C}}\int_{G^{E_C}}\dd{\mathbf g} a(\mathbf{g}) \,\delta(\mathbf{e},\mathbf{e})\,,
\end{equation}
so that when $G$ is continuous (and $\delta(\mathbf{e},\mathbf{e})$ diverges) no element in $\mathcal{A}^M_C$ is trace-class w.r.t.\ the standard Hilbert space trace. By contrast, when $G$ is finite (and $\delta(\mathbf{e},\mathbf{e})=1$ and the integral is replaced by a trace), we have instead that all $a_C$ are trace-class.\footnote{For a discrete structure group $G$ with infinite cardinality, one would likewise find trace-class operators, similarly to the electric center construction. However, we do not entertain such structure groups in this work.} Thus, in this latter case, we may invoke again the Hilbert space trace to proceed as in the electric center case.

For continuous $G$, we may still construct a trace on $\mathcal{A}_C^M$, given that it is of Type I. Besides linearity and cyclicity, a trace $\tau:\mathcal{A}_C^M\to\mathbb{C}$ must also obey ``faithfulness, normality and semifiniteness'' (see \cite{Sorce:2023fdx} for an introduction and \cite[Ch.~V]{Takesaki1} for details). It turns out that any of
\begin{equation}\label{eq:gentrace}
    \tau_\mu(a_C):=\frac{1}{\rm{Vol}(G)^{E_C}}\int_{G^{E_C}}\dd{\mathbf{g}}\,\mu(\mathbf{g})\,a(\mathbf{g})
\end{equation}
defines a valid trace on this algebra, provided $\mu(\mathbf{g})$ is a positive function on $G^{E_C}$ \cite[Ch.~V]{Takesaki1}. This ambiguity corresponds precisely to scaling the trace in each factor, i.e.\ in each different superselection sector labeled by $\mathbf{g}\in G^{E_C}$, in different ways. Crucially, none of these is induced by the Hilbert space trace.
Equipped with any one of these ``renormalized'' traces $\tau_\mu$, we may in principle proceed to compute entropies by using the trace
\begin{equation}\label{eq:tracefact}
    \Tr_M(O_{\mathring{\Sigma}}\otimes O_C):=\Tr_{\mathring{\Sigma}}\l O_{\mathring{\Sigma}}\r \tau_\mu(O_C)\,
\end{equation}
on the magnetic center $\mathcal{A}_M$. However, no choice of $\mu$ is better than any other. Note that we may also not inherit the trace on any of the other three algebras of which $\mathcal{A}_M$ is a subalgebra according to Eq.~\eqref{eq:alghierarchy} to break this ambiguity and fix a $\mu$. This is because the unique (up to scaling) trace of the immediate superalgebra, the intrinsic relational algebra $\mathcal{A}^{\rm phys}_{\Sigma\setminus\tilde\Phi|\tilde\Phi}\simeq\mathcal{B}(\H^{\rm phys}_{\mathring{\Sigma}})\otimes\mathcal{B}(\H_C)$, \emph{is} the standard Hilbert space trace, being a Type I \emph{factor} in the Abelian $G$ case.  Thus, there is infinite ambiguity in defining entanglement entropies associated with magnetic center algebras for continuous Abelian $G$. Without a physical means to break this ambiguity, inquiring about such entanglement entropies is not a physically well-defined question.

Let us nevertheless proceed with the (arbitrary) choice $\mu=1$, simply because the resulting formulae will be of the same form as for finite $G$, in which case the entanglement entropy associated with $\mathcal{A}_M$ \emph{is} well-defined. We will henceforth continue writing an integral over the group to account for this case and it is understood that $1/\rm{Vol}(G)^{E_C}\int_{G^{E_C}}\dd{\mathbf{g}}\to 1/|G^{E_C}|\sum_{\mathbf{g}\in G^{E_C}}$ and $\tau_1\to \Tr_C$ everywhere below when $G$ is finite.

On account of the above, and similarly to the previous appendix, when $\rho^{\rm phys}_{\Sigma\setminus\tilde\Phi}\in\mathcal{A}_M$, we must have
\begin{equation}
    \rho^i_C=\frac{1}{\rm{Vol}(G)^{E_C}}\int_{G^{E_C}}\dd{g}\rho^i(\mathbf{g})\ket{\mathbf{g}}\!\bra{\mathbf{g}}
\end{equation}
with $\rho^i(\mathbf{g})\geq0$ and $\tau_1(\rho^i_C)=1/\rm{Vol}(G)^{E_C}\int\dd{\mathbf{g}}\rho^i(\mathbf{g})=1$, and so
\begin{equation}
    \rho^{\rm phys}_{\Sigma\setminus\tilde\Phi}=\frac{1}{\rm{Vol}(G)^{E_C}}\int_{G^{E_C}}\dd{\mathbf{g}}\tilde{\rho}_{\mathring{\Sigma}}(\mathbf{g})\otimes\ket{\mathbf{g}}\!\bra{\mathbf{g}}\,,\qquad\qquad \tilde{\rho}_{\mathring{\Sigma}}(\mathbf{g}):=\sum_i p_i\rho^i(\mathbf{g})\rho^i_{\mathring{\Sigma}}\,.
\end{equation}
Hence,
\begin{equation}
    \rho^{\rm phys}_{\Sigma\setminus\tilde{\Phi}}\log\rho^{\rm phys}_{\Sigma\setminus\tilde\Phi}=\frac{1}{\rm{Vol}(G)^{E_C}}\int_{G^{E_C}}\dd{\mathbf{g}}\tilde\rho_{\mathring{\Sigma}}(\mathbf{g})\log\tilde\rho_{\mathring{\Sigma}}(\mathbf{g})\otimes\ket{\mathbf{g}}\!\bra{\mathbf{g}}\,,
\end{equation}
which, invoking Eq.~\eqref{eq:tracefact}, yields
\begin{equation}
    S_{\rm vN}(\rho^{\rm phys}_{\Sigma\setminus\tilde\Phi})=-\frac{1}{\rm{Vol}(G)^{E_C}}\int_{G^{E_C}}\dd{\mathbf{g}}\Tr_{\mathring{\Sigma}}\tilde\rho_{\mathring{\Sigma}}(\mathbf{g})\log\tilde\rho_{\mathring{\Sigma}}(\mathbf{g})\,.
\end{equation}
Again,
\begin{equation}
    0\leq p(\mathbf{g}):=\Tr_{\mathring{\Sigma}}\tilde\rho_{\mathring{\Sigma}}(\mathbf{g})=\sum_i p_i\rho^i(\mathbf{g})
\end{equation}
is not normalized. Renormalizing ${\rho}_{\mathring{\Sigma}}(\mathbf{g}):=\tilde\rho_{\mathring{\Sigma}}(\mathbf{g})/p(\mathbf{g})$ whenever $p(\mathbf{g})>0$, we arrive at 
\begin{equation}\label{eq:magneticentropy}
    S_{\rm vN}(\rho^{\rm phys}_{\Sigma\setminus\tilde\Phi})=-\frac{1}{\rm{Vol}(G)^{E_C}}\int_{G^{E_C}}\dd{\mathbf{g}}p(\mathbf{g})\log p(\mathbf{g}) +\frac{1}{\rm{Vol}(G)^{E_C}}\int_{G^{E_C}}\dd{\mathbf{g}} p(\mathbf{g})\,S_{\rm vN}(\rho_{\mathring{\Sigma}}(\mathbf{g}))\,,
\end{equation}
which is the magnetic center analog of Eq.~\eqref{eq:electricentropy}. In the case that $G$ is finite Abelian, this result is unambiguous and also recovers the earlier result of \cite{Casini:2013rba}. When $G$ is continuous, on the other hand, this result is highly ambiguous due to the reasons presented above; different choices of trace will change the measure by $\dd{\mathbf{g}}\to\dd{\mathbf{g}}\mu(\mathbf{g})$ in Eq.~\eqref{eq:magneticentropy}.

We will discuss the non-Abelian case in App.~\ref{app_nonAbmagnetic}.

\section{Magnetic center from magnetic dual twirl in the Abelian case}\label{app_magneticcenter}

In this appendix, we show that, for gauge theories with finite Abelian structure group $G$, the magnetic center algebra $\mathcal{A}_M$ can also be obtained as a $\hat{G}$-twirl of any intrinsic relational algebra, i.e.\ by incoherently averaging that algebra over the Pontryagin dual group $\hat G$. The unitary representation of $\hat G$ on $\Hphys$ will be constituted by the corner Wilson loops and so the magnetic corner degrees of freedom. This observation will feed into the proof in the next appendix that intrinsic relational and magnetic center algebras also obey an entanglement entropy inequality in the finite Abelian case.

We will also briefly explain why this observation fails for continuous Abelian groups $G$. This is related to the result of the previous appendix, showing that there is an infinite ambiguity in defining a trace and entanglement entropies for the magnetic center in that case. For non-Abelian theories, there is no natural dual to $G$, which is why we restrict to the Abelian case. The dual twirl for finite Abelian groups is therefore the maximum we can ask for.

Thus, the electric center is obtained by averaging extrinsic relational algebras over an \emph{electric corner group} $\mathbb{G}_{\p\Sigma}$, generated by electric charges on $\p\Sigma$, while the magnetic center in finite Abelian theories arises by averaging intrinsic relational algebras over a dual \emph{magnetic corner group} (of smaller cardinality), generated by magnetic charges (Wilson loops) on $\p\Sigma$. In the electric case this is tantamount to averaging over the orientations of a QRF for the electrically generated group, while for the magnetic case we will see that this is equivalent to averaging over the orientations of a dual QRF for the dual magnetically generated corner group. Note that Pontryagin duality underpins electromagnetic duality \cite{Ben-Zvi}.

We have seen in App.~\ref{app_pontryagin} that Wilson loops generate a dual unitary representation $\{\hat U^\chi\}_{\chi\in\hat{G}}$ of $\hat{G}$. Let us apply this to the corner loops in $\p\Sigma$ and any intrinsic relational algebra $\mathcal{A}^{\rm phys}_{\Sigma\setminus\tilde\Phi|\tilde\Phi}$. Recall that when $G$ is Abelian, as it will henceforth be, the intrinsic algebra is a type I factor,
\begin{equation}
    \mathcal{A}^{\rm phys}_{\Sigma\setminus\tilde\Phi|\tilde\Phi}\simeq\mathcal{B}(\H^{\rm phys}_{\mathring{\Sigma}})\otimes\mathcal{B}(\H_C)\,.
\end{equation}
 There are $E_C=|\mathcal{E}\cap\p\Sigma|-(V_{\p\Sigma}-1)$ edges in the complement of the tree defining an intrinsic edge mode QRF $\tilde\Phi$ in $\p\Sigma$, so that $\H_C\simeq L^2(G^{E_C})$. The corner loops live in this space, as does therefore the dual unitary representation of the magnetic corner group $\hat{G}_{\p\Sigma}\simeq\hat{G}^{E_C}$ generated by them.

 Suppose now that $G$ is finite. We may then define the dual magnetic twirl $\hat{G}_{\p\Sigma}(\bullet):\mathcal{A}^{\rm phys}_{\Sigma\setminus\tilde\Phi|\tilde\Phi}\to\mathcal{A}_M=\l\mathcal{A}^{\rm phys}_{\Sigma\setminus\tilde\Phi|\tilde\Phi}\r^{\hat{G}_{\p\Sigma}}$ as\footnote{Such a twirl over the Pontryagin dual appeared previously in the context of an error duality and a correspondence between quantum error correcting codes and QRF setups of the perspective-neutral framework used in this article \cite{Carrozza:2024smc}.}
\begin{equation}\label{eq:magnetictwirl}
    \hat{G}_{\p\Sigma}(\bullet):=\frac{1}{|G^{E_C}|}\sum_{\chi\in\hat{G}^{E_C}} \hat U^\chi\l\bullet\r\l\hat U^\chi\r^\dag\,.
\end{equation}
(Recall that $|G|=|\hat{G}|$.) Indeed, we have for the generating basis $\hat U^\eta U(\mathbf{g})$ of $\mathcal{B}(\H_C)$
\begin{equation}\label{eq:dualtwirl}
    \hat{G}_{\p\Sigma}\l\hat U^\eta U(\mathbf{g})\r=\hat U^\eta U(\mathbf{g})\frac{1}{|G^{E_C}|}\sum_{\chi\in\hat{G}^{E_C}}\chi(\mathbf{g})=\delta_{\mathbf{g},\mathbf{e}}\,\hat U^\eta U(\mathbf{g})\,,
\end{equation}
where in the first equality we made use of the duality \eqref{dualrep} of the representations $U,\hat U$, and in the second we invoked the orthogonality relation Eq.~\eqref{Eq:dualortho}. Thus, $\hat{G}_{\p\Sigma}(\mathcal{B}(\H_C))=\mathcal{A}_C^M$. Furthermore, operators in $\mathcal{B}(\H^{\rm phys}_{\mathring{\Sigma}})$ are clearly left invariant, so that 
\begin{equation}
    \hat{G}_{\p\Sigma}\l\mathcal{A}^{\rm phys}_{\Sigma\setminus\tilde\Phi|\tilde\Phi}\r = \mathcal{A}_M\,
\end{equation}
as desired. In particular, as a standard incoherent group average, $\hat{G}_{\p\Sigma}(\bullet)$ defines a CPTP map.

For continuous Abelian $G$, on the other hand, its Pontryagin dual $\hat{G}$ would have infinite cardinality and we could not normalize as in Eq.~\eqref{eq:magnetictwirl}. On the r.h.s.\ of Eq.~\eqref{eq:dualtwirl}, one would have to replace $\delta_{\mathbf{g},\mathbf{e}}\to\rm{Vol}(G^{E_C})\,\delta(\mathbf{g},\mathbf{e})$, which thus diverges for $\mathbf{g}=\mathbf{e}$. Somewhat informally, we would obtain
\begin{equation}
    \hat{G}_{\p\Sigma}\l\mathcal{A}^{\rm phys}_{\Sigma\setminus\tilde\Phi|\tilde\Phi}\r = \rm{Vol}(G^{E_C})\,\delta(\mathbf{e},\mathbf{e})\,\mathcal{A}_M\,.
\end{equation}
One may also check that regularizing the twirl by first defining it for a finite range and then taking limits, while overcoming the divergences, no longer averages the electric operators $U(\mathbf{g})$ to zero (or multiples of the identity). Similarly, one may define a somewhat modified regularized twirl by
\begin{equation}
    \hat{ G}'_{\p\Sigma}\l\bullet\r:=\sum_{\chi\in\hat{G}^{E_C}}\hat U^\chi(\bullet)\ket{1}\!\bra{1}\l\hat U^\chi\r^\dag\,,
\end{equation}
where $\ket{1}=\ket{\chi=1}$ is the trivial representation.
It may be checked, using the identities provided in App.~\ref{app_pontryagin}, that this twirl does satisfy $\hat{G}'_{\p\Sigma}\l\mathcal{A}^{\rm phys}_{\Sigma\setminus\tilde\Phi|\tilde\Phi}\r=\mathcal{A}_M$ for both finite and continuous Abelian $G$, overcoming the infinities. This map is also a projector on the algebra. However, it turns out not be trace-preserving, which makes it uninteresting for applications to states (as in the next appendix).
In any case, given that this case of continuous Abelian $G$ does not admit a physically well-defined definition of trace and entanglement entropies for the magnetic center, this failure will not further concern us.

\section{Ambiguities in the entropy for magnetic center algebras: non-Abelian theories}
\label{app_nonAbmagnetic}

Finally, let us also comment on entanglement entropies for magnetic center algebras $\mathcal{A}_M$ when $G$ is non-Abelian. Recall from the main body that our definition in Eq.~\eqref{eq:magneticcenterdef} is a generalization of previous proposals \cite{Casini:2013rba,Delcamp:2016eya}, both of which were formulated for finite groups only; the original proposal in \cite{Casini:2013rba} pertains to finite Abelian groups, while \cite{Delcamp:2016eya} also encompasses finite non-Abelian groups (and reproduces \cite{Casini:2013rba} in the Abelian case), yet is formulated only on $(2+1)D$ lattices using a fusion basis. Our $\mathcal{A}_M$ encompasses finite and compact Abelian and non-Abelian groups and is not restricted to a specific dimension. We have already computed the associated entanglement entropy in App.~\ref{app_magneticambiguity} for Abelian theories, reproducing the results of \cite{Casini:2013rba} for finite $G$ and noting the infinite ambiguity in defining it for continuous $G$. 

We will now argue that for finite non-Abelian case, the entanglement entropy continues to be well-defined, as in \cite{Delcamp:2016eya}, and that in the continuous non-Abelian case, infinite ambiguities persist. We will be somewhat schematic in this argument.

To this end, we note that, as usual for an algebra with non-trivial center, we can decompose $\mathcal{A}_M$ into charge sectors labeled by the eigenvalues of an independent basis of its center elements:
\begin{equation}\label{eq:decomp}
    \mathcal{A}_M=\int^\oplus_{\mathbf{G}}\dd{\mathbf{g}}\bigoplus_{\chi\in\hat{\mathbf{G}}}\tilde{\mathcal{A}}_{M}^\chi(\mathbf{g})\,.
\end{equation}
Here, $\mathbf{g}$ labels the eigenvalues of a choice of independent magnetic degrees of freedom in the center that come from an independent set of Wilson loops and take value in some index set $\mathbf{G}$. For corners $\p\Sigma$ of dimension $D\geq2$ this will require modding the set of corner Wilson loops by the redundancies encoded in the Mandelstam constraints \cite{Giles:1981ej,Gambini:1986ew,Loll:1991mh}; we assume this has been taken care of as in our Abelian discussion in App.~\ref{app_magneticambiguity}. Since Wilson loops define ``character operators'' on the structure group $G$, cf.~Eq.~\eqref{eq:loopcharacter}, the index set $\mathbf{G}$ will be a suitable set of copies of $G$ itself. Thus, the label $\mathbf{g}$ runs over a continuous/finite set for $G$ continuous/finite. For finite $G$, the direct integral $\int^\oplus$ can be replaced by a direct sum; only in that case will the corresponding eigenstates $\ket{\mathbf{g}}$ be proper normalized Hilbert space elements. 

The index $\chi$, on the other hand, labels the eigenvalues of the independent electric degrees of freedom in the center, taking value in some set $\hat{\mathbf{G}}$. When $G$ is finite, these eigenvalues remain discrete, in which case the corresponding eigenstates $\ket{\chi}$ will be proper Hilbert space elements. When $G$ is continuous, they may be continuous. For example, for one-dimensional corners, it appears that $Z_{1d}$ in Eq.~\eqref{eq:ZE_app} has continuous spectrum, due to the Wilson loop contribution. In this case too, the $\ket{\chi}$ would rather be distributions on Hilbert space. In what follows, it will not matter much whether $\hat{\mathbf{G}}$ is discrete or continuous; for simplicity we keep a discrete label set. The argumentation would not change for a continuous one.

Now, in non-Abelian theories, due to their non-trivial representation theory, there is generically a degeneracy in the eigenvalues of center elements. For example, in the electric center case in App.~\ref{app_electriccenter}, this is encoded in the dimension $d_r$ of the $r$-irrep spaces $\H_r$. The same will be true here: each label pair $(\mathbf{g},\chi)$ will occur a certain number of times $d_{\mathbf{g},\chi}$. Thus, we can also write the diagonal blocks in the decomposition Eq.~\eqref{eq:decomp} as 
\begin{equation}\label{eq:decomp2}
    \mathcal{A}_M=\int^\oplus_{\mathbf{G}}\dd{\mathbf{g}}\bigoplus_{\chi\in\hat{\mathbf{G}}}\mathds1_{\mathbf{g},\chi}\otimes{\mathcal{A}}_{M}^\chi(\mathbf{g})\,,
\end{equation}
where $\mathds1_{\mathbf{g},\chi}$ is a $(d_{\mathbf{g},\chi}\times d_{\mathbf{g},\chi})$ identity matrix. Equivalently, generalizing the discussion in App.~\ref{app_magneticambiguity}, we can write this as
\begin{equation}
    \mathcal{A}_M=\int_{\mathbf{G}}\dd{\mathbf{g}}\sum_{\chi\in\hat{\mathbf{G}}}\ket{\mathbf{g}}\!\bra{\mathbf{g}}_M\otimes\ket{\chi}\!\bra{\chi}_E\otimes\mathds1_{\mathbf{g},\chi}\otimes{\mathcal{A}}_{M}^\chi(\mathbf{g})\,,
\end{equation}
where $\ket{\mathbf{g}}_M,\ket{\chi}_E$ define a suitable, non-degenerate eigenbases on some Hilbert spaces $\H_M,\H_E$.

In order to compute entanglement entropies associated with such an algebra, we first need to define a trace $\Tr_M$ on it. Equipped with such a trace, we can then define regional density operators $\rho_M\in\mathcal{A}_M$ according to 
\begin{equation} 
\Tr_M\l\rho_M O_M\r=\Tr_{\rm phys}\l\rho_{\rm phys} O_M\r\,,\qquad\qquad \forall\,O_M\in\mathcal{A}_M\,,
\end{equation} 
where $\rho_{\rm phys}\in\mathcal{S}\l\Hphys\r$ is any global physical lattice state and $\Tr_{\rm phys}$ the standard physical Hilbert space trace. Since $\mathcal{A}_M$ has a nontrivial center, as in previous cases, there is \emph{a priori} no unique trace on it, as we can weigh each superselection sector $(\mathbf{g},\chi)$ in different ways.

It is clear that when $G$ is finite\,---\,and so $\mathbf{G}$ is a discrete index set\,---\,all operators in $\mathcal{A}_M$ are trace-class w.r.t.\ the standard Hilbert space trace $\Hphys$, as in the Abelian case of App.~\ref{app_magneticambiguity}. Thus, in this case, we can simply again use the standard Hilbert space to break this ambiguity. By contrast, when $G$ is continuous\,---\,and so is $\mathbf{G}$\,---\,then, just like in the Abelian case of App.~\ref{app_magneticambiguity}, there is \emph{no} element in $\mathcal{A}_M$ that is trace-class w.r.t.\ the standard Hilber space trace (for reasons similar to those discussed around Eqs.~\eqref{eq:magneticelement} and~\eqref{eq:nomagnetictrace}). The troublemaker is once more the improper eigenbasis $\ket{\mathbf{g}}$. As this pertains only to the center elements which commute with everything else, we may again use any of the trace constructions in Eq.~\eqref{eq:gentrace} to regularize the trace on $\H_M$ and use this to define $T_M$ (invoking factorization as in Eq.~\eqref{eq:tracefact}). However, just as in the Abelian case, no choice of trace on $\mathcal{A}_M$ is physically more natural than any other and none of them is induced by the standard physical Hilbert space trace. There is thus an infinite ambiguity in defining entanglement entropies for magnetic center algebras for continuous $G$, Abelian or non-Abelian.

For definiteness, let us assume we pick the $\tau_1$ trace for continuous $G$, and the standard Hilbert space trace for $G$ finite. Adapting then the arguments leading to Eq.~\eqref{eq:magneticentropy} to account for the additional degeneracy similarly to Eq.~\eqref{eq:entropyformula}, we arrive at the magnetic entanglement entropy
\begin{equation}
    S_{\rm vN}(\rho_M)=-\frac{1}{\rm{Vol}\mathbf{G}}\int_{\mathbf{G}}\sum_{\chi\in\hat{\mathbf{G}}}\dd{\mathbf{g}}p_\chi(\mathbf{g})\log p_\chi(\mathbf{g}) +\frac{1}{\rm{Vol}\mathbf{G}}\int_{\mathbf{G}}\dd{\mathbf{g}}\sum_{\chi\in\hat{\mathbf{G}}} p_\chi(\mathbf{g})\l\log d_{\mathbf{g},\chi}+S_{\rm vN}(\rho_M^\chi(\mathbf{g}))\r\,.
\end{equation}
Here, $\rho_M^\chi(\mathbf{g})\in\mathcal{A}_M^\chi$ and the probabilities $p_\chi(\mathbf{g})$ that superselection sectors $(\mathbf{g},\chi)$ are ``switched on'' are normalized, $\Tr\rho_M^\chi\rho(\mathbf{g})=1$ and $\int_\mathbf{G}\dd{\mathbf{g}}\sum_\chi p_\chi(\mathbf{g})=\rm{Vol}(\mathbf{G})$. Note that in the Abelian case, when $d_{\mathbf{g},\chi}=1$, this recovers Eq.~\eqref{eq:magneticentropy}.

For finite $G$, in which case $1/\rm{Vol}(\mathbf{G})\int_\mathbf{G}\dd{\mathbf{g}}\to 1/|\mathbf{G}|\sum_{\mathbf{g}\in\mathbf{G}}$, this result is unambiguous. For continuous $G$, due to the trace ambiguity, this result is physically rather arbitrary; choosing a different trace for $\H_M$ will change the $\mathbf{G}$ measure by some positive function $\mu(\mathbf{g})$ with no obvious physical means to break this ambiguity. We conclude that for continuous groups $G$, both Abelian and non-Abelian, entanglement entropies for magnetic center algebras are physically not well-defined. One could possibly sidestep this issue with a (quantum) deformation or (Bohr) compactification of the group. However, we will not consider such drastic modifications here. By contrast, for finite groups, the entropy is well-defined.

\section{Entropy inequalities from $G$-twirls}
\label{app:entropy_proof}

Here, we briefly establish the entanglement entropy inequalities between relational and center algebras discussed in the main body.

\begin{prop}[\textbf{Twirl-induced entropy inequality}]\label{prop:2}
Let $\H$ be a Hilbert space on which a finite or compact Lie group $G$ acts unitarily. We have $S_{\rm vN}(\rho)\leq S_{\rm vN}\left(\mathcal{G}(\rho)\right)$ for all $\rho\in\mathcal{S}(\mathcal{H})$, where $\mathcal{G}$ denotes the $G$-twirl. 
\end{prop}

\begin{proof}
The G-twirl is a CPTP map. Thus, the positivity of relative entropy implies
\begin{equation}
    0\leq S\left(\rho||\mathcal{G}(\rho)\right)=-S_{\rm vN}(\rho)-\Tr\left(\rho\log\mathcal{G}(\rho)\right)=S_{\rm vN}\left(\mathcal{G}(\rho)\right)-S_{\rm vN}(\rho)-\Tr\big[\left(\rho-\mathcal{G}(\rho)\right)\log\mathcal{G}(\rho)\big]\,.
\end{equation}
Now $\mathcal{G}$ is the orthogonal (with respect to the Hilbert-Schmidt inner product) projector onto the $G$-invariant subalgebra $\mathcal{A}^G$ of $\mathcal{A}=\mathcal{B}(\H)$. Thus, writing $\mathcal{A}=\mathcal{A}^G\oplus\left(\mathcal{A}^G\right)^\perp$, we have $\rho-\mathcal{G}(\rho)\in\left(\mathcal{A}^G\right)^\perp$ and $\log\mathcal{G}(\rho)\in\mathcal{A}^G$. Hence, the last term vanishes and we obtain the desired inequality. 
\end{proof}

The electric center $\mathcal{A}_E$ is obtained by twirling any of the extrinsic relational algebras $\mathcal{A}_{\Sigma|\Phi}^{\rm phys}$ over the electric corner group (cf.~App.~\ref{app:subrel}). This means also that the twirl is the correct channel mapping density operators from any relational into the center algebra. Indeed, this follows from the preservation of expectation values, noting that $\mathcal{A}_E$ is a subalgebra of any extrinsic relational algebra:
\begin{equation}
\forall\,a\in\mathcal{A}_E:\qquad \langle a\rangle= \Tr\l\rho_{\Sigma|\Phi}^{\rm phys}\,a\r=\Tr\l\rho_{\Sigma|\Phi}^{\rm phys}\,\mathbb{G}_{\p\Sigma}(a)\r=\Tr\l\mathbb{G}_{\p\Sigma}\l\rho_{\Sigma|\Phi}^{\rm phys}\r\,a\r\,,
\end{equation}
where $\rho^{\rm phys}_{\Sigma|\Phi}\in\mathcal{A}^{\rm phys}_{\Sigma|\Phi}$ and we used that the electric corner twirl $\mathbb{G}_{\p\Sigma}(\bullet)$ is an orthogonal projector with respect to the Hilbert-Schnmidt inner product.
Since the magnetic center $\mathcal{A}_M$, for finite Abelian groups $G$, is obtained by twirling any of the intrinsic relational algebras over the magnetic corner group (cf.~App.~\ref{app_magneticcenter}), the magnetic twirl is similarly the correct channel to map states from any intrinsic into the magnetic center algebra. 

Altogether, we arrive at the following conclusion:
\begin{corol}[\textbf{Relational entanglement entropy hierarchy}]
Let $\rho^{\rm phys}_{\Sigma|\Phi}\in\mathcal{A}_{\Sigma|\Phi}^{\rm phys}$ be an arbitrary state in an extrinsic relational algebra and $\mathbb{G}_{\p\Sigma}(\rho^{\rm phys}_{\Sigma|\Phi})\in\mathcal{A}_E$ its counterpart in the electric center algebra. Then
\begin{equation}\label{ineq:electric}
    S_{\rm vN}(\rho^{\rm phys}_{\Sigma|\Phi})\leq S_{\rm vN}\l\mathbb{G}_{\p\Sigma}\l\rho^{\rm phys}_{\Sigma|\Phi}\r\r\,,
\end{equation}
i.e.\ the entanglement entropy of \emph{any} extrinsic relational algebra is upper bounded by that of the electric center algebra. 

Similarly, suppose $G$ is finite Abelian. Let $\rho^{\rm phys}_{\Sigma\setminus\tilde\Phi|\tilde\Phi}\in\mathcal{A}_{\Sigma\setminus\tilde\Phi|\Phi}^{\rm phys}$ be an arbitrary state in an intrinsic relational algebra and $\hat{G}_{\p\Sigma}(\rho^{\rm phys}_{\Sigma\setminus\tilde\Phi|\tilde\Phi})\in\mathcal{A}_M$ its counterpart in the magnetic center algebra. Then
\begin{equation}
    S_{\rm vN}(\rho^{\rm phys}_{\Sigma\setminus\tilde\Phi|\tilde\Phi})\leq S_{\rm vN}\l\hat{G}_{\p\Sigma}\l\rho^{\rm phys}_{\Sigma\setminus\tilde\Phi|\tilde\Phi}\r\r\,.
\end{equation}
Thus, the entanglement entropy of \emph{any} intrinsic relational algebra is upper bounded by that of the magnetic center when $G$ is finite Abelian. 
\end{corol}

The $G$-twirl enhances our ignorance about the global state by eliminating degrees of freedom on the corner.

Finally, we point out that the content of inequality~\eqref{ineq:electric} also appears in an entropy inequality at the very end of \cite[SM, Sec.~VI]{VanAcoleyen:2015ccp}, but the relational content of it was not recognized. There, it bounds the entropy of ``gauge-fixed'' states by the regional center algebra entropy. As we discussed in the main body, any gauge fixing implicitly invokes a choice of global QRF. However, one could also use our extrinsic frame to gauge fix and express our inequality~\eqref{ineq:electric} in gauge-fixed language using $\mathcal{R}_\Phi$ (gauge fixing away from the corner is not important for this argument). Interestingly, the proof in \cite[SM, Sec.~VI]{VanAcoleyen:2015ccp} invokes quantum information theory arguments to arrive at the inequality. This complements our simple proof of Proposition~\ref{prop:2}, which essentially only invokes symmetry properties.

\section{Relationship between global and corner charges}
\label{app:global_decoherence}

In this appendix, we prove the statement, used in the main text, that whenever the global state commutes with the global charges, the reduced state in any (extrinsically dressed) subregion is necessarily reorientation invariant. To this end, let us work in the relational framework, where, for some complete extrinsic frame we have $\Hphys \simeq \H_{\Sigma|R }\otimes \H_{\Bar \Sigma|R}$. Large gauge transformations generated by the sum of incoming electric fields at each boundary node on $\cB$ act as (see discussion on frame reorientations in App.~\ref{app:GM})
\be
\mathbb{U}_{\Sigma|R}(\textbf{g}_{\cB})\otimes \mathbb{U}_{\bar \Sigma|R}(\textbf{g}_{\cB})\,.
\ee
In particular, we focus on the ones acting on the nodes where the extrinsic frame is anchored to. When the extrinsic frame $\Phi$ has no overlapping anchor points, there is a one-to-one correspondence between the node $w
$ in $\cB$ and the one ($v$) on $\p \Sigma$ where the Wilson line hits it. This means that the generator of the transformations above can be written in the form 
\be
Q_w^a = \Id_{\Sigma|R} \otimes \bar Q_v^a +\mathbb{Q}_v^a \otimes \Id_{\bar \Sigma|R}\,,
\ee
for some $\bar Q_v^a$ (generating some possibly-complicated unitary on the complement), and $\mathbb{Q}_v^a$, the subregion-supported corner charge at $v$, i.e.\ the generator of frame reorientations $\mathbb{G}_v$ for the subregional factors. It is then straightforward to see that, if for some global state $\varrho = \sum c_i \varrho^i_{\Sigma|R} \otimes \bar \varrho^i_{\Bar \Sigma|R} $ we have that
\be
[\varrho, Q_w^a] =0 \q\Rightarrow \q  0=&\Tr_{\bar \Sigma|R}[ \varrho,Q_w^a] \notag\\
0=&\Tr_{\bar \Sigma|R} \left [ \sum c_i \varrho^i_{\Sigma|R} \otimes \bar \varrho^i_{\Bar \Sigma|R} ,\Id_{\Sigma|R} \otimes \bar Q_v^a +\mathbb{Q}_v^a \otimes \Id_{\bar \Sigma|R}\right] \notag\\
0=& \sum c_i\left [  \varrho^i_{\Sigma|R} , \mathbb{Q}_v^a \right]  \Tr_{\bar \Sigma|R} \bar \varrho^i_{\Bar \Sigma|R} + 
\sum c_i \varrho^i_{\Sigma|R}  \cancel{ \Tr_{\bar \Sigma|R} \left [  \bar \varrho^i_{\Bar \Sigma|R} , \bar Q_v^a\right] } \notag\\
0=& \left [ \Tr_{\bar \Sigma|R}  \varrho, \mathbb{Q}_v^a \right] \,,
\ee 
meaning that the reduced state commutes with \emph{all} frame reorientations, irrespective of whether they are generated by large-gauge transformations or by some other mechanism in the complement. See \cite{Araujo-Regado:2024dpr} for examples of other more operationally-meaningful ways to realize frame reorientations, in the context of Maxwell theory.

\bibliographystyle{apsrev4-1}
\bibliography{EE}

\end{document}